\documentclass[a4paper,onecolumn,twoside,11pt]{article}
\usepackage{amssymb}
\usepackage{amsfonts}
\usepackage{amsmath}
\usepackage[nodisplayskipstretch,singlespacing]{setspace}
\usepackage[left=2.0cm, right=2.0cm, top=2.8cm, bottom=2.5cm,twoside]{geometry}
\usepackage{rotating}
\usepackage{float}
\usepackage{placeins}
\usepackage{hyperref}
\usepackage{appendix}
\usepackage[usenames,dvipsnames]{color}
\usepackage[comma,sort,longnamesfirst]{natbib}
\usepackage{fancyhdr}

\newtheorem{theorem}{Theorem}

\newtheorem{definition}{Definition}

\newtheorem{lemma}{Lemma}
\makeatletter
\@addtoreset{lemma}{theorem}
\makeatother

\newtheorem{proposition}{Proposition}
\newtheorem{remark}{Remark}

\newenvironment{proof}[1][Proof]{\noindent\textbf{#1.} }{\ \rule{0.5em}{0.5em}}

\begin{document}

\title{Optimal Bias Correction of the Log-periodogram Estimator of the
Fractional Parameter: A Jackknife Approach\thanks{%
This research has been supported by Australian Research Council Discovery
Grants No. DP150101728 and DP170100729. We would like to thank two anonymous
referees for very helpful comments on an earlier version of the paper. We
would also like to acknowledge the comments of participants at: the Royal
Statistical Society Conference, University of Glasgow, August, 2017; and the
11th International Conference on Computational and Financial Econometrics,
University of London, December, 2017. We thank the Monash e-research centre
for the use of their computing facilities.}}
\author{K. Nadarajah\thanks{%
Department of Economics, University of Sheffield, United Kingdom.
Corresponding author:\textbf{\ }k.nadarajah@sheffield.ac.uk.}, Gael M. Martin%
\thanks{%
Department of Econometrics and Business Statistics, Monash University,
Australia.} \ \& D. S. Poskitt\thanks{%
Department of Econometrics and Business Statistics, Monash University,
Australia.}}
\maketitle

\begin{abstract}
{\small We use the jackknife to bias correct the log-periodogram regression
(LPR) estimator of the fractional parameter\ in a stationary fractionally
integrated model. The weights for the jackknife estimator are chosen in such
a way that bias reduction is achieved without the usual increase in
asymptotic variance, with the estimator viewed as `optimal' in this sense.
The theoretical results are valid under both the non-overlapping and
moving-block sub-sampling schemes that can be used in the jackknife
technique, and do not require the assumption of Gaussianity for the data
generating process. A Monte Carlo study explores the finite sample
performance of different versions of the jackknife estimator, under a
variety of scenarios. The simulation experiments reveal that when the
weights are constructed using the parameter values of the true data
generating process, a version of the optimal jackknife estimator almost
always out-performs alternative semi-parametric bias-corrected estimators. A
feasible version of the jackknife estimator, in which the weights are
constructed using estimates of the unknown parameters, whilst not dominant
overall, is still the least biased estimator in some cases. Even when
misspecified short run dynamics are assumed in the construction of the
weights, the feasible jackknife still shows significant reduction in bias
under certain designs. As is not surprising, parametric maximum likelihood
estimation out-performs all semi-parametric methods when the true values of
the short memory parameters are known, but is dominated by the
semi-parametric methods (in terms of bias) when the short memory parameters
need to be estimated, and in particular when the model is misspecified.}

\medskip

\noindent {\footnotesize \emph{Keywords}: {\small Long memory; bias
adjustment; cumulants; discrete Fourier transform; periodograms;
log-periodogram regression}.\smallskip }

\noindent {\footnotesize \emph{MSC2010 subject classifications}: Primary
62M10, 62M15; Secondary 62G09\smallskip }\newline
\end{abstract}

\newpage

\baselineskip18pt

\section{Introduction\label{Sec1:Intro}}

Data on many climate, hydrological, economic and financial variables exhibit
dynamic patterns characterized by a long lasting response to past shocks.
Notable examples include, water levels in rivers (\color{blue}\citealp*{Hurst1951}\color{black}), rainfall (\color{blue}\citealp*{Gil-Alana2012}\color{black}), aggregate output (\color{blue}\citealp*{DieboldRudebusch1989}; \citealp*{HasslerWolters1995}\color{black}), interest rates (\color{blue}\citealp*{Baillie1996}\color{black}), exchange rates (\color{blue}\citealp*{Cheung2016}\color{black}) and stock market volatility (\color{blue}\citealp*{BollerslevMikkelsen1996}\color{black}; \color{blue}\citealp*{AndersenBollerslevDieboldLabys2003}\color{black}). Such `long memory processes' are characterized by non-summable autocovariances that decline at a (slow) hyperbolic rate, in contrast to the usual exponential, and summable, decay associated with a
short memory process; the fractionally integrated autoregressive moving average (ARFIMA) model of \color{blue}\cite{Adenstedt1974}\color{black},\ \color{blue}\cite{GrangerJoeyux1980}\color{black}\ and \color{blue}\cite{Hosking1981}\color{black}\ being a popular representation. Equivalently, a
stationary (potentially) long memory process, $\left\{ Y_{t}\right\} ,$ $%
t=0,\pm 1,\pm 2,\ldots ,$ can be represented by the spectral density,%
\begin{equation}
f_{YY}\left( \lambda \right) =\left( 2\sin \left( \lambda /2\right) \right)
^{-2d}f_{YY}^{\ast }\left( \lambda \right) ,\text{ for }\lambda \in \left[
-\pi ,\pi \right] ,  \label{Spectral density_ARFIMA}
\end{equation}%
where the fractional differencing parameter $d$ satisfies\textbf{\ }$d\in
\left( -0.5,0.5\right) $, and $f_{YY}^{\ast }\left( \cdot \right) $ is an
even function that is continuous on $\left[ -\pi ,\pi \right] $, is bounded
above and bounded away from zero, and satisfies $\int_{-\pi }^{\pi }\log
f_{YY}^{\ast }\left( \lambda \right) d\lambda =0$. The process is said to
have \textit{long} \textit{memory} when $d\in \left( 0,0.5\right) $, \textit{%
intermediate} \textit{memory} when $d\in \left( -0.5,0\right) $ and \textit{%
short} \textit{memory} when $d=0$. The factor $f_{YY}^{\ast }\left( \cdot
\right) $ controls the (remaining) short memory behaviour associated with
the process. For detailed expositions of processes described by (\ref{Spectral density_ARFIMA}), including applications, see, \color{blue}\cite{Beran1994_B}\color{black}, \color{blue}\cite{OppenheimTaqquMurad2003}\color{black}\ and \color{blue}\cite{Robinson2003}\color{black}.

In estimating the parameter $d$, the semi-parametric\textbf{\ } log-periodogram regression (LPR) estimator of \color{blue}\cite{GPH1983}\color{black}\ and \color{blue}\citet{Robinson1995a,Robinson1995b}\color{black}\ has been widely used, due to the simplicity of its
construction as an ordinary least squares (OLS) estimator, and its avoidance
of potentially incorrect specification of the short memory component.
However, consistency of the LPR estimator is achieved only at the cost of
both a slower rate of convergence than the usual parametric rate and
substantial finite sample bias in the presence of ignored short run dynamics
(see, for example, \color{blue}\citealp*{ANW1993}\color{black}\ and \color{blue}\citealp*{NielsenFrederiksen2005}\color{black}).

Given this well-documented bias, \textit{bias reduction} of the LPR
estimator has been a focus of the literature. \color{blue}\cite{AndrewsGuggenberger2003}\color{black}, for example, include additional
frequencies, to degree $2r$\ for $r\geq 0$, in the log-periodogram
regression that defines the LPR estimator, producing an estimator (denoted
hereafter by $\widehat{d}_{r}^{AG}$) whose bias converges to zero at a
faster rate than that of the LPR estimator (recovered by setting $r=0$),
when $r>1$. Alternative analytical procedures appear in \color{blue}\cite{MoulinesSoulier1999}\color{black}, \color{blue}\cite{Hurvich2001}\color{black}\ and \color{blue}\cite{RobinsonHenry2003}\color{black}, whilst a method based on the pre-filtered sieve bootstrap has been introduced by \color{blue}\cite{PoskittMartinGrose2016}\color{black}. Critically, all such bias-correction methods come at a cost: namely, an increase in asymptotic variance. Notably, \color{blue}\cite{GuggenbergerSun2006}\color{black}\ produce a weighted average of LPR estimators over different
bandwidths that achieves the same degree of bias reduction as $\widehat{d}%
_{r}^{AG}$\ for any given $r$, but with less variance inflation. This estimator, along with that of \color{blue}\cite{PoskittMartinGrose2016}\color{black}, serve as important comparators for the\textbf{\ }alternative
bias-corrected estimator that we develop herein.

The approach to bias adjustment adopted in this paper\ applies the jackknife
principle, with the bias-corrected estimator constructed as a weighted
average of LPR estimators computed, in turn, from the full sample and $m$\
sub-samples of a given length. The sub-samples may be created by using
either the non-overlapping or the moving-block method. Motivated by the jackknife technique proposed by \color{blue}\cite{ChenYu2015}\color{black}\
in a unit root setting, weights are chosen to remove bias up to a given
order and, at the same time, to minimize the increase in asymptotic
variance. The weights are `optimal' in this sense and the associated
jackknife estimator referred to as `optimal' accordingly. In the fractional
setting, with the LPR estimator being the method to be adjusted, these
optimal weights\ involve two types of covariance terms: $(i)$\ covariances
between the full-sample and sub-sample log-periodogram ordinates, and $(ii)$%
\ covariances between distinct sub-sample log-periodogram values. These
covariance terms may, in turn, be represented by cumulants of the discrete
Fourier transform (DFT) of the time series. Building on results in \color{blue}\citet[Chapters~2 and 4]{Brillinger1981}\color{black}, we first derive closed-form expressions for the association
between the corresponding DFTs in terms of cumulants. These expressions are
used to derive the form of dependence between the periodograms (at a given
ordinate or at different ordinates) associated with the full sample and the
sub-samples, which allows us to obtain closed-form expressions for the
covariances terms, $\left( i\right) $ and $\left( ii\right) $, and, hence,
to evaluate the optimal weights.

We prove the consistency and asymptotic normality of the optimal jackknife
estimator. Most notably, we establish that the convergence rate and
asymptotic variance are equal to those of the unadjusted LPR estimator. This
implies that there is \textit{no} inflation\ in asymptotic efficiency
compared to the \textit{unadjusted} LPR estimator of $d,$ despite the bias
reduction that is achieved. This compares with \color{blue}\cite{GuggenbergerSun2006}\color{black}, in which the goal is to produce an
estimator (for a given value of $r$) with an asymptotic variance that is
smaller than that of the corresponding bias-adjusted estimator of \color{blue}\cite{AndrewsGuggenberger2003}\color{black}, as based on the
same value of $r$, $\widehat{d}_{r}^{AG}$. In particular, in the case where $r=0$, and no bias adjustment is achieved (with $\widehat{d}_{r}^{AG}$
equivalent to the raw LPR estimator), the\ estimator of \color{blue}\citeauthor{GuggenbergerSun2006}\color{black}\ is still biased, but with a (possibly) reduced asymptotic
variance. In addition, in contrast with \color{blue}\citeauthor{GuggenbergerSun2006}\color{black}, and the other analytical bias adjustment methods cited above,
our theoretical results do not rely on the assumption of Gaussianity.
Specifically, expressions for the dominant bias term and variance of the LPR
estimator -- needed in the construction of the jackknife estimator and as originally derived by \color{blue}\cite{HurvichDeoBrodshky1998}\color{black}\ for fractional \textit{Gaussian} processes - are shown to hold under
non-Gaussian assumptions. Hence, all theoretical results for the
bias-adjusted estimator hold under similar generality.\footnote{%
We refer the reader to \color{blue}\cite{HahnNewey2004}\color{black}, %
\color{blue}\cite{Chambers2013}\color{black}, \color{blue}\cite{ChenYu2015}%
\color{black}\ and \color{blue}\cite{RobinsonKaufmann2015}\color{black}\ for
other applications of the jackknife in time series settings. To our
knowledge the technique has been used only once in a long memory setting%
\textit{\ per se}, namely in the numerical work of\ \color{blue}\cite%
{EkonomiButka2011}\color{black}, where the method of\ \color{blue}\cite%
{Chambers2013}\color{black}\ is adopted for the purpose of reducing the bias
of the LPR estimator to the first order. However, no rigorous proofs of the
properties of the estimator are provided, and no attempt at yielding an optimal estimator in the sense given in the current paper, is made.}

Extensive simulation exercises are conducted in order to compare the finite
sample performance of the jackknife estimator with that of alternative
approaches, including the bias-adjusted estimators of \color{blue}\cite{GuggenbergerSun2006}\color{black}\ and \color{blue}\cite{PoskittMartinGrose2016}\color{black}. Results show that certain versions of
the optimally bias-corrected\ jackknife estimator outperform the alternative bias-adjusted estimators of \color{blue}\citeauthor{GuggenbergerSun2006}%
\color{black}\ and \color{blue}\citeauthor{PoskittMartinGrose2016}\color{black}, in terms of bias-reduction and root mean squared error
(RMSE), with the RMSE being somewhat close to, or even smaller than, that of
the LPR in some cases.\textbf{\ }In the empirically realistic case where the
true values of the parameters - required in order to evaluate the optimal
weights in the jackknife estimator - are unknown, we implement the jackknife
technique using an iterative procedure. This feasible version of the
estimator does not consistently outperform either the bootstrap-based
estimator of \color{blue}\citeauthor{PoskittMartinGrose2016}\color{black}\ or (a feasible version of) the method of \color{blue}\citeauthor{GuggenbergerSun2006}\color{black}, but is not substantially
inferior, in terms of either bias or RMSE, and is sometimes still the least
biased estimator of all.

We assess the finite sample performance of all bias-adjusted estimators
under scenarios of both correct model specification and misspecification
and, for completeness, parametric methods based on maximum likelihood
estimation (MLE) and pre-whitening are included in the assessment.\footnote{%
We thank a referee for these suggestions.} As would be anticipated, given
the asymptotic efficiency of MLE under correct specification, no
semi-parametric method out-performs the optimal parametric approach in terms
of RMSE in this case. However, when the short memory dynamics need to be
estimated, a semi-parametric method is typically less biased than both
parametric methods. When the model is misspecified, the semi-parametric
methods are dominant in terms of both bias and RMSE, with the feasible
jackknife estimator producing the least bias in some cases, most notably
when the true process has a moving average component that is omitted in the
model specification.

In summary, the paper makes two important contributions to the literature on
semi-parametric estimation in fractional models. First, a new estimator is
derived that bias-corrects the popular LPR estimator to a given order, with
no associated variance inflation asymptotically. Second, that estimator is
shown to perform well in finite samples, under ideal conditions, and to hold
its own in empirically relevant scenarios, relative to existing comparators.

The remainder of the paper is organized as follows. In Section~\ref{Sec2:LPR}, we introduce two log-periodogram regression estimators; namely, the LPR estimator originally proposed by \color{blue}\cite{GPH1983}\color{black}\ and the particular bias-reduced estimator of \color{blue}\cite{GuggenbergerSun2006}\color{black}. In Section \ref{Sec3:Jackknife LPR}, we
develop the new jackknife estimator that accommodates both bias correction
and variance minimization\ via the appropriate choice of weights. All
theoretical results pertaining to the construction of the afore-mentioned
covariance terms, and the resultant asymptotic properties of the optimal
estimator, are given in Section~\ref{Sec3:Properties}. Section \ref%
{Sec4:Simulation} documents the finite sample performance of the estimator
by means of a\ Monte Carlo study.

The proofs of\ all results are contained in Appendix A, while\ Appendix B
provides various technical results, including the evaluation of the
covariances required for the construction of the weights for the optimal
jackknife estimator. Appendix C contains Tables \ref%
{Table:Bias_ARFIMA(1,d,0)_b} to \ref{Table:MSE_ARFIMA(1,d,1)_mis}, which
document the results of the Monte Carlo study, with these results summarized
briefly in Table \ref{rank} in the same appendix. The following notation is
used throughout: \textquotedblleft $\rightarrow ^{P}$\textquotedblright\
denotes convergence in probability, \textquotedblleft $\rightarrow ^{D}$%
\textquotedblright\ denotes convergence in distribution, and
\textquotedblleft $\rightarrow $\textquotedblright\ is used to indicate the
limit as $n\rightarrow \infty $, (unless otherwise stated). The $k^{th}$%
-order spectral density function of the time series $\left\{ X_{t}\right\} $
is denoted by $f_{X\ldots X}\left( \lambda _{1},\lambda _{2},\ldots
,,\lambda _{k-1}\right) $, where $\lambda _{1},\lambda _{2},\ldots ,,\lambda
_{k-1}$ are fundamental frequencies. For instance, the density function
given in (\ref{Spectral density_ARFIMA}) is the second-order spectral
density of $\left\{ Y_{t}\right\} .$

\section{Log-periodogram regression estimation methods\label{Sec2:LPR}}

In this section we briefly review two log-periodogram regression estimators;
namely, the raw (unadjusted) LPR estimator\ and the bias-reduced weighted-average estimator of \color{blue}\cite{GuggenbergerSun2006}\color{black}\ (GS). These estimators are used as benchmarks for later
comparisons, and the raw LPR estimator, of course, underpins the jackknife
method developed in Section \ref{Sec3:Jackknife LPR}. We summarize the
asymptotic properties of these estimators and the assumptions underlying
those properties. In contrast to earlier proofs related to the LPR estimator
(e.g. \color{blue}\citealp*{HurvichDeoBrodshky1998}\color{black})\ we do not assume that the data generating process (DGP) is
Gaussian. This extension to non-Gaussian processes means that the properties
subsequently derived for the optimal jackknife estimator are also applicable
for this general case.

\subsection{The log-periodogram regression estimator}

Let $\mathbf{y}^{\top }=(y_{1},y_{2},...,y_{n})$\ be a sample of $n$
observations from a process with a spectral density as given in~(\ref%
{Spectral density_ARFIMA}). The LPR estimator, $\widehat{d}_{n},$ is
motivated by the following simple linear regression model that is formed
directly from the spectral density given in (\ref{Spectral density_ARFIMA}), 
\begin{equation}
\log I_{Y}^{\left( n\right) }\left( \lambda _{j}\right) =(\log f_{YY}^{\ast
}(0)-C)-2d\log (2\sin (\lambda _{j}/2))+\xi _{j},  \label{reg:mod}
\end{equation}%
where 
\begin{equation}
I_{Y}^{\left( n\right) }\left( \lambda \right) =|D_{Y}^{\left( n\right)
}\left( \lambda \right) |^{2},\quad D_{Y}^{\left( n\right) }\left( \lambda
\right) =\frac{1}{\sqrt{2\pi n}}\sum_{t=1}^{n}y_{t}\exp \left( -\imath
\lambda t\right) ,  \label{Periodogram}
\end{equation}%
and $D_{Y}^{\left( n\right) }\left( \lambda _{j}\right) $ is the DFT of the
vector of realizations, $\mathbf{y},$ measured at Fourier frequencies, $%
\lambda _{j}=2\pi j/n,$ $(j=1,2...,N_{n}),$ $N_{n}=\lfloor n^{\alpha
}\rfloor $ for $0<\alpha <1,$ and $\imath =\sqrt{-1}$ is the imaginary unit.
Here, the error terms $\xi _{j}=\log \left( I_{Y}^{(n)}\left( \lambda
_{j}\right) /f_{YY}(\lambda _{j})\right) +C+V_{j},$ $j=1,2,...,N_{n}$, where%
\begin{equation}
V_{j}=\log \left( f_{YY}^{\ast }(\lambda _{j})/f_{YY}^{\ast }(0)\right) ,
\label{v}
\end{equation}%
are assumed to be asymptotically independently and identically distributed (%
\textit{i.i.d.}) and $C$ is the Euler constant. The LPR estimator of $d$ is
simply the OLS estimator of the slope\ parameter in\ (\ref{reg:mod}) and is
given by%
\begin{equation}
\widehat{d}_{n}=\frac{-0.5\sum_{j=1}^{N_{n}}\left( x_{j}-\overline{x}\right)
z_{j}}{\sum_{j=1}^{N_{n}}(x_{j}-\overline{x})^{2}},  \label{LPR}
\end{equation}%
where $z_{j}=\log I_{Y}^{\left( n\right) }(\lambda _{j}),$ $x_{j}=\log
(2\sin (\lambda _{j}/2)),$ and $\overline{x}=\frac{1}{N_{n}}%
\sum\nolimits_{j=1}^{N_{n}}x_{j}.$ The subscript $n$ is introduced here in
order to distinguish this full-sample version of the estimator from that
computed subsequently from sub-samples, in the process of applying the
jackknife.

Certain statistical properties of the LPR estimator such as its bias,
variance, mean-squared-error (MSE) and asymptotic distribution have been derived by \color{blue}\cite{HurvichDeoBrodshky1998}\color{black}\ under
given regularity conditions, and with certain approximations invoked.
Alternative expressions for the bias and variance of the LPR estimator are
provided in Theorem $1$ of \color{blue}\cite{AndrewsGuggenberger2003}%
\color{black}, plus in Theorem 3.1 of \color{blue}\cite{GuggenbergerSun2006}%
\color{black}, by setting $r=0$. \color{blue}\cite{Lieberman2001}%
\color{black}\ also provides a formula for the expectation of the LPR
estimator under the same\ conditions as \color{blue}%
\citeauthor{HurvichDeoBrodshky1998}%
\color{black}; however, his expression is an infinite sum of\ a quantity
that depends on the true values of $d$ and the short memory parameters,
which renders a feasible version of the jackknife technique using his
expression more cumbersome.

With all results cited above derived under the assumption of Gaussianity, we
now extend the results stated in Theorems 1 and 2 of \color{blue}\cite{HurvichDeoBrodshky1998}\color{black}\ to the general (potentially
non-Gaussian) case. In particular, the resultant expression for the
expectation of the LPR estimator is used in the specification of the optimal
jackknife estimator, and in the proof of its properties.

We begin with the following assumptions on the DGP:

\begin{enumerate}
\item[$\left( A.1\right) $] There exists $G>0,$ such that%
\begin{equation*}
f_{YY}(\lambda )=G\lambda ^{-2d}+O(\lambda ^{2-2d})\text{ as }\lambda
\rightarrow 0+,
\end{equation*}%
where `$\rightarrow 0+$' denotes an approach from above.

\item[$\left( A.2\right) $] In a neighbourhood $\left( 0,\varepsilon \right) 
$ of the origin, $f_{YY}(\lambda )$ is differentiable on $\left[ -\pi ,\pi %
\right] \backslash \left\{ 0\right\} $ and%
\begin{equation*}
\left\vert \frac{d}{d\lambda }\log f_{YY}(\lambda )\right\vert =O(\lambda
^{-1}),\text{ as }\lambda \rightarrow 0+.
\end{equation*}%
In addition, for some $0<\widetilde{B}_{2}$, $\widetilde{B}_{3}<\infty $, $%
f_{YY}^{\ast \prime }(0)=0,$ $\left\vert f_{YY}^{\ast \prime \prime
}(\lambda )\right\vert <\widetilde{B}_{2}$ and $\left\vert f_{YY}^{\ast
\prime \prime \prime }(\lambda )\right\vert <\widetilde{B}_{3}$, where $%
f_{YY}^{\ast \prime }\left( \lambda \right) ,$ $f_{YY}^{\ast \prime \prime
}\left( \lambda \right) $ and $f_{YY}^{\ast \prime \prime \prime }\left(
\lambda \right) $ denote, respectively, the first-, second- and third-order
derivatives of $f_{YY}^{\ast }$ with respect to $\lambda $ in a neighborhood
of zero$.$

\item[$\left( A.3\right) $] $\left\{ Y_{t}\right\} ,$ $t\in \mathbb{Z}%
:=\{0,\pm 1,\pm 2,\cdots \},$ satisfies%
\begin{equation*}
Y_{t}-\mu _{Y}=\sum_{j=0}^{\infty }b_{j}\varepsilon _{t-j},\text{ \ }%
\sum_{j=0}^{\infty }b_{j}^{2}<\infty ,\text{ \ }\left\vert \frac{d}{d\lambda 
}b(\lambda )\right\vert =O(\lambda ^{-1})\text{ \ \ as }\lambda \rightarrow
0+,
\end{equation*}%
where $b(\lambda )=\sum_{j=0}^{\infty }b_{j}\exp \left( \imath j\lambda
\right) $ and $\left\{ \varepsilon _{t}\right\} $ is a strictly stationary
process with $E\left( \varepsilon _{t}\right) =0$ and$\ E\left( \varepsilon
_{t}^{2}\right) =1$.

\item[$\left( A.4\right) $] The innovation process $\left\{ \varepsilon
_{t}\right\} $ satisfies the conditions in $\left( A.3\right) $. In
addition, $E\left( \varepsilon _{t}\right) ^{3}<\infty $ and $E\left(
\varepsilon _{t}\right) ^{4}<\infty $.
\end{enumerate}

\noindent Assumptions $\left( A.1\right) -\left( A.3\right) $\ are standard
in the long memory literature (see, \color{blue}\citealp*{FoxTaqqu1986}\color{black}, \color{blue}\citealp*{HurvichDeoBrodshky1998}\color{black}\ \color{blue}\citealp*{LiebermanRosemarinRousseau2012}\color{black}, among others) and are satisfied by the class of ARFIMA
models. The boundedness of the first three derivatives of $f_{YY}^{\ast }$\
in Assumption $\left( A.2\right) $\ is required to control\ the fourth-order
moment of the sine and cosine components of the standardized DFTs that are
used to derive the bias term of the LPR. Assumption $\left( A.4\right) $
specifies the third and fourth moments of $\left\{ \varepsilon _{t}\right\} $%
\ to be finite,\ as we do not invoke Gaussianity. The boundedness imposed on
the higher-order moments of $\left\{ \varepsilon _{t}\right\} $ ensures the
asymptotic normality of the DFTs associated with the process $\left\{
Y_{t}\right\} $. The asymptotic normality of the DFTs is, in turn, used in
proving Theorems \ref{LPR prop} -- \ref{AP_LPR}.

We now state Theorem \ref{LPR prop}, which gives the mean, variance and
asymptotic distribution of the LPR estimator. We subsequently exploit these
results to construct the optimal jackknife estimator, and to prove its
properties, in Section \ref{Sec3:Jackknife LPR}.

\begin{theorem}
\label{LPR prop}\noindent Let Assumptions $\left( A.1\right) -\left(
A.3\right) $ hold. Given $N_{n}\rightarrow \infty ,$ $n\rightarrow \infty ,$
with $\frac{N_{n}\log N_{n}}{n}\rightarrow 0,$ 
\begin{eqnarray}
E%
\big(%
\widehat{d}_{n}%
\big)
&=&d_{0}-\frac{2\pi ^{2}}{9}\frac{f_{YY}^{\ast \prime \prime }\left(
0\right) }{f_{YY}^{\ast }\left( 0\right) }\frac{N_{n}^{2}}{n^{2}}+o%
\Big(%
\frac{N_{n}^{2}}{n^{2}}%
\Big)%
+O%
\Big(%
\frac{\log ^{3}N_{n}}{N_{n}}%
\Big)%
,  \label{Expectation of LPR} \\
Var%
\big(%
\widehat{d}_{n}%
\big)
&=&\frac{\pi ^{2}}{24N_{n}}+o%
\Big(%
\frac{1}{N_{n}}%
\Big)
\label{Variance of LPR}
\end{eqnarray}%
and $\widehat{d}_{n}\rightarrow ^{P}d_{0}.$ Given that $\left( A.4\right) $
also holds and if $N_{n}=o\left( n^{4/5}\right) $\ and $\log ^{2}n=o\left(
N_{n}\right) ,$\ then,%
\begin{equation}
\sqrt{N_{n}}(\widehat{d}_{n}-d_{0})\rightarrow ^{D}N%
\Big(%
0,\tfrac{\pi ^{2}}{24}%
\Big)%
\text{ as }n\rightarrow \infty .  \label{Distribution of LPR}
\end{equation}
\end{theorem}

\subsection{The weighted-average log-periodogram regression estimator\label%
{Sec2:WLPR}}

The motivation for the estimator of \color{blue}\cite{GuggenbergerSun2006}\color{black}\ stems from the work of \color{blue}\cite{AndrewsGuggenberger2003}\color{black}. With (\ref{v}) being the term that
causes the dominant bias in the LPR estimator, \color{blue}\citeauthor{AndrewsGuggenberger2003}\color{black}\ use a Taylor series expansion around $j=0$ to approximate (%
\ref{v}) as an even polynomial in the frequencies of order $r.$\footnote{%
The odd-order terms of the Taylor's expansion around zero are exactly zero.
This leads to the expansion with only even-order terms.} Including the first 
$2r$ terms (with $r\geq 1$) in the log-periodogram regression in (\ref%
{reg:mod}) as additional regressors leads to%
\begin{equation}
\ln I_{Y}^{\left( n\right) }\left( \lambda _{j}\right) =(\log f_{YY}^{\ast
}(0)-C)-2d\log (2\sin (\lambda _{j}/2))+\sum\limits_{k=1}^{r}\frac{b_{2k}}{%
\left( 2k\right) !}\lambda _{j}^{2k}+\zeta _{j},  \label{reg:pseudo}
\end{equation}%
where $\zeta _{j}=\xi _{j}-\sum\nolimits_{k=1}^{r}\frac{b_{2k}}{\left(
2k\right) !}\lambda _{j}^{2k}$. Application of OLS to (\ref{reg:pseudo})
then yields an estimator of $d$, $\widehat{d}_{r}^{AG}$, with reduced bias
relative to the raw LPR estimator, $\widehat{d}_{n}$. The bias-adjusted
estimator is shown to be $\sqrt{N_{n}}$- consistent, with an asymptotic
variance equal to $\tfrac{\pi ^{2}}{24}c_{r}$, with $c_{r}>1$ for $r\geq 1$
and $c_{r}=1$ for $r=0$.

\color{blue}\cite{GuggenbergerSun2006}\color{black}\ proceed to show that an
appropriate weighted average of raw LPR estimators, as based on different
bandwidths, $N_{n,i}=\left\lfloor q_{i}N_{n}\right\rfloor ;$ $i=1,\ldots ,K,$
for fixed numbers $q_{i}$ chosen suitably, has the same asymptotic bias as $%
\widehat{d}_{r}^{AG}$ (constructed using $N_{n}$), but with a reduced
asymptotic variance. That is, bias reduction is achieved at a smaller cost than is the original method of \color{blue}\cite{AndrewsGuggenberger2003}%
\color{black}. Further, for the case of $r=0$, the bias of the raw LPR
estimator is retained but with reduced asymptotic variance. The authors also
demonstrate that the weighted-average estimator, denoted by $\widehat{d}%
_{r}^{GS}$\ hereafter, can be implemented via a simple two-step procedure.
In the first step, a series of $K$ LPR estimates are obtained using the
regression model in (\ref{reg:mod}) and\textbf{\ }for bandwidths, $N_{n,i},$ 
$i=1,\ldots ,K.$ Then, in the second step, the following pseudo-regression
is estimated, using the $K$ estimates produced in the first step as
observations of the dependent variable in the regression,\textbf{\ } 
\begin{equation}
\widehat{d}_{N_{n,i}}=d+\sum_{j=1}^{r}\beta _{2j}q_{i}^{2j}+\beta _{2+2r}%
\Big(%
q_{i}^{2+2r}-\delta \sum_{p=1}^{K}q_{p}^{2+2r}%
\Big)%
+u_{i},\text{ }i=1,\ldots ,K,  \label{dhat_GS}
\end{equation}%
where $u_{i}$ is the error term, and $\mathbf{u}^{\top
}=(u_{1},u_{2},...,u_{K})$ has a zero (vector) mean and asymptotic
variance-covariance matrix, 
\begin{equation*}
\mathbf{\Omega }=\left( \Omega _{i,j}\right) \in 
\mathbb{R}
^{K\times K},\text{ with }\Omega _{i,j}=\frac{1}{\max \left(
q_{i},q_{j}\right) }.
\end{equation*}%
The tuning parameter $\delta $ on the right-hand-side of (\ref{dhat_GS}) is
a fixed non-zero constant that is used to control the multiplicative
constant of the dominant bias term and render that term equivalent to the
dominant bias term of $\widehat{d}_{r}^{AG}$. The estimator, $\widehat{d}%
_{r}^{GS}$, is then defined as the first component of the GLS estimator of $%
\left( d,\mathbf{\beta }^{\top }\right) ^{\top }$, where $\mathbf{\beta }%
^{\top }=\left( \beta _{2},\beta _{4},\ldots ,\beta _{2+2r}\right) ,$ that
is,%
\begin{equation}
\Big(%
\widehat{d}_{r}^{GS},\widehat{\mathbf{\beta }}^{\top }%
\Big)%
^{\top }=\left( \mathbf{Z}^{\top }\mathbf{\Omega }^{-1}\mathbf{Z}\right)
^{-1}\mathbf{Z}^{\top }\mathbf{\Omega }^{-1}\widehat{\mathbf{d}},
\label{GLS}
\end{equation}%
where $\widehat{\mathbf{d}}$ is the $(K\times 1)$ dimensional vector with $%
i^{th}$ element $\widehat{d}_{N_{n,i}}$, and%
\begin{equation*}
\mathbf{Z}^{\top }=%
\big(%
\mathbf{z}_{1},\ldots ,\mathbf{z}_{K}%
\big)%
\in 
\mathbb{R}
^{\left( 2+r\right) \times K},\text{ with }\mathbf{z}_{i}^{\top }=%
\Big(%
1,q_{i}^{2},\ldots ,q_{i}^{2r},%
\Big(%
q_{i}^{2+2r}-\delta \sum_{p=1}^{K}q_{p}^{2+2r}%
\Big)%
\Big)%
.
\end{equation*}

Both the raw LPR estimator, $\widehat{d}_{n}$, and the weighted-average estimator, $\widehat{d}_{r}^{GS}$, with $r=1$, are used as comparators of our proposed jackknife procedure in the Monte Carlo simulation exercises in Section~\ref{Sec4:Simulation}.

\section{The optimal jackknife log-periodogram regression estimator\label%
{Sec3:Jackknife LPR}}

\subsection{Definition of the jackknife estimator}

The\ idea behind jackknifing is to generate a set of sub-samples, by
deleting one or more observations of the original sample, while preserving
the structure of dependence within the sub-samples; the aim being to use
(weighted) sub-sample estimates to produce a bias-corrected estimator of the
parameter of interest. Let $\mathbf{y}_{i}$ $\left( i=1,2,...,m\right) $
denote a set of $m$ sub-samples of $\mathbf{y},$ each of which has equal
length, $l,$ such that $n=l\times m.$ If sub-samples are chosen using the
`non-overlapping' method, then $\mathbf{y}_{i}^{\top }=\left( y_{\left(
i-1\right) l+1},\ldots ,y_{il}\right) $ for $i=1,\ldots ,m$; alternatively
if the sub-sampling scheme is `moving-block' then $\mathbf{y}_{i}^{\top
}=\left( y_{i},\ldots ,y_{i+l-1}\right) $ for all $i$. In the current
context we use the jackknife technique to bias correct the LPR estimator.
Hence, we need to produce the full-sample estimator, $\widehat{d}_{n}$, and
the LPR estimators produced by applying OLS to the model in (\ref{reg:mod}),
using the relevant sub-sample. We denote these $m$ sub-sample estimators
(based on either the non-overlapping or moving-block method) by $\widehat{d}%
_{i}$, $i=1,2,...,m.$ We summarize notation corresponding to the full-sample
estimation and both forms of sub-sample estimation in Table \ref%
{Table:Notations}, for ease of subsequent referencing.

\begin{table}[tbp]
\vspace{-0.5cm}%
\caption{\protect\small Quantities related to the full sample and the
sub-samples used in the construction of the jackknife estimator}\label%
{Table:Notations}

\medskip

\begin{tabular}{lc|c|c}
\hline\hline
&  &  &  \\ 
&  & Full sample & $i^{th}$ sub-sample \\ \hline
&  &  &  \\ 
$\left( i\right) $ & \multicolumn{1}{l|}{\small Frequency} & 
\multicolumn{1}{|l|}{$\lambda _{j}=2\pi j/n$} & \multicolumn{1}{|l}{$\mu
_{j}=2\pi j/l=2\pi jm/n=m\lambda _{j}$} \\ 
$\left( ii\right) $ & \multicolumn{1}{l|}{\small Frequency range} & 
\multicolumn{1}{|l|}{$j=1,...,N_{n}$} & \multicolumn{1}{|l}{$j=1,...,N_{l}$}
\\ 
$\left( iii\right) $ & \multicolumn{1}{l|}{\small Spectral density} & 
\multicolumn{1}{|l|}{$f_{YY}\left( \lambda \right) =\left( 2\sin \left(
\lambda /2\right) \right) ^{-2d}f_{YY}^{\ast }\left( \lambda \right) $} & 
\multicolumn{1}{|l}{$f_{Y_{i}Y_{i}}\left( \mu \right) =\left( 2\sin \left(
\mu /2\right) \right) ^{-2d}f_{Y_{i}Y_{i}}^{\ast }\left( \mu \right) $} \\ 
$\left( iv\right) $ & \multicolumn{1}{l|}{\small DFT} & \multicolumn{1}{|l|}{%
$D_{Y}^{\left( n\right) }\left( \lambda \right) =\frac{1}{\sqrt{2\pi n}}%
\sum\nolimits_{t=1}^{n}y_{t}\exp \left( -\imath \lambda t\right) $} & 
\multicolumn{1}{|l}{$D_{Y_{i}}^{\left( l\right) }\left( \mu \right) =\frac{1%
}{\sqrt{2\pi l}}\sum\nolimits_{t=1}^{l}y_{t+i^{\prime }}\exp \left( -\imath
\mu t\right) $} \\ 
$\left( v\right) $ & \multicolumn{1}{l|}{\small Periodogram} & 
\multicolumn{1}{|l|}{$I_{Y}^{\left( n\right) }\left( \lambda \right)
=|D_{Y}^{\left( n\right) }\left( \lambda \right) |^{2}$} & 
\multicolumn{1}{|l}{$I_{Y_{i}}^{\left( l\right) }\left( \mu \right)
=|D_{Y_{i}}^{\left( l\right) }\left( \mu \right) |^{2}$} \\ 
$\left( vi\right) $ & \multicolumn{1}{l|}{\small Error term} & 
\multicolumn{1}{|l|}{$\xi _{j}=\log \left( I_{Y}^{\left( n\right) }\left(
\lambda _{j}\right) /f_{YY}\left( \lambda _{j}\right) \right) $} & 
\multicolumn{1}{|l}{$\xi _{j}^{(i)}=\log \left( I_{Y_{i}}^{\left( l\right)
}\left( \mu _{j}\right) /f_{Y_{i}Y_{i}}\left( \mu _{j}\right) \right) $} \\ 
& \multicolumn{1}{l|}{\small Other notation:} &  &  \\ 
$\left( vii\right) $ &  & \multicolumn{1}{|l|}{$x_{j}=\log (2\sin \left(
\lambda _{j}/2\right) )$} & \multicolumn{1}{|l}{$x_{j}^{^{\prime }}=\log
(2\sin \left( \mu _{j}/2\right) )$} \\ 
$\left( viii\right) $ &  & \multicolumn{1}{|l|}{$\overline{x}=\left.
\sum\nolimits_{t=1}^{N_{n}}x_{j}\right/ N_{n}$} & \multicolumn{1}{|l}{$%
\overline{x^{^{\prime }}}=\left. \sum\nolimits_{t=1}^{N_{l}}x_{j}^{^{\prime
}}\right/ N_{l}$} \\ 
$\left( ix\right) $ &  & \multicolumn{1}{|l|}{$a_{j}=x_{j}-\overline{x}$} & 
\multicolumn{1}{|l}{$a_{j}^{^{\prime }}=x_{j}^{^{\prime }}-\overline{%
x^{^{\prime }}}$} \\ 
$\left( x\right) $ &  & \multicolumn{1}{|l|}{$S_{xx}=\sum%
\nolimits_{j=1}^{N_{n}}a_{j}^{2}$} & \multicolumn{1}{|l}{$S_{xx}^{^{\prime
}}=\sum\nolimits_{j=1}^{N_{l}}a_{j}^{^{\prime }2}$} \\ 
&  &  &  \\ \hline\hline
\end{tabular}

\medskip {\small Note, regarding the sub-sample notation in point }${\small %
(iv)}$, {\small if the sub-samples are drawn with the non-overlapping scheme then, }${\small i}^{\prime }{\small =(i-1)l.}${\small \ If the moving-block scheme is used then, }${\small i}^{\prime }{\small =i-1.}$
\end{table}
Define the jackknife estimator, $\widehat{d}_{J,m},$ as%
\begin{equation}
\widehat{d}_{J,m}=w_{n}\widehat{d}_{n}-\sum_{i=1}^{m}w_{i}\widehat{d}_{i},
\label{Jackknife statistic of d}
\end{equation}%
where $w_{n}$ and $\left\{ w_{i}\right\} _{i=1}^{m}$ are the weights
assigned to the full-sample estimator and the sub-sample estimators,
respectively. Re-iterating, $\widehat{d}_{n}$ is the LPR estimator obtained
from the full sample (as defined directly in (\ref{LPR}))\textbf{\ }and $%
\widehat{d}_{i}$ $\left( i=1,2,...,m\right) $ denotes the $i^{th}$
sub-sample LPR estimator. Under the conditions of Theorem \ref{LPR prop}, it
is straightforward to show that%
\begin{eqnarray}
E%
\big(%
\widehat{d}_{J,m}%
\big)
&=&%
\big(%
w_{n}-\sum_{i=1}^{m}w_{i}%
\big)%
d_{0}-%
\Big(%
\frac{2\pi ^{2}}{9}\frac{f_{YY}^{\ast \prime \prime }\left( 0\right) }{%
f_{YY}^{\ast }\left( 0\right) }\frac{N_{n}^{2}}{n^{2}}w_{n}-\frac{2\pi ^{2}}{%
9}\frac{f_{Y_{i}Y_{i}}^{\ast \prime \prime }\left( 0\right) }{%
f_{Y_{i}Y_{i}}^{\ast }\left( 0\right) }\frac{N_{l}^{2}}{l^{2}}%
\sum_{i=1}^{m}w_{i}%
\Big)
\notag \\
&&+o%
\Big(%
\frac{N_{n}^{2}}{n^{2}}%
\Big)%
+O%
\Big(%
\frac{\log ^{3}N_{n}}{N_{n}}%
\Big)%
,  \label{Expectation of dhat_J}
\end{eqnarray}%
and%
\begin{eqnarray}
Var%
\big(%
\widehat{d}_{J,m}%
\big)
&=&\frac{\pi ^{2}}{24N_{n}}w_{n}^{2}+\frac{\pi ^{2}}{24N_{l}}%
\sum_{i=1}^{m}w_{i}^{2}+2\sum_{i=1}^{m-1}\sum_{j=i+1}^{m}w_{i}w_{j}Cov%
\big(%
\widehat{d}_{i},\widehat{d}_{j}%
\big)
\notag \\
&&-2w_{n}\sum_{i=1}^{m}w_{i}Cov%
\big(%
\widehat{d}_{n},\widehat{d}_{i}%
\big)%
+o%
\Big(%
\frac{1}{N_{n}}%
\Big)%
.  \label{Variance of dhat_J}
\end{eqnarray}%
The covariance between the full-sample LPR estimator and each sub-sample LPR
estimator, $Cov\left( \widehat{d}_{n},\widehat{d}_{i}\right) ,$ and the
covariances\textbf{\ }between the different sub-sample LPR estimators, $%
Cov\left( \widehat{d}_{i},\widehat{d}_{j}\right) ,$ for $i\neq j,$ $%
i,j=1,2,...,m,$ are given respectively by,%
\begin{eqnarray}
Cov%
\big(%
\widehat{d}_{n},\widehat{d}_{i}%
\big)
&=&\frac{1}{4S_{xx}}\frac{1}{S_{xx}^{^{\prime }}}\sum\limits_{j=1}^{N_{n}}%
\sum\limits_{k=1}^{N_{l}}a_{j}a_{k}^{(i)}Cov%
\big(%
\log I_{Y}^{\left( n\right) }\left( \lambda _{j}\right) ,\log
I_{Y_{i}}^{\left( l\right) }\left( \mu _{k}\right) 
\big)
\label{Covariance(dhat,dhat_i)} \\
Cov%
\big(%
\widehat{d}_{i},\widehat{d}_{i^{\prime }}%
\big)
&=&\frac{1}{4}\frac{1}{\left( S_{xx}^{^{\prime }}\right) ^{2}}%
\sum\limits_{j=1}^{N_{l}}\sum\limits_{k=1}^{N_{l}}a_{j}^{\prime
}a_{k}^{\prime }Cov%
\big(%
\log I_{Y_{i}}^{\left( l\right) }\left( \mu _{j}\right) ,\log
I_{Y_{i^{\prime }}}^{\left( l\right) }\left( \mu _{k}\right) 
\big)%
,  \label{Covariance(dhat_i,dhat_j)}
\end{eqnarray}%
with all notation as defined in Table \ref{Table:Notations}.

Our aim is to obtain\ the set of weights, $\left\{ w_{n},w_{1},\ldots
,w_{m}\right\} ,$ such that $\widehat{d}_{J,m}$ has the following properties:

\begin{enumerate}
\item[$\left( P.1\right) $] $\widehat{d}_{J,m}$ is an asymptotically
unbiased estimator of $d_{0}$, with bias reduced to an order of $o(\left.
N_{n}^{2}\right/ n^{2})$, and,

\item[$\left( P.2\right) $] $\widehat{d}_{J,m}$ achieves minimum variance
among all such bias-reduced estimators.
\end{enumerate}

\noindent \noindent The `optimal' jackknife estimator so defined is derived
via the Lagrangian method in the following section. In Section 4, the
asymptotic properties of the covariances in (\ref{Covariance(dhat,dhat_i)})
and (\ref{Covariance(dhat_i,dhat_j)}) that determine the asymptotic
behaviour of the estimator are derived, and the asymptotic efficiency of the
estimator then proven.

\subsection{Derivation of the optimal estimator\label{Optimization}}

The minimization problem is formulated as follows. Produce weights, $\left\{
w_{n},w_{1},\ldots ,w_{m}\right\} $, that satisfy:\textbf{\ } 
\begin{equation}
\min_{w_{n},\left\{ w_{i}\right\} _{i=1}^{m}}Var%
\big(%
\widehat{d}_{J,m}%
\big)%
,  \label{opt}
\end{equation}%
subject to two constraints%
\begin{eqnarray}
g^{1}(w_{n},w_{1},\ldots ,w_{m}) &=&w_{n}-\sum_{i=1}^{m}w_{i}-1=0,
\label{cons_1} \\
g^{2}\left( w_{n},w_{1},...,w_{m}\right) &=&\frac{N_{n}^{2}}{n^{2}}%
w_{n}-m^{2}\frac{N_{l}^{2}}{l^{2}}\sum_{i=1}^{m}w_{i}=0.  \label{cons_2}
\end{eqnarray}%
We refer to the optimal estimator so produced as $\widehat{d}_{J,m}^{Opt}$
hereinafter.

Constraints\ (\ref{cons_1}) and (\ref{cons_2}) ensure that Property $\left(
P.1\right) $ holds for the resultant estimator. Specifically, (\ref{cons_1})
ensures that $\widehat{d}_{J,m}^{Opt}$ is asymptotically unbiased for $d_{0}$%
, as can be seen by inspection of (\ref{Expectation of dhat_J}). The\
dominant bias term of $\widehat{d}_{J,m}^{Opt}$ will be eliminated if and
only if the second component\ appearing in (\ref{Expectation of dhat_J}) is
set to zero; that is, if and only if%
\begin{equation}
\frac{2\pi ^{2}}{9}\frac{f_{YY}^{\ast \prime \prime }\left( 0\right) }{%
f_{YY}^{\ast }\left( 0\right) }\frac{N_{n}^{2}}{n^{2}}w_{n}-\frac{2\pi ^{2}}{%
9}\frac{f_{Y_{i}Y_{i}}^{\ast \prime \prime }\left( 0\right) }{%
f_{Y_{i}Y_{i}}^{\ast }\left( 0\right) }\frac{N_{l}^{2}}{l^{2}}%
\sum_{i=1}^{m}w_{i}=0.  \label{unbiased_cond}
\end{equation}%
Using Point $\left( iii\right) $ of Table \ref{Table:Notations}, we have
that $f_{Y_{i}Y_{i}}^{\ast }\left( 0\right) =f_{YY}^{\ast }\left( 0\right) $
and $f_{Y_{i}Y_{i}}^{\ast \prime \prime }\left( 0\right) =m^{2}f_{YY}^{\ast
\prime \prime }\left( 0\right) $. Hence, the condition in (\ref%
{unbiased_cond}) collapses to constraint (\ref{cons_2}). Given (\ref{opt}),
Property $\left( P.2\right) $ is satisfied by construction.

Henceforth writing, $Cov%
\big(%
\widehat{d}_{n},\widehat{d}_{i}%
\big)%
=c_{n,i}^{\ast }$ and $Cov%
\big(%
\widehat{d}_{i},\widehat{d}_{i^{\prime }}%
\big)%
=c_{i,j}^{\dagger },$ such that $c_{i,j}^{\dagger }=c_{j,i}^{\dagger },$ the
Lagrangian function\ is given by,%
\begin{eqnarray}
\tilde{L}\left( w_{n},w_{1},\ldots ,w_{m},\delta _{1},\delta _{2}\right) &=&%
\frac{\pi ^{2}}{24N_{n}}w_{n}^{2}+\frac{\pi ^{2}}{24N_{l}}%
\sum_{i=1}^{m}w_{i}^{2}+2\sum_{i=1}^{m-1}\sum_{j=i+1}^{m}w_{i}w_{j}c_{i,j}^{%
\dagger }  \notag \\
&&-2w_{n}\sum_{i=1}^{m}w_{i}c_{n,i}^{\ast }+\delta _{1}%
\big(%
w_{n}-\sum_{i=1}^{m}w_{i}-1%
\big)
\notag \\
&&+\delta _{2}%
\Big(%
\frac{N_{n}^{2}}{n^{2}}w_{n}-m^{2}\frac{N_{l}^{2}}{l^{2}}\sum_{i=1}^{m}w_{i}%
\Big)%
.  \label{L}
\end{eqnarray}%
The first-order conditions (FOCs) are thus given by,%
\begin{eqnarray*}
\frac{\partial \tilde{L}}{\partial \delta _{1}} &=&0\Rightarrow
w_{n}-\sum_{i=1}^{m}w_{i}=1, \\
\frac{\partial \tilde{L}}{\partial \delta _{2}} &=&0\Rightarrow \frac{%
N_{n}^{2}}{n^{2}}w_{n}-m^{2}\frac{N_{l}^{2}}{l^{2}}\sum_{i=1}^{m}w_{i}=0, \\
\frac{\partial \tilde{L}}{\partial w_{n}} &=&0\Rightarrow \frac{2\pi ^{2}}{%
24N_{n}}w_{n}-2\sum_{i=1}^{m}w_{i}c_{n,i}^{\ast }+\delta _{1}+\frac{N_{n}^{2}%
}{n^{2}}\delta _{2}=0, \\
\frac{\partial \tilde{L}}{\partial w_{i,m}} &=&0\Rightarrow
-2w_{n}c_{n,i}^{\ast }+\frac{2\pi ^{2}}{24N_{l}}w_{i}+2\sum_{j=1,j\neq
i}^{m}w_{j}c_{i,j}^{\dagger }-\delta _{1}-m^{2}\frac{N_{l}^{2}}{l^{2}}\delta
_{2}=0;\text{ }i=1,\ldots ,m.
\end{eqnarray*}%
Defining 
\begin{equation}
\mathbf{A}=\left[ 
\begin{array}{cccccc}
1 & -1 & \ldots & -1 & 0 & 0 \\ 
\frac{N_{n}^{2}}{n^{2}} & -m^{2}\frac{N_{l}^{2}}{l^{2}} & \ldots & -m^{2}%
\frac{N_{l}^{2}}{l^{2}} & 0 & 0 \\ 
\frac{\pi ^{2}}{12N_{n}} & -2c_{n,1}^{\ast } & \ldots & -2c_{n,m}^{\ast } & 1
& \frac{N_{n}^{2}}{n^{2}} \\ 
-2c_{n,1}^{\ast } & \frac{\pi ^{2}}{12N_{l}} & \ldots & 2c_{1,m}^{\dagger }
& -1 & -m^{2}\frac{N_{l}^{2}}{l^{2}} \\ 
\vdots & \vdots & \ddots & \vdots & \vdots & \vdots \\ 
-2c_{n,m}^{\ast } & 2c_{1,m}^{\dagger } & \ldots & \frac{\pi ^{2}}{12N_{l}}
& -1 & -m^{2}\frac{N_{l}^{2}}{l^{2}}%
\end{array}%
\right] ,\text{ }\mathbf{w}=\left[ 
\begin{array}{c}
w_{n} \\ 
w_{1} \\ 
\vdots \\ 
w_{m} \\ 
\delta _{1} \\ 
\delta _{2}%
\end{array}%
\right] \text{ and }\mathbf{b}=\left[ 
\begin{array}{c}
1 \\ 
0 \\ 
0 \\ 
\vdots \\ 
0 \\ 
0%
\end{array}%
\right] ,  \label{matrix form}
\end{equation}%
the optimal solution, $\mathbf{w}^{\ast }=\left[ 
\begin{array}{cccccc}
w_{n}^{\ast } & w_{1}^{\ast } & \ldots & w_{m}^{\ast } & \delta _{1}^{\ast }
& \delta _{2}^{\ast }%
\end{array}%
\right] ^{\top },$ is given by 
\begin{equation}
\mathbf{w}^{\ast }=\mathbf{A}^{-1}\mathbf{b.}  \label{w_star}
\end{equation}%
Given the structure of $\mathbf{b}$ this means that the solutions for the
weights are given by the elements of the first column of $\mathbf{A}^{-1}$,
and the optimal jackknife estimator is accordingly given as:%
\begin{equation}
\widehat{d}_{J,m}^{Opt}=w_{n}^{\ast }\widehat{d}_{n}-\sum_{i=1}^{m}w_{i}^{%
\ast }\widehat{d}_{i},  \label{Optimal jackknife}
\end{equation}%
where $w_{n}^{\ast }=\left[ 1-\left( \left. N_{n}l\right/ \left(
N_{l}mn\right) \right) ^{2}\right] ^{-1}$, given immediately by solving the
first two FOCs.

To complete the result we need to show that (\ref{w_star}) is a local
minimizer of $\tilde{L}\left( \cdot \right) .$ To do so, we need to show
that: $\left( i\right) $ the constraint qualification -- that the rank of
the matrix formed by the first-order derivatives at the solution of the
constraints with respect to parameters, except the Lagrangian parameters, is
equal to the number of conditions -- is met, $\left( ii\right) $ the
solution of the Lagrangian function satisfies the FOCs, and, $\left(
iii\right) $ the leading principal minors of the bordered Hessian matrix, $%
\mathbf{H}_{\left( m+3\right) \times \left( m+3\right) }^{B},$ all take the
same sign of $\left( -1\right) ^{k},$ where $k$ is the number of constraints
(see, Chapter 12 of\ \color{blue}\citealp*{ChiangWainwright2005}\color{black}, for more details).

In our problem, the number of constraints equals $2$ and 
\begin{equation*}
Rank\left[ 
\begin{array}{cc}
\frac{\partial g^{1}}{\partial w_{n}} & \frac{\partial g^{2}}{\partial w_{n}}
\\ 
\frac{\partial g^{1}}{\partial w_{1}} & \frac{\partial g^{2}}{\partial w_{1}}
\\ 
\vdots & \vdots \\ 
\frac{\partial g^{1}}{\partial w_{m}} & \frac{\partial g^{2}}{\partial w_{m}}%
\end{array}%
\right] =Rank\left[ 
\begin{array}{cc}
1 & 1 \\ 
\frac{N_{n}^{2}}{n^{2}} & m^{2}\frac{N_{l}^{2}}{l^{2}} \\ 
\vdots & \vdots \\ 
\frac{N_{n}^{2}}{n^{2}} & m^{2}\frac{N_{l}^{2}}{l^{2}}%
\end{array}%
\right] =2.
\end{equation*}%
Hence, the rank condition is met. The second condition is met by default.
The important condition is the third one, where we need to show that the
leading principal minors of $\mathbf{H}_{\left( m+3\right) \times \left(
m+3\right) }^{B},$ exceed zero for every $m=2,3,\ldots .$ The bordered
Hessian matrix for our case is given by 
\begin{equation*}
\mathbf{H}_{\left( m+3\right) \times \left( m+3\right) }^{B}=\left[ 
\begin{array}{cccccc}
0 & 0 & 1 & -1 & \cdots & -1 \\ 
0 & 0 & \frac{N_{n}^{2}}{n^{2}} & -m^{2}\frac{N_{l}^{2}}{l^{2}} & \ldots & 
-m^{2}\frac{N_{l}^{2}}{l^{2}} \\ 
1 & \frac{N_{n}^{2}}{n^{2}} & \frac{\pi ^{2}}{12N_{n}} & -2c_{n,1}^{\ast } & 
\ldots & -2c_{n,m}^{\ast } \\ 
-1 & -m^{2}\frac{N_{l}^{2}}{l^{2}} & -2c_{n,1}^{\ast } & \frac{\pi ^{2}}{%
12N_{l}} & \ldots & 2c_{1,m}^{\dagger } \\ 
\vdots & \vdots & \vdots & \vdots & \ddots & \vdots \\ 
-1 & -m^{2}\frac{N_{l}^{2}}{l^{2}} & -2c_{n,m}^{\ast } & 2c_{1,m}^{\dagger }
& \ldots & \frac{\pi ^{2}}{12N_{l}}%
\end{array}%
\right] .
\end{equation*}%
The proof of positivity of the principal minors of the above matrix is given
in Appendix B. Hence, the solution in (\ref{w_star}) is a local minimizer of 
$\tilde{L}\left( \cdot \right) $.

We complete this section with three remarks:
\begin{remark}
If we consider only bias reduction to the order $\left. N_{n}^{2}\right/
n^{2}$, without concurrent variance reduction; that is, we produce an
estimator that satisfies only (P.1), and not (P.2), then the formulae for the%
\textbf{\ }weights are%
\begin{equation}
w_{n}^{\ast }=
\Big[%
1-%
\Big(%
\frac{N_{n}}{N_{l}}\frac{l}{nm}%
\Big)%
^{2}%
\Big]%
^{-1}\text{ and }w_{i}^{\ast }=\frac{1}{m}\left( w_{n}^{\ast }-1\right) ,%
\text{ for }i=1,\ldots ,m.  \label{weights by Chambers}
\end{equation}
These weights mimic those of\ \color{blue}\cite{Chambers2013}\color{black}\
in the short memory setting (under a non-overlapping sub-sampling scheme),
in which variance minimization was not a consideration.
\end{remark}
\begin{remark}
When \color{blue}\cite{Chambers2013}\color{black}\ considers the
moving-block sub-sampling scheme (again, in the short memory setting), he
chooses the sub-sample length to be $l=n-m+1$. In this case, when $n$\ is
large and $m$\ is small, the sub-sample length is $l\approx n$, and the
impact of bias correction is reduced as a consequence; something that is in
evidence in the Monte Carlo simulation results reported by that author. As a
result of this observation, in our investigations we use the common
sub-sample length of $l=n/m$, under both the non-overlapping and
moving-block schemes.
\end{remark}
\begin{remark}
Condition $3.3$ of \color{blue}\cite{GuggenbergerSun2006}\color{black}\ has
a\ similar purpose to our (\ref{cons_2}). The difference is that we
eliminate the $O%
\Big(%
\left. N_{n}^{2}\right/ n^{2}%
\Big)$ term from the bias of the LPR estimator, whereas they eliminate bias up to
an order of $\left. N_{n}^{2r}\right/ n^{2r},$ for some $r\geq 1.$ The role
played by (\ref{opt}) is somewhat different from that played by Condition $3.4$ of \color{blue}\cite{GuggenbergerSun2006}\color{black}. The latter
condition is imposed mainly to link the bias and variance of $\widehat{d}%
_{r}^{GS}$ to that of $\widehat{d}_{r}^{AG}$, for any given $r$; this link
occurring via the introduction of the tuning parameter, $\delta $ (see (\ref%
{dhat_GS}) above), on which the finite sample performance of their estimator
depends. In our method, (\ref{opt}) is used to control the increase in
variance that occurs due to the reduction in bias, with the optimal weights
determined by (\ref{opt})-(\ref{cons_2}) not depending on any arbitrary
quantities.
\end{remark}

\section{Asymptotic results\label{Sec3:Properties}}

The asymptotic properties of the optimal jackknife estimator\ depend on the
optimal weights which, in turn, are functions of the covariance terms
between the log-periodograms associated with the full sample and the
sub-samples, as seen in (\ref{Covariance(dhat,dhat_i)}) and (\ref%
{Covariance(dhat_i,dhat_j)}). Provided that the DGP satisfies assumptions $\left( A.1\right) -\left( A.3\right) $, \color{blue}\cite{Lahiri2003}%
\color{black}\ has shown that periodogram ordinates are asymptotically
independent when the frequencies are at a sufficient distance apart,
provided that the set of observations remain the same. However, in our case,
we are dealing with periodograms calculated both for the full set of
observations, and for subsets of the full set. Thus, two questions that
arise here are: $\left( i\right) $ Are the periodograms of the full sample
and the sub-samples at different frequency ordinates asymptotically
independent$?$ and, $\left( ii\right) $ When $d\neq 0$, do the periodograms
still converge to a chi-square distribution as they do when $d=0\ $(see
Theorem $5.2.6$ of\ \color{blue}\citealp*{Brillinger1981}%
\color{black})$?$ We address both questions in Section \ref{period_props}
and provide formulae for calculating the relevant covariance terms
algebraically, adopting the procedure used in \color{blue}\cite{Brillinger1981}\color{black}. In Section \ref{asymptotic} we then use these
results to derive the asymptotic properties of the optimal jackknife
estimator.

\subsection{Stochastic properties of periodograms in the full sample and in
sub-samples \label{period_props}}

We begin by defining $\left\{ X_{1},X_{2},\ldots ,X_{h}\right\} $ as an
arbitrary set of $h$ stationary time series. We link these series to the
full sample and the $m$ sub-samples of observations below. Our use of
notation in this section mimics, in large part, that of \color{blue}\citet[\S. 2.6]{Brillinger1981}\color{black}.

\begin{definition}
\label{Def:cum_series}Suppose $\left\{ X_{1},X_{2},\ldots ,X_{h}\right\} $
is a set of $h$ stationary time series. The $k^{th}$-order cumulant $\kappa
_{X_{a_{1}},\ldots ,X_{a_{k}}}\left( u_{1},...,u_{k-1}\right) ,$ for $%
k=1,2,\ldots ,h$, \ and $u_{j}=0,\pm 1,\pm 2...$ for $j=1,2,...,k-1$, is
defined as follows,%
\begin{equation}
\kappa _{X_{a_{1}},\ldots ,X_{a_{k}}}\left( u_{1},...,u_{k-1}\right)
=\int_{-\pi }^{\pi }\ldots \int_{-\pi }^{\pi }\exp 
\Big(%
-\imath \sum\limits_{j=1}^{k-1}\lambda _{j}u_{j}%
\Big)%
f_{X_{a_{1}},\ldots ,X_{a_{k}}}\left( \lambda _{1},\ldots ,\lambda
_{k-1}\right) d\lambda _{1}\ldots d\lambda _{k-1},
\label{kth order cumulant}
\end{equation}%
where $f_{X_{a_{1}},\ldots ,X_{a_{k}}}\left( \lambda _{1},\ldots ,\lambda
_{k-1}\right) $ is the $k^{th}$-order joint spectral density of $\left\{
X_{a_{1}},\ldots ,X_{a_{k}}\right\} $, for $-\pi <\lambda _{j}<\pi $, $%
j=1,2,...,k-1,$ with $a_{1},\ldots ,a_{k}=1,2,\ldots ,h,$ and $k=1,2,\ldots $%
.\newline
For $\sum_{u_{1}=-\infty }^{\infty }\cdots \sum_{u_{k-1}=-\infty }^{\infty
}\left\vert \kappa _{X_{a_{1}},\ldots ,X_{a_{k}}}\left(
u_{1},...,u_{k-1}\right) \right\vert <\infty ,$ then the inverse form of (%
\ref{kth order cumulant}) is given by,%
\begin{equation}
f_{X_{a_{1}},\ldots ,X_{a_{k}}}\left( \lambda _{1},\ldots ,\lambda
_{k-1}\right) =\left( 2\pi \right) ^{-k+1}\sum_{u_{1}=-\infty }^{\infty
}\cdots \sum_{u_{k-1}=-\infty }^{\infty }\kappa _{X_{a_{1}},\ldots
,X_{a_{k}}}\left( u_{1},...,u_{k-1}\right) \exp 
\Big(%
-\imath \sum\limits_{j=1}^{k-1}\lambda _{j}u_{j}%
\Big)%
.  \label{general relatioship of kth order spectrum}
\end{equation}
\end{definition}

Now let $X_{1}=\mathbf{y}$ denote the full sample of $n$ observations on the
random variable following the model in (\ref{Spectral density_ARFIMA});
whilst $X_{1+i}=\mathbf{y}_{i}$ denotes the vector of observations for the
sub-sample $i=1,2,\ldots ,m,$ with length $l$. Set $h=m+1$ in Definition \ref%
{Def:cum_series}. Let $D_{X_{1}}^{\left( n\right) }\left( .\right) $ and $%
D_{X_{1+i}}^{\left( l\right) }\left( .\right) $ respectively be the DFT of
the full sample and $i^{th}$ sub-sample at some frequency. Set%
\begin{equation}
L_{i}=\left\{ 
\begin{array}{cl}
n & if\text{ }i=1 \\ 
l & otherwise%
\end{array}%
\right. .  \label{Sample_length}
\end{equation}

In Proposition \ref{Proposition1} we give the expression for the $k^{th}$%
-order joint cumulant of the DFTs of the $h=m+1$ series associated with the
full sample and the $m$ sub-samples.

\begin{proposition}
\label{Proposition1}Suppose Assumptions $\left( A.1\right) -\left(
A.3\right) $ hold. The $k^{th}$-order cumulant of $%
\big\{%
D_{X_{a_{1}}}^{\left( L_{1}\right) }\left( \lambda _{1}\right) ,$ $%
D_{X_{a_{2}}}^{\left( L_{2}\right) }\left( \lambda _{2}\right) ,$ $...,$ $%
D_{X_{a_{k}}}^{\left( L_{k}\right) }\left( \lambda _{k}\right) 
\big\}%
$, for $k=1,2,\ldots $, is given by,%
\begin{equation}
\kappa _{D_{X_{a_{1}}},\ldots ,D_{X_{a_{k}}}}\left( \lambda _{1},...,\lambda
_{k-1}\right) =L^{-\frac{k}{2}}\left( 2\pi \right) ^{\frac{k}{2}-1}\Delta
^{\left( L\right) }%
\Big(%
\sum_{j=1}^{k}\lambda _{j}%
\Big)%
f_{X_{a_{1}},\ldots ,X_{a_{k}}}\left( \lambda _{1},...,\lambda _{k-1}\right)
+o%
\Big(%
L^{1-2d-\frac{k}{2}}%
\Big)%
,  \label{cumulant of our interest}
\end{equation}%
where, $L=\min \left\{ L_{1},\ldots ,L_{k}\right\} $.\footnote{%
The $k^{th}$-order cumulant associated with the DFTs should, for
completeness, be denoted by $\kappa _{D_{X_{a_{1}}}^{{\small (L}_{{\small 1}}%
{\small )}},\ldots ,D_{X_{a_{k}}}^{{\small (L}_{{\small k}}{\small )}%
}}\left( .,\ldots ,.\right) $. For notational ease, however, we express the cumulant without making explicit the relevant sample sizes.}
\end{proposition}

From Proposition \ref{Proposition1} we can derive the relationship between
the DFTs corresponding to full sample and the $m$\ sub-samples as the\
sample size increases. The result is given in the following theorem:

\begin{theorem}
\label{AP_DFT}Suppose Assumptions $\left( A.1\right) -\left( A.4\right) $
hold, and suppose $\lambda =2\pi r/L_{i}$ and $\omega =2\pi s/L_{j}$ for
integers $r$ and $s$. Then for a fixed value of $L_{i}$ and $L_{j}$, $%
D_{X_{a_{i}}}^{\left( L_{i}\right) }\left( \lambda \right) $ and $%
D_{X_{a_{j}}}^{\left( L_{j}\right) }\left( \mu \right) $ are asymptotically
independent, whenever $\max \left\{ L_{i}\lambda ,L_{j}\mu \right\}
\rightarrow \infty $, for $i\neq j$.
\end{theorem}

Theorem \ref{AP_DFT} immediately implies the asymptotic independence of the
periodograms of the full sample and all sub-samples. However, in finite
samples, the dependence structure across these periodograms may play an
important role in determining the variance of the jackknife estimator in (%
\ref{Variance of dhat_J}), through the form of the covariances in (\ref%
{Covariance(dhat,dhat_i)}) and (\ref{Covariance(dhat_i,dhat_j)}).
Expressions for the covariances between the periodograms corresponding to\
the full sample and the sub-samples are provided in the following theorem,
from which further insights on this point can be gleaned.

\begin{theorem}
\label{AP_I}Let $I_{X_{a_{i}}}^{\left( L_{i}\right) }\left( \lambda \right) $
and $I_{X_{a_{j}}}^{\left( L_{j}\right) }\left( \lambda \right) $ be the
periodograms associated with DFTs $D_{X_{a_{i}}}^{\left( L_{i}\right)
}\left( \lambda \right) $ and $D_{X_{a_{j}}}^{\left( L_{j}\right) }\left(
\mu \right) $\ respectively. Suppose Assumptions $\left( A.1\right) -\left(
A.3\right) $ hold. Then,%
\begin{eqnarray}
Cov%
\big(%
I_{X_{a_{i}}}^{\left( L_{i}\right) }\left( \lambda \right)
,I_{X_{a_{j}}}^{\left( L_{j}\right) }\left( \mu \right) 
\big)
&=&\frac{2\pi }{L}f_{X_{a_{i}},X_{a_{i}},X_{a_{j}},X_{a_{j}}}\left( \lambda
,-\lambda ,\mu \right) +\frac{2\pi }{L}\left[ \eta \left( \lambda -\mu
\right) +\eta \left( \lambda +\mu \right) \right] \left\{
f_{X_{a_{i}}X_{a_{j}}}\left( \lambda \right) \right\} ^{2}  \notag \\
&&+2\pi \left[ \eta \left( \lambda -\mu \right) +\eta \left( \lambda +\mu
\right) \right] f_{X_{a_{i}}X_{a_{j}}}\left( \lambda \right) o%
\big(%
L^{^{-2d}}%
\big)%
+o%
\big(%
L^{^{-1-2d}}%
\big)%
,  \label{Covariance between samples}
\end{eqnarray}%
where $\eta \left( \omega \right) =\lim_{T\rightarrow \infty }\dfrac{1}{2\pi 
}\sum_{t=-T}^{T}\exp \left\{ -\imath \omega t\right\} $, and $L$ is as
defined in Proposition \ref{Proposition1}. When Assumption $\left(
A.4\right) $ also holds, the periodogram ordinates $I_{X_{a_{i}}}^{\left(
L_{i}\right) }\left( \mu \right) $ and $I_{X_{a_{j}}}^{\left( L_{j}\right)
}\left( \omega \right) $ with $i\neq j$, are asymptotically $%
f_{X_{1}X_{1}}\left( \cdot \right) \left. \chi _{\left( 2\right)
}^{2}\right/ 2$ random variables.
\end{theorem}

Theorem \ref{AP_I} is a generalization of the result of Theorem $5.2.6$ of \color{blue}\cite{Brillinger1981}\color{black}\ to the context of
jackknifing. Equation (\ref{Covariance between samples}) provides the first
few dominant terms of the covariance between the periodograms associated
with the full sample and a particular sub-sample, or between distinct
sub-samples, at various frequency ordinates. Further, (\ref{Covariance
between samples}) reflects the fact that, for finite $n$, the relevant
periodograms are positively correlated. This result is to be anticipated
given that the sub-samples are subsets of the full sample and, hence, retain
the same dependence structure as the full sample. Furthermore, the theorem
states that the periodogram ordinates (for either the full sample and a
given sub-sample, or between sub-samples) have a limiting joint distribution
of the form, $f_{X_{1}X_{1}}(\lambda )\left. \chi _{\left( 2\right)
}^{2}\right/ 2$, where $f_{X_{1}X_{1}}(.)$ is the spectral density of the
time series from which the full sample is generated.%

Using the covariance terms and the distribution of the periodograms provided
in the above theorem, we can find the joint distribution of the
log-periodograms associated with the full sample and any sub-sample (or for
two distinct sub-samples). Using the joint distribution of the
log-periodograms, we can derive the moment generating function of the joint
distribution. This leads to the derivation of the covariance terms for the
log-periodogram. This result is provided in\ Appendix B. The covariances
between log-periodograms allow us to obtain the covariances between the\
full-sample and sub-sample LPR estimators given in (\ref%
{Covariance(dhat,dhat_i)}) and (\ref{Covariance(dhat_i,dhat_j)}). Exploiting
the relationship between the different LPR estimators, we then establish the
consistency and asymptotic normality of the optimal jackknife estimator in
the following section.

\subsection{Asymptotic properties of the optimal jackknife estimator\label%
{asymptotic}}

Using the results established in the previous section, we state the
relationship between the full-sample and sub-sample LPR estimators in
Theorem \ref{AI_LPR}. The asymptotic properties of the optimal jackknife
estimator are then established in Theorem \ref{AP_LPR}.
\begin{theorem}
\label{AI_LPR}Let $\widehat{d}_{n}$ and $\widehat{d}_{i}$ be the LPR
estimators for the full sample and the $i^{th}$ sub-sample with sub-sample
length, $l.$ Suppose Assumptions $\left( A.1\right) -\left( A.4\right) $
hold. Then,\ for a fixed value of $m,$

\begin{enumerate}
\item[$\left( i\right) $] $\widehat{d}_{n}$ and $\widehat{d}_{i}$ are
asymptotically independent.

\item[$\left( ii\right) $] $\widehat{d}_{i}$ and $\widehat{d}_{j}$ for $%
i\neq j,$ $i,j=1,\ldots ,m,$ are asymptotically independent.
\end{enumerate}
\end{theorem}

From Theorem \ref{LPR prop}, the LPR estimator constructed from the full
sample is consistent and satisfies (\ref{Distribution of LPR}). Similarly,
allowing the number of sub-samples, $m$, to be fixed (hence $l$ changes as $%
n $ changes such that $n=m\times l$), as $l\rightarrow \infty ,$ $\widehat{d}%
_{i}\rightarrow ^{P}d_{0}$, and $\sqrt{N_{l}}%
\big(%
\widehat{d}_{i}-d_{0}%
\big)%
\rightarrow ^{D}N%
\Big(%
0,\tfrac{\pi ^{2}}{24}%
\Big)%
$. This implies the sub-sample LPR estimators have the same limiting
distribution as the full-sample estimator. The asymptotic properties of $%
\widehat{d}_{J,m}^{Opt}$ are given in the following theorem.

\begin{theorem}
\label{AP_LPR}Under the same assumptions and conditions given in Theorem \ref%
{LPR prop}, for a fixed value of $m$,%
\begin{equation*}
\widehat{d}_{J,m}^{Opt}\rightarrow ^{P}d_{0},\text{ and }\sqrt{N_{n}}%
\big(%
\widehat{d}_{J,m}^{Opt}-d_{0}%
\big)%
\rightarrow ^{D}N%
\Big(%
0,\tfrac{\pi ^{2}}{24}%
\Big)%
\text{ as }n\rightarrow \infty
\end{equation*}%
where $d_{0}$ is the true value of $d$ and $\widehat{d}_{J,m}^{Opt}$ is as
given in (\ref{Optimal jackknife}).
\end{theorem}

Thus, it follows from Theorem \ref{AP_LPR} that $\widehat{d}_{J,m}^{Opt}$\
is consistent for $d_{0}$ and achieves a limiting normal distribution with
the\ same variance as the base LPR estimator itself. Further, the rate of
convergence of the optimal jackknife estimator, $\sqrt{N_{n}},$ is also the
same as that of the LPR\ estimator. That is, there is no loss of asymptotic
efficiency compared to $\widehat{d}_{n}$. Importantly, these asymptotic
properties of the jackknife estimator do not depend on the number of
sub-samples or the sub-sample length, as long as the former is fixed and the
latter increases with $n\ $such that that $n=m\times l$.

\section{Simulation exercise\label{Sec4:Simulation}}

In this section, Monte Carlo simulation is used to compare the finite
sample\ performance of the proposed jackknife estimator with: $\left(
i\right) $\ the weighted-average estimator of \color{blue}\cite{GuggenbergerSun2006}\color{black}, $\widehat{d}_{r}^{GS},$ with $r=1$, $%
\left( ii\right) $ the bias-corrected pre-filtered sieve bootstrap-based estimator of \color{blue}\cite{PoskittMartinGrose2016}\color{black}, $%
\widehat{d}^{PFSB},$ $\left( iii\right) $ the unadjusted LPR estimator, $%
\widehat{d}_{n}$, $\left( iv\right) $\textbf{\ }the\ MLE,\ $\widehat{d}%
^{MLE} $, and,\textbf{\ }$\left( v\right) $\textbf{\ }the pre-whitened (PW)
estimator,\textbf{\ }$\widehat{d}^{PW}$.\textbf{\ }Performance is assessed
in terms of bias and RMSE,\ and under a variety of true DGPs. In Section \ref%
{mc} details of the basic Monte Carlo design are provided. Results under
correct and incorrect specification of the model are then documented in
Section \ref{FSP:CS} and Section \ref{FSP:MS} respectively, with further
computational details that pertain to those specific settings provided
therein. In order to assist the reader, we tabulate a ranking of the
different estimators, under the variety of settings considered, in Section%
\textbf{\ }\ref{summ}. All numerical results are produced using MATLAB $%
2015b $, version $8.6.0.267246$, and all tables of results are collected in
Appendix C.

\subsection{Monte Carlo \textbf{design\label{mc}}}

Data are generated from various versions of a Gaussian fractional process,
ARFIMA($p_{0},d_{0},q_{0}$), where $p_{0}$\ is the lag length of the true
autoregressive (AR) component and $q_{0}$\ the lag length of the true moving
average (MA) component. The lag lengths $p_{0}$ and $q_{0}$ equal either one
or zero in all settings. For\textbf{\ }$p_{0}=q_{0}=1$, the process is given
by\textbf{\ }%
\begin{equation}
\left( 1+\phi _{0}B\right) \left( 1-B\right) ^{d_{0}}(Y_{t}-\mu _{0})=\left(
1+\theta _{0}B\right) \varepsilon _{t},  \label{ARFIMA(1,d,0)}
\end{equation}%
where $B$ is the backward shift operator, $B^{k}x_{t}=x_{t-k},$ for $%
k=1,2,\ldots ,$ and $\varepsilon _{t}\sim i.i.d$ $N\left( 0,1\right) $.
Here, $\mu _{0}$\ is the mean parameter for $Y_{t}$, and\textbf{\ }without
loss of generality we assume that $\mu _{0}=0.$ For the parameter of
interest, $d,$ we select true values from the set, $d_{0}=$ $\left\{
-0.25,0,0.25,0.45\right\} $. Values from the set $\left\{
-0.9,-0.4,0.4,0.9\right\} $ are adopted for both $\phi _{0}$ and $\theta
_{0} $.\footnote{%
Additional results based on the assumption that $\varepsilon _{t}$\ is
distributed as Student $t$\ with $5$\ degrees of freedom are available from
the authors on request. This additional design feature did not lead to
qualitatively different results.} Additional details are provided in
Sections \ref{FSP:CS} and \ref{FSP:MS}.

Sample sizes $n\in \left\{ 96,576\right\} $\ are considered. These values
are chosen to reflect the size of samples used in real world examples (see,
for example, \color{blue}\citealp*{DieboldHustedRush1991}\color{black},\ \color{blue}\citealp*{DelgadoRobinson1994}\color{black},\ \color{blue}\citealp*{Gil-Alana1997}\color{black},\ and\ \color{blue}\citealp*{ReisenLopes1999}\color{black}). However, one should note that, in general, the size of data
sets from finance, in particular those recorded at high frequency (for example, \color{blue}\citealp*{GrangerHyung2004}\color{black}), or from biology (for example, the tree-ring data set of \color{blue}\citealp*{ContrerasPalma2013}%
\color{black}), or in certain other of the examples mentioned in the
Introduction, are much larger than the sample sizes considered here. On the
other hand, these sample sizes are large enough to enable a range of values
for the number of sub-samples, $m$, to be explored, with the chosen range of 
$m$ being $\left\{ 2,3,4,6,8\right\} $. We also consider only sub-samples
that have equal length, $l=n/m$, under both sub-sampling approaches, with
the slightly unorthodox values of $n\in \left\{ 96,576\right\} $ chosen in
order to ensure that $l$ is an integer.

We adopt the following procedure in implementing the \textit{optimal}
jackknife bias-adjustment technique:

\begin{enumerate}
\item[\textbf{Step 1:}] Generate the\ full sample of size $n,$ $\mathbf{y}$,
from the relevant stationary ARFIMA($p_{0},d_{0},q_{0}$) model.

\item[\textbf{Step 2:}] Compute the LPR estimator of $d_{0},$ $\widehat{d}%
_{n}$\ using (\ref{LPR}).

\item[\textbf{Step 3:}] Draw the sub-samples, $\mathbf{y}_{i}$ $\left(
i=1,2,...,m\right) ,$ from the full sample based on the relevant
sub-sampling technique (non-overlapping or moving-block) and compute the LPR
estimator of $d_{0},$ $\widehat{d}_{i},$ for each sub-sample.

\item[\textbf{Step 4:}] Depending on the sub-sample selection method chosen
in Step $3$, obtain the optimal\ weights for the corresponding method\textbf{%
\ }based on the parameters of\textbf{\ }the\ (true) ARFIMA($%
p_{0},d_{0},q_{0} $) model and compute the optimal jackknife estimator, $%
\widehat{d}_{J,m}^{Opt}.$ In the empirically realistic case in which the
true model parameters are unknown, a feasible version of the jackknife
estimator is implemented, with all details provided in the relevant sections
below.

\item[\textbf{Step 5:}] Repeat Steps $1-4$ $100,000$ times and compute
estimates of the bias and RMSE of the optimal jackknife estimator.
\end{enumerate}

In Steps $2$ and $3$, the number of frequencies used to calculate the
relevant LPR estimator\ is set to $N_{L}=\left\lfloor L^{\alpha
}\right\rfloor $, with $\alpha =0.65$, where $L$ is as defined in (\ref%
{Sample_length}).\footnote{%
Certain simulation results based on $\alpha =0.5$ have also been produced,
but are not presented here due to space considerations. These additional
numerical results can be provided by the authors on request.} The optimal
jackknife estimators calculated using the\ non-overlapping (abbreviated to
Opt-NO), and moving-block (abbreviated to Opt-MB) schemes, are denoted by $%
\widehat{d}_{J,m}^{Opt-NO}$ and $\widehat{d}_{J,m}^{Opt-MB},$ respectively.

The weighted-average estimator of \color{blue}\cite{GuggenbergerSun2006}%
\color{black}\ is computed as described in Section \ref{Sec2:WLPR}, with the
following additional details. For a given $N_{n}$, the set of bandwidths
used to calculate the constituent estimators in (\ref{dhat_GS}) are $%
N_{n,i}=\left\lfloor q_{i}N_{n}\right\rfloor $, where $\mathbf{q}^{\top
}=\left( q_{1},q_{2},\ldots ,q_{K}\right) =\left( 1,1.05,\ldots ,2\right) $.
We produce the GS estimator (based on\textbf{\ }$r=1$) using two different
choices of $N_{n}$: $\left( i\right) $ $N_{n}=\left\lfloor n^{\alpha
}\right\rfloor $, with $\alpha =0.65$ (denoting this estimator by $\widehat{d%
}_{1}^{GS}$), and $\left( ii\right) $ the optimal choice of $N_{n}$ as
suggested in \color{blue}\citet[page 876]{GuggenbergerSun2006}\color{black}\ (denoting this version by $\widehat{d}_{1}^{Opt-GS}$).
Importantly, bandwidth choice $\left( ii\right) $ depends on knowledge of
the true values of the short memory parameters, whereas bandwidth choice $%
\left( i\right) $ yields an estimator that is feasible empirically. The
parameter $\delta ,$ required for both versions of the GS estimator, is
evaluated using the formula $\delta =\tau _{r}/(\tau _{r}^{\ast
}\sum_{i=1}^{K}q_{k}^{2+2r}),$ where $\tau _{r-1}^{\ast }=-\left. \left(
2\pi \right) ^{2r}r\right/ [\left( 2r\right) !\left( 2r+1\right) ^{2}]$ and
the number $\tau _{r}$ is as defined in \color{blue}\cite{AndrewsGuggenberger2003}\color{black}. Details regarding the construction
of the pre-filtered sieve bootstrap estimator ($\widehat{d}^{PFSB}$) can be found in \color{blue}\cite{PoskittMartinGrose2016}\color{black}. In
implementing this method, we set the number of bootstrap samples to $B=1000$.

The MLE, $\widehat{d}^{MLE}$, is obtained by concentrating the Gaussian
log-likelihood associated with an assumed ARFIMA($p,d,q$) model with respect
to $\mu $ and $\sigma ^{2}$, subtracting from that the resulting constant $%
(n\log n-n)/2$, and maximizing the profile log-likelihood function 
\begin{equation}
L(\mathbf{\eta })=-\frac{n}{2}\log \left[ \left( \mathbf{y}-\widehat{\mu }%
\mathbf{1}\right) ^{\top }\mathbf{\Sigma }_{\mathbf{\eta }}^{-1}\left( 
\mathbf{y}-\widehat{\mu }\mathbf{1}\right) \right] -\frac{1}{2}\log |\mathbf{%
	\Sigma }_{\mathbf{\eta }}\mathbf{|,}  \label{TML_PLH}
\end{equation}%
where $\mathbf{1}$ is the vector of ones and $\widehat{\mu }=\left. \mathbf{1%
}^{\top }\mathbf{\Sigma }_{\mathbf{\eta }}^{-1}\mathbf{y}\right/ \left( 
\mathbf{1}^{\top }\mathbf{\Sigma }_{\mathbf{\eta }}^{-1}\mathbf{1}\right) $%
.\ The parameter $\mathbf{\mathbf{\eta }}=(d,\mathbf{\phi }^{\top },\mathbf{%
	\theta }^{\top })^{\top }$, with $d$ the fractional differencing parameter, $%
\mathbf{\phi }$ the $p$-dimensional vector of AR parameters, $\mathbf{\theta 
}$ the $q$-dimensional vector of MA parameters, and $\sigma ^{2}\mathbf{%
	\Sigma }_{\mathbf{\eta }}:=\left[ \gamma _{i-j}\left( \mathbf{\eta }\right) %
\right] ,$ $i,j=1,\ldots ,n,$with $\gamma _{k}\left( \mathbf{\eta }\right) $
being the autocovariance\ at lag $k$. When the MLE is implemented under
correct model specification the assumed model corresponds to an ARFIMA($%
p_{0},d_{0},q_{0}$) with $\mu _{0}$ set to $0;$ under misspecification, the
ARFIMA($p,d,q$) model differs from the true data generating process (DGP) in
some way. The estimation procedure under misspecification is detailed in
Section \ref{FSP:MS}.

The\ PW estimator of $d$, $\widehat{d}^{PW}$, is obtained in two steps. In
the first step, an autoregressive-moving average (ARMA) model of order ($p,q$%
) is fit to the data, and in the second step, $d$ is estimated by minimizing
the sum of squares of the resultant residuals. Again, the PW estimator is
implemented under correct and incorrect specification of the short memory
dynamics.

\subsection{Finite sample bias and RMSE: Correct model specification\label%
{FSP:CS}}

In this section, we document results under correct specification of the true
DGP. Results are presented for the full set of values: $d_{0}=$ $\left\{
-0.25,0,0.25,0.45\right\} $\textbf{\ }and\textbf{\ }$\phi _{0}/\theta _{0}$%
\textbf{\ }$=\left\{ -0.9,-0.4,0.4,0.9\right\} $. The relative performance
of the\ jackknife method is, in turn, assessed under two scenarios: $\left(
i\right) $ when the values of all parameters in the true DGP are used in the
construction of the optimal jackknife weights and, $\left( ii\right) $ when
all parameters in the true DGP are estimated, but the correct values for lag
lengths $p_{0}$ and $q_{0}$ are still adopted (with the superscript `$Opt$'\
omitted in this case). An iterative method (described in Section \ref{UKSMD})%
\textbf{\ }is used to produce this feasible version of the jackknife\
estimator. To save on space, results for both $\widehat{d}_{J,m}^{Opt-NO}$%
and $\widehat{d}_{J,m}^{NO}$ are recorded for $m=2,3,4,6,8$, whilst results
for both $\widehat{d}_{J,m}^{Opt-MB}$ and $\widehat{d}_{J,m}^{MB}$ based on
only $m=2$ are documented. We do note that the patterns exhibited (in terms
of both bias and RMSE) for $\widehat{d}_{J,m}^{Opt-MB}$ and $\widehat{d}%
_{J,m}^{MB}$, across $m$, are similar to those exhibited for $\widehat{d}%
_{J,m}^{Opt-NO}$and $\widehat{d}_{J,m}^{NO}$ respectively.

In case $\left( i\right) $ we compare the jackknife estimator with the GS
estimator obtained with the optimal choice of $N_{n}$ ($\widehat{d}%
_{1}^{Opt-GS}$)\textbf{\ }- which, of course, relies on the known values of
the short memory parameters -\ and with the sub-optimal estimator, $\widehat{%
d}_{1}^{GS}$. In case $\left( ii\right) $ results for only $\widehat{d}%
_{1}^{GS}$ are included for comparison, as $\widehat{d}_{1}^{Opt-GS}$\ is
infeasible when the true values of the short memory parameters are unknown.%
\footnote{%
Note that in the case where the short memory parameters are unknown %
\color{blue}\cite{GuggenbergerSun2006}\color{black}\ suggest that an
adaptive procedure for the local Whittle-based estimator that they propose
could be extended to the weighted-average estimator based on LPR. Since the
adaptive method is not provided in detail in their paper, we do not pursue this option here.} Note that the finite sample results for the (raw) LPR and
PFSB estimators remain the same in both scenarios, $(i)$ and $(ii)$, as the
construction of neither estimator relies on knowledge of the true
parameters, nor of $p_{0}$ and $q_{0}$. As concerns the parametric
estimators, $\widehat{d}^{MLE}$ and $\widehat{d}^{PW}$, in this correct
specification scenario, $p=p_{0},$ $q=q_{0}$ and $\mu =\mu _{0}=0$, and $d$
is estimated under the two cases: $\left( i\right) $\ where the short memory
parameters are set at their true values; and $\left( ii\right) $\ where the
short memory parameters are estimated simultaneously with $d$.

All relevant finite sample results for case $(i)$ and case $(ii)$\ are
presented and discussed\textbf{\ }in Section \ref{KSMD} and Section \ref%
{UKSMD} respectively.

\subsubsection{Case 1: True parameters are known\label{KSMD}}

Tables \ref{Table:Bias_ARFIMA(1,d,0)_b} and \ref{Table:MSE_ARFIMA(1,d,0)_b}
record the bias and RMSE of the various optimal jackknife estimators, the
two different GS estimators, and the LPR, PFSB, MLE and PW\textbf{\ }%
estimators, for the\textbf{\ }case where the DGP is ARFIMA($1,d_{0},0$) and
the short memory parameter $\phi _{0}$ is known. The corresponding results
for the ARFIMA($0,d_{0},1$)\ DGP are presented in Tables \ref%
{Table:Bias_ARFIMA(0,d,1)_b} and \ref{Table:MSE_ARFIMA(0,d,1)_b}. The lowest
biases and RMSEs for each design are marked in boldface. The second lowest
values are italicized. Only that number which is smallest at the precision
of 8 decimal places is bolded. Values highlighted with a `*' are equally
small to 4 decimal places.

\bigskip

\begin{center}
- Table \ref{Table:Bias_ARFIMA(1,d,0)_b} here -

\medskip

- Table \ref{Table:MSE_ARFIMA(1,d,0)_b} here -

\medskip

- Table \ref{Table:Bias_ARFIMA(0,d,1)_b} here -

\medskip

- Table \ref{Table:MSE_ARFIMA(0,d,1)_b} here -

\medskip
\end{center}

With reference to Tables \ref{Table:Bias_ARFIMA(1,d,0)_b} and \ref%
{Table:MSE_ARFIMA(1,d,0)_b}: as would be anticipated in this situation, in
which the true model is estimated and the true value of $\phi _{0}$ is
imposed, the MLE is the least biased estimator of all methods considered,
and has the smallest RMSE. The parametric PW method has the\textbf{\ }second
least bias in a small number of cases, and also performs relatively well in
terms of RMSE.

As is also consistent with expectations, and existing results (see, for
example, \color{blue}\citealp*{ANW1993}\color{black},\textbf{\ }\color{blue}\citealp*{NielsenFrederiksen2005}\color{black}\ and \color{blue}\citealp*{PoskittMartinGrose2016}\color{black}), when short memory dynamics are present, the raw, unadjusted,
LPR estimator is biased, as the low frequencies are contaminated by the
spectral density of the short run dynamics, particularly for \textit{negative%
} values of $\phi _{0}$ (which corresponds to positive first-order\textbf{\ }%
autocorrelation). As is evident from the recorded results, the bias is
particularly large when there is a large negative value for $\phi _{0}$ in (%
\ref{ARFIMA(1,d,0)}), and it decreases as this value increases. Further,
both bias and RMSE decline as the sample size increases, illustrating the
consistency of the estimator.

We shall now comment on the performance of all nine bias-corrected
semi-parametric\textbf{\ }estimators under the ARFIMA($1,d_{0},0$) process.
With reference to Table \ref{Table:Bias_ARFIMA(1,d,0)_b}, for the great majority of designs, $\widehat{{\small d}}_{J,m}^{{\small Opt-NO}}$ with $%
m=2 $, has the smallest bias of all nine such estimators. For $\phi
_{0}=-0.9 $ and $n=96$, the bias reduction of $\widehat{d}_{J,m}^{Opt-NO}$ ($%
m=2$), relative to the raw LPR estimator is up to $3.6\%$,\ and when $n=576$%
, this rises to $5.7\%$.\footnote{%
We remind the reader that when $\phi _{0}=-0.9$ \textit{all} estimators
remain very biased.} For the larger values of $\phi _{0}$, when $n=96$, the
bias reduction ranges from $48\%$ to $60\%$, and from $56\%$ to $97\%$ when $%
n=576$. Only occasionally is this particular version of the jackknife
estimator inferior to an alternative semi-parametric estimator. Importantly,
however, an increase in $m$ leads to an increase in bias for $\widehat{d}%
_{J,m}^{Opt-NO}$ and, hence, a reduction in its superiority over all
alternatives, including the raw LPR method in some cases. The reason is that
the increase in $m$ leads to a smaller sub-sample length and, hence,
increases the finite sample impact of the dominant bias term on the
sub-sample estimators used in the construction of the jackknife estimator.

Now referencing the results in Table \ref{Table:MSE_ARFIMA(1,d,0)_b}, we see
that, despite the lack of variance inflation in the asymptotic distribution
of the optimal jackknife estimator, the reduction in bias does cause some
finite sample increase in variance, leading to RMSEs for $\widehat{d}%
_{J,m}^{Opt-NO}$ that are occasionally slightly larger than the RMSE of the
raw LPR estimator. That said, in the vast majority of cases $\widehat{d}%
_{J,m}^{Opt-NO}$ with $m=8$, has the smallest RMSE of all semi-parametric%
\textbf{\ }estimators (including the raw LPR) and, in many cases, the RMSE
of the jackknife estimator with the smallest bias ($\widehat{d}%
_{J,m}^{Opt-NO}$, $m=2$) has a RMSE which remains less than that of the raw
estimator. In addition, all versions of the jackknife estimator (including
the moving-block version) tend to have smaller RMSEs than the three
alternative bias-corrected methods ($\widehat{d}_{1}^{GS}$, $\widehat{d}%
_{1}^{Opt-GS}$ and $\widehat{d}_{1}^{PFSB}$), most notably for the smaller
sample size ($n=96\,$). As befits the optimality of the estimator, in almost
all cases, $\widehat{d}_{1}^{Opt-GS}$ out-performs $\widehat{d}_{1}^{GS}$,
in terms of both bias and RMSE, although both estimators, as already noted,
are virtually always out-performed by a version of the jackknife procedure.

The broad conclusions drawn above obtain under the ARFIMA($0,d_{0},1$) DGP,
as seen from the results recorded in Tables \ref{Table:Bias_ARFIMA(0,d,1)_b}
and \ref{Table:MSE_ARFIMA(0,d,1)_b}. The only notable difference is the
improved RMSE performance of $\widehat{d}_{J,8}^{Opt-NO}$, with this
estimator ranked second overall (after $\widehat{d}^{MLE}$) in terms of this
measure.

\subsubsection{Case 2: True parameters are unknown\label{UKSMD}}

Evaluation of the optimal weights in (\ref{w_star}), required for the
construction of the optimal jackknife estimator, depends on the covariances
between both the different sub-sample LPR estimators and between the
full-sample and sub-sample estimators, as given in (\ref%
{Covariance(dhat,dhat_i)}) and (\ref{Covariance(dhat_i,dhat_j)}). These
covariances depend, in turn, on covariances between the various
log-periodograms and, hence, on the values of the parameters that underpin
the true DGP, as is made explicit in (\ref{Covariance between samples}) and
Appendix B. Hence, implementation of the optimal bias-correction procedure
via the jackknife is not feasible in practice, without modification. To this
end, we propose\ the following iterative method for obtaining a feasible
version of the jackknife-based estimator; one still appropriate, however,
for the case where the specified model (i.e. the values of $p_{0}$ and $%
q_{0} $) is correct.

\paragraph{An iterative version of the optimal jackknife estimator\label%
{Sec:Iterative method}}

\begin{enumerate}
\item \textbf{Prerequisite:} Estimate the relevant short memory parameter(s)
in the ARFIMA($p_{0},d_{0},q_{0}$) model, using pre-filtered data based on $d^{f}=\widehat{{\small d}}^{GS}.$

\item \textbf{Initialization:} Set $k=1$ and tolerance level $\tau =\tau^{\left( 0\right) }$.

\item \textbf{Recursive step:} For the $k^{th}$ recursion, perform the
jackknife bias-correction procedure of Section \ref{Optimization}, but%
\textbf{\ }with the estimates of the short memory parameters from step 1,
and $d^{f}=\widehat{d}^{GS}$, now inserted into the formulae for the
covariance terms \ in (\ref{Covariance(dhat,dhat_i)}) and (\ref%
{Covariance(dhat_i,dhat_j)}). Denote the resulting estimator by $\widehat{d}%
_{J,m}^{\left( k\right) }$.

\item \textbf{Stopping rule:} If $\left\vert \widehat{d}_{J,m}^{\left(
k+1\right) }-\widehat{d}_{J,m}^{\left( k\right) }\right\vert >\tau $ set $%
k=k+1$ and $\tau =\tau ^{\left( k\right) },$ and repeat steps $1$ and $3$
after updating $d^{f}=\widehat{d}_{J,m}^{\left( k\right) }$.
\end{enumerate}

The basic idea behind the algorithm is as follows: estimation of the short
memory parameter requires pre-filtering via some preliminary estimate of $%
d_{0}$. An obvious initial (consistent) choice is $d^{f}=\widehat{d}^{GS}$,
as this estimator is already bias-adjusted, and a feasible estimator in the
presence of unknown values for the short memory parameter(s). However $%
\widehat{d}^{GS}$ will still exhibit some bias in finite samples. Hence,
iteration of the above algorithm, which involves replacing the initial
pre-filtering value with successively less biased values, $d^{f}=\widehat{d}%
_{J,m}^{\left( k\right) }$, is expected to yield a final feasible version of
the jackknife estimator, $\widehat{d}_{J,m}^{\left( k+1\right) }$, based on
accurate estimates of all unknown parameters. (See also \color{blue}\citealp*{PoskittMartinGrose2016}\color{black}\ for a related application of this form of iterative
procedure). The feasible version of the jackknife statistic at the final
iteration\textbf{\ }is denoted hereafter by $\widehat{d}_{J,m}^{NO}$ if the
sub-sampling method is non-overlapping and $\widehat{d}_{J,m}^{MB}$ if the
sub-sampling method is moving-block.

\medskip

\begin{center}
- Table \ref{Table:Bias_ARFIMA(1,d,0)_b_I} here -

\medskip

- Table \ref{Table:MSE_ARFIMA(1,d,0)_b_I} here -

\medskip

- Table \ref{Table:Bias_ARFIMA(0,d,1)_b_I} here -

\medskip

- Table \ref{Table:MSE_ARFIMA(0,d,1)_b_I} here -

\medskip
\end{center}

Tables \ref{Table:Bias_ARFIMA(1,d,0)_b_I} and \ref%
{Table:MSE_ARFIMA(1,d,0)_b_I} record the bias and RMSE of all versions of
the feasible jackknife estimator, the feasible GS estimator, $\widehat{d}%
_{1}^{GS}$, and the LPR, PFSB, MLE and PW\textbf{\ }estimators, for the%
\textbf{\ }case where the DGP is ARFIMA($1,d_{0},0$) and the short memory
parameter $\phi _{0}$ is now estimated. The corresponding results for the
ARFIMA($0,d_{0},1$)\ DGP are presented in Tables \ref%
{Table:Bias_ARFIMA(0,d,1)_b_I} and \ref{Table:MSE_ARFIMA(0,d,1)_b_I}. The
results for the $\widehat{d}_{1}^{GS}$, LPR and PFSB are the same as in the
earlier corresponding tables, as these estimators do not depend on knowledge
or estimation of the short memory dynamics. The parametric estimators, MLE
and PW, do of course change when $\phi _{0}$ is estimated.\textbf{\ }Once
again, the minimum bias and RMSE are shown in bold font, and the second
lowest values are italicized.

Consider the results for the ARFIMA($1,d_{0},0$) process (Tables \ref%
{Table:Bias_ARFIMA(1,d,0)_b_I} and \ref{Table:MSE_ARFIMA(1,d,0)_b_I}). The
(various versions of the) feasible jackknife estimators show similar
characteristics to the corresponding optimal estimators, except for
exhibiting larger bias and RMSE. This is to be expected given that the
optimal weights are now functions of estimates of\textbf{\ }both $d_{0}$ and 
$\phi _{0}$. The increase in bias (relative to the known parameter case) is
particularly marked when $\phi _{0}=-0.9$, with the feasible jackknife
estimators seen to be more biased overall than the raw LPR estimator itself,
in three cases. However, for all other values of $\phi _{0}$, the least
biased versions of the\textbf{\ }feasible jackknife estimators are still
almost always less biased than the LPR estimator. For example, when $\phi
_{0}=-0.4$ and $n=96$, the bias reduction of $\widehat{d}_{J,m}^{NO}$ with $%
m=2$ compared to the raw LPR\ estimator is up to $35\%$ and when $n=576$,
the bias reduction rises to $69\%$. Overall, the $\widehat{d}_{J,2}^{NO}$, $%
\widehat{d}_{1}^{GS}$, $\widehat{d}_{1}^{PFSB}$ and $\widehat{d}^{MLE}$
estimators share the title of the least, or second-least biased estimator.%
\textbf{\ }The RMSE results in Table \ref{Table:MSE_ARFIMA(1,d,0)_b_I}
indicate the consistency of the feasible jackknife estimators.\ However, the
MLE estimator still exhibits the least RMSE of all estimators considered,
even when $\phi _{0}$ is estimated, with the LPR estimator taking second
place.

The results in Tables \ref{Table:Bias_ARFIMA(0,d,1)_b_I} and \ref%
{Table:MSE_ARFIMA(0,d,1)_b_I}, for the ARFIMA($0,d_{0},1$) process, tell a
broadly similar story to those for the ARFIMA($1,d_{0},0$) case, except for
the fact that $\widehat{d}_{J,2}^{NO}$\ is now the least biased estimator in
more cases than any other competing estimator, and $\widehat{d}_{J,8}^{NO}$\
is sometimes ranked second in terms of RMSE.

\subsection{Finite sample bias and RMSE: Model misspecification\label{FSP:MS}}

Misspecification occurs when the true DGP is ARFIMA($p_{0},d_{0},q_{0}$) and
the fitted model is ARFIMA($p,d,q$), where $p$ and $q$ are\textbf{\ }such
that $\left\{ p\neq p_{0}\cup q\neq q_{0}\right\} \backslash \left\{
p_{0}\leq p\cap q_{0}\leq q\right\} $. We consider three different forms of
misspecification: $\left( i\right) $ true\textbf{\ }DGP: ARFIMA($1,d_{0},0$%
); fitted model: ARFIMA($0,d,0$); $\left( ii\right) $ true\textbf{\ }DGP:
ARFIMA($0,d_{0},1$); fitted model: ARFIMA($0,d,0$); and\textbf{\ }$\left(
iii\right) $ true\textbf{\ }DGP: ARFIMA($1,d_{0},1$); fitted model: ARFIMA($%
2,d,0$). The first two forms of misspecification mimic a situation in which
short memory dynamics are present, but are ignored. In particular, these
scenarios allow us to assess the relative performance of the feasible
jackknife estimator when no aspect of the short memory specification is used
in the calculation of the weights. The third form of misspecification allows
for another type of error in the specification of the short memory
component. In all cases, we restrict both the DGP and the fitted model to be
within the stationary region. In order to reduce the number of results to be
tabulated and discussed, in Tables 10 to 13 we present results for the
reduced set of values: $d_{0}=$ $\left\{ -0.25,0.25,0.45\right\} $\ and $%
\phi _{0}/\theta _{0}$\ $=\left\{ -0.9,-0.4,0.4,0.9\right\} $. In Tables 14
and 15, we reduce the settings further by omitting results for $\phi
_{0}/\theta _{0}=0.9.$

Under misspecification, the feasible jackknife estimates are obtained by
using the fitted ARFIMA($p,d,q$) model in the Prerequisite step in Section %
\ref{Sec:Iterative method}, while the remaining steps are unchanged. The MLE
and PW estimators are produced as explained in Section \ref{mc}.\footnote{%
For more details on MLE under mis-specification of the short-memory dynamics see, \color{blue}\cite{NadarajahMartinPoskitt2017}\color{black}}\ The
estimators $\widehat{d}_{n}$,\ $\widehat{d}^{GS}$ and $\widehat{d}^{PFSB}$
are not affected by misspecification of the short memory dynamics, as
specification of that component of the model plays no role in their
construction. Hence, the results for these estimators in Tables 10 to 13
match the corresponding results in Tables 2 to 5. The results for all
estimators in Table 14, under the ARFIMA($1,d_{0},1$) DGP, are distinct from
results in all other tables.

Tables \ref{Table:Bias_ARFIMA(1,d,0)_mis} and \ref%
{Table:MSE_ARFIMA(1,d,0)_mis} display the bias and RMSE results of all
estimators under misspecification type $\left( i\right) $. The corresponding
results for the misspecification types $\left( ii\right) $ and $\left(
iii\right) $ are presented in Tables \ref{Table:Bias_ARFIMA(0,d,1)_mis} to %
\ref{Table:MSE_ARFIMA(1,d,1)_mis}. As previously,\textbf{\ }the minimum bias
and RMSE are shown in bold font, and the second lowest values are italicized.

\bigskip

\begin{center}
- Table \ref{Table:Bias_ARFIMA(1,d,0)_mis} here -

\medskip

- Table \ref{Table:MSE_ARFIMA(1,d,0)_mis} here -

\medskip

- Table \ref{Table:Bias_ARFIMA(0,d,1)_mis} here -

\medskip

- Table \ref{Table:MSE_ARFIMA(0,d,1)_mis} here -

\medskip

- Table \ref{Table:Bias_ARFIMA(1,d,1)_mis} here -

\medskip

- Table \ref{Table:MSE_ARFIMA(1,d,1)_mis} here -
\end{center}

\bigskip

From Tables \ref{Table:Bias_ARFIMA(1,d,0)_mis} and \ref%
{Table:MSE_ARFIMA(1,d,0)_mis}, we observe that under misspecification the
(various versions of the) feasible jackknife estimators show similar
characteristics to those\textbf{\ }observed under correct model
specification, although with larger bias and RMSE. This is not surprising,
given that the\textbf{\ }weights are now functions of estimates of $d$ only,
with information on the true or estimated autoregressive coefficient in the
DGP ignored. The increase in bias (relative to the correct specification
case) is particularly marked when $\phi _{0}=-0.9$. When $\phi _{0}>-0.9$,
the feasible jackknife estimators still tend to show reduced bias compared
to the LPR estimator, almost uniformly for $m=2$. For example, when $\phi_{0}=-0.4$ and $n=96$, the bias reduction of $\widehat{{\small d}}%
_{J,2}^{NO} $ compared to the raw LPR\ estimator is up to $12\%$, and when $%
n=576$, the bias reduction rises to $39\%$. Moreover when $\phi _{0}=0.9$, $\widehat{{\small d}}_{J,2}^{NO}$ is either the least, or second-least biased estimator of all estimators considered.

Under this form of misspecification, at least for\textbf{\ }$\phi _{0}>-0.9$, the\textbf{\ }MLE and PW estimators are more biased than the feasible
jackknife estimators, and much more so in some cases. This is expected, as
model misspecification has a\textit{\ direct} impact on the parametric
estimators. In contrast, for the jackknife estimators, misspecification
impinges more \textit{indirectly,} only via the choice of\textbf{\ }weights.
Overall, the PFSB estimator shows the least bias and the feasible\textbf{\ }%
GS estimator shows the second-least bias.

The RMSE results in Table \ref{Table:MSE_ARFIMA(1,d,0)_mis} demonstrate that
neither the feasible jackknife estimators, nor the parametric methods,
out-perform the raw LPR estimator, which\textbf{\ }shows the least RMSE; the
second lowest RMSE usually being observed with one of\textbf{\ }$\widehat{d}%
_{J,8}^{NO}$, $\widehat{d}_{1}^{GS}$, $\widehat{d}^{PFSB}$ or $\widehat{d}%
^{MLE}$.

The bias results for misspecification\ form $\left( ii\right) $ in Table \ref%
{Table:Bias_ARFIMA(0,d,1)_mis} display similar characteristics to those
described for misspecification form $\left( i\right) $, apart from the much
more distinct dominance of $\widehat{d}_{J,2}^{NO}$\ in this case. The RMSE
results in Table \ref{Table:MSE_ARFIMA(0,d,1)_mis} reveal that, once again, $%
\widehat{d}_{n}$ has the smallest RMSE values, with $\widehat{d}^{MLE}$
taking the second place.

Under the third form of misspecification, the bias estimates indicate that $%
\widehat{d}_{1}^{GS}$\textbf{\ }exhibits the least bias, with $\widehat{d}%
_{J,2}^{NO}$ taking the second place (refer Table \ref%
{Table:Bias_ARFIMA(1,d,1)_mis}). In terms of the\textbf{\ }RMSE results in
Table \ref{Table:MSE_ARFIMA(1,d,1)_mis}, once again the raw LPR has the
lowest values most frequently, followed by\textbf{\ }$\widehat{d}_{1}^{GS}$,
with the \textit{second} smallest RMSE values mostly observed for $\widehat{d%
}_{J,8}^{NO}$.

\subsection{A summary of the simulation results\label{summ}}

\begin{center}
\medskip

- Table \ref{rank} here -

\medskip
\end{center}

To assist the reader, in Table \ref{rank} we summarize all of the simulation
results tabulated in Tables \ref{Table:Bias_ARFIMA(1,d,0)_b} to \ref%
{Table:MSE_ARFIMA(1,d,1)_mis}, by ranking the estimators - from first to
third - under the different scenarios. Panel A in Table \ref{rank}
summarizes the results in Tables \ref{Table:Bias_ARFIMA(1,d,0)_b} to \ref%
{Table:MSE_ARFIMA(0,d,1)_b}; Panel B summarizes the results in Tables \ref%
{Table:Bias_ARFIMA(1,d,0)_b_I} to\ \ref{Table:MSE_ARFIMA(0,d,1)_b_I}; and
Panel C summarizes the results in Tables \ref{Table:Bias_ARFIMA(1,d,0)_mis}
to \ref{Table:MSE_ARFIMA(1,d,1)_mis}, with the three misspecification types
-- $\left( i\right) $ to $\left( iii\right) $ -- corresponding to those
described in the above section. An estimator is ranked first if it has the
smallest value (in bold font) for the relevant measure (bias or RMSE) the
largest number of times in a given table. The other ranks follow
accordingly. If needed to complete the ranking, the method with the largest
number of second-smallest values (in italic font) in a given Table is
referenced. And so on. The rankings accord with the narrative in the
preceding sections.

\section{Discussion}

With the fractionally integrated autoregressive moving-average model being
one of the key model classes for describing long memory processes, much
effort has been expended on producing accurate estimates of the fractional
differencing parameter, $d$, in particular. This quest has been hampered by
certain problems, for both parametric and semi-parametric approaches.
Specifically, the need to fully specify the model for parametric estimation
means that any incorrect specification of the short memory dynamics has
serious consequences, in terms of both finite sample and asymptotic
properties (see, for example, \color{blue}\citealp*{ChenDeo2006}\color{black}\ and \color{blue}\citealp*{NadarajahMartinPoskitt2017}
\color{black}). On the other hand, the semi-parametric estimators, whilst
not requiring explicit modelling of the short memory component, can suffer
substantial finite sample bias in the presence of unaccounted for short
memory dynamics. It is bias-correction of this latter class of estimator
that has been the focus of this paper.

A natural way of producing a bias-corrected version of the commonly used the
log-periodogram regression (LPR) estimator\ is suggested in this article,
based on the jackknife technique. Optimality is achieved by allocating
weights within the jackknife that are adjusted for the bias to a particular
order,\ and that minimize the increase in variance caused by the reduction
in bias. The construction of the optimally\ bias-corrected estimator
requires expressions for the dominant bias term and variance of the
unadjusted LPR estimator. We show that the statistical properties of the LPR
estimator, as originally established by \color{blue}\cite{HurvichDeoBrodshky1998}\color{black}, are valid for a more general class of
fractional process that is not necessarily Gaussian. Hence, the jackknife
estimator that we construct from the optimally weighted average of LPR
estimators also has proven optimality under this general form of process. In
addition to proving the consistency of the optimal jackknife estimator, we
have the important result that the asymptotic variance of the estimator is
equivalent to that of the unadjusted LPR estimator. That is, bias adjustment
is effected without any associated increase in asymptotic variance.

Our Monte Carlo study shows that, amongst the semi-parametric estimators,
the optimal jackknife estimator based on a small number of non-overlapping
sub-samples outperforms (in terms of bias reduction) both the pre-filtered sieve bootstrap estimator of \color{blue}\cite{PoskittMartinGrose2016}\color{black}\ and the weighted-average estimator of \color{blue}\cite{GuggenbergerSun2006}\color{black}, albeit in the somewhat artificial case in which the parameters of the DGP are correctly identified and known, for
the purpose of computing optimal weights.\textbf{\ }In the realistic case in
which these parameters are not known, we suggest an iterative procedure in
which the weights are constructed using consistent estimates. In this case
the method is not dominant overall, compared to alternative bias-corrected
methods, but is still the least biased in some cases. The relationship
between the semi-parametric methods and the two parametric methods is much
as anticipated. In particular, the semi-parametric methods dominate in terms
of both bias and RMSE when the short memory dynamics are misspecified. Once
again, a version of the feasible jackknife method is ranked highly under
certain misspecified settings, despite the fact that the misspecification
impacts on the construction of the jackknife weights.

Throughout the paper we assume that the number of sub-samples is fixed. One
may wish to allow the number of sub-samples to vary and explore the
characteristics of the resultant bias-adjusted estimators in this case.
Importantly, alternative methods of estimating the weights are to be
investigated, including the possible use of a non-parametric estimate of the
spectral density (see, \color{blue}\citealp*{MoulinesSoulier1999}\color{black}), rather than replacing the true values with their consistent estimates, or the use of an adaptive method in the spirit of that suggested by \color{blue}\cite{GuggenbergerSun2006}\color{black}.

Finally, although we focus on the LPR estimator, the jackknife procedure can
easily be applied to other estimators such as the local Whittle estimator of \color{blue}\cite{Kunsch1987}\color{black}, the local polynomial Whittle estimator of \color{blue}\cite{AndrewSun2004}\color{black}\ or even to the (already analytically) bias-reduced estimators of \color{blue}\cite{AndrewsGuggenberger2003}\color{black}\ and \color{blue}\cite{GuggenbergerSun2006}\color{black}. Another possible extension is to relax the assumption of stationarity of the process using the results \color{blue}\cite{Velasco1999b}\color{black}, and to derive the properties the optimal jackknife estimators in the nonstationary setting.

\bibliographystyle{appa}
\bibliography{My_library}

\begin{subappendices}

\section*{Appendix A: Proofs of Theorems and Lemmas\label{Appendix:A}}

\begin{proof}[Proof of Theorem 1]
	Under Assumptions $\left( A.1\right) -\left( A.4\right) $, the proof of the
	theorem follows immediately after applying the results of Corollary A.1 of \color{blue}\cite{NadarajahMartinPoskitt2017}\color{black}\ to Lemmas, 2, 5, 6 and 7 of \color{blue}\cite{HurvichDeoBrodshky1998}\color{black}. Hence we
	omit the proof.
\end{proof}

Prior to providing the proofs of the other theorems and lemmas, we will
introduce the following definition, and its properties, to be used
hereinafter.

Define $\Delta ^{\left( T\right) }\left( \lambda \right) =\sum_{t=1}^{T}\exp
\left( -\imath \lambda t\right) .$ Then,%
\begin{equation}
\Delta ^{\left( T\right) }\left( \lambda \right) =\exp \left( -\imath 
\frac{\lambda }{2}\left( T+1\right) \right) \frac{\sin \left( \frac{\lambda T%
	}{2}\right) }{\sin \left( \frac{\lambda }{2}\right) } \\
= \begin{cases}
0 & \quad \textit{if} \quad  \lambda \not \equiv 0 (mod \ \pi)  \\
T & \quad \textit{if} \quad  \lambda \equiv 0 (mod \ 2\pi)  \\ 
0 \quad \text{or} \quad T & \quad \textit{if} \quad  \lambda = \pm\pi, \pm 3\pi, \dots 
\end{cases} \tag{A.1} \label{delta}
\end{equation} 
where, $a\equiv b\left({mod} \alpha \right) $ means that the
difference $\left( a-b\right) $ is an integral multiple of $\alpha $ for $%
\alpha ,x,y\in 
\mathbb{R}
.$

Consider%
\begin{eqnarray*}
	\sum_{t=-T}^{T}\exp \left\{ -\imath \lambda t\right\} 
	&=&1+\sum_{t=1}^{T}\exp \left\{ -\imath \lambda t\right\}
	+\sum_{t=1}^{T}\exp \left\{ -\imath \left( -\lambda \right) t\right\}  \\
	&=&1+2\Delta ^{\left( T\right) }\left( \lambda \right) ,\text{ using (\ref%
		{delta}).}
\end{eqnarray*}%
This immediately gives that%
\begin{equation}
\lim_{T\rightarrow \infty }\dfrac{1}{2\pi }\sum_{t=-T}^{T}\exp \left\{
-\imath \lambda t\right\} =\eta \left( \lambda \right) .  \tag{A.2}
\label{limit of big delta1}
\end{equation}

We will derive the following two properties of $\Delta ^{\left( T\right)
}\left( \lambda \right) $.

\begin{enumerate}
	\item Sum:%
	\begin{align}
	\lim_{T\rightarrow \infty }\left[ \Delta ^{\left( T\right) }\left( \lambda
	\right) +\Delta ^{\left( T\right) }\left( -\lambda \right) \right] 
	&=\lim_{T\rightarrow \infty }%
	\Big(%
	\sum_{t=-T}^{T}\exp \left\{ \imath \lambda t\right\} -1%
	\Big)
	\notag \\
	&=2\pi \eta \left( \lambda \right) -1,\text{ by (\ref{limit of big delta1}).%
	}  \tag{A.3}  \label{Sum of deltas}
	\end{align}
	
	\item Product:%
	\begin{align}
	T^{-2}\Delta ^{\left( T\right) }\left( -\lambda \right) \Delta ^{\left(
		T\right) }\left( \lambda \right)  &=T^{-2}\sum_{t=1}^{T}\sum_{s=1}^{T}\exp
	\left\{ -\imath \lambda \left( t-s\right) \right\}   \notag \\
	&=T^{-2}\sum_{t=-\left( T-1\right) }^{T-1}\left( T-\left\vert t\right\vert
	\right) \exp \left\{ -\imath \lambda t\right\}   \notag \\
	&=T^{-1}\sum_{t=-\left( T-1\right) }^{T-1}\exp \left\{ -\imath \lambda
	t\right\} -\sum_{t=-\left( T-1\right) }^{T-1}\frac{\left\vert t\right\vert }{%
		T^{2}}\exp \left\{ -\imath \lambda t\right\} .  \tag{A.4}
	\label{Product of deltas step 1}
	\end{align}%
	Consider the second term in the above expression,%
	\begin{equation*}
	\Big|%
	\sum_{t=-\left( T-1\right) }^{T-1}\frac{\left\vert t\right\vert }{T^{2}}\exp
	\left\{ -\imath \lambda t\right\} 
	\Big|%
	\leq 
	\Big|%
	\sum_{t=-\left( T-1\right) }^{T-1}\frac{\left\vert t\right\vert }{T^{2}}%
	\Big|%
	\rightarrow 0\text{ as }T\rightarrow \infty .
	\end{equation*}%
	Hence the expression in (\ref{Product of deltas step 1}) is given by,%
	\begin{equation}
	T^{-2}\Delta ^{\left( T\right) }\left( -\lambda \right) \Delta ^{\left(
		T\right) }\left( \lambda \right) =T^{-1}2\pi \eta \left( \lambda \right)
	+o(1).  \tag{A.5}  \label{Product of deltas}
	\end{equation}
\end{enumerate}

\begin{lemma}
	\label{Lemma A.1}Let $\mathbf{W}_{t}$ be a stationary $h$ vector-valued time
	series with $n$ observations satisfying the spectral density given in (1).
	Suppose that Assumptions $\left( A.1\right) -\left( A.3\right) $ hold. The $%
	k^{th}$-order cumulant of the multivariate series, $\kappa \left\{
	D_{W_{a_{1}}}^{\left( n\right) }\left( \lambda _{1}\right)
	,...,D_{W_{a_{k}}}^{\left( n\right) }\left( \lambda _{k}\right) \right\} $ is%
	\begin{equation}
	n^{-\frac{k}{2}}\left( 2\pi \right) ^{\frac{k}{2}-1}\Delta ^{\left( n\right)
	}%
	\Bigg(%
	\sum_{j=1}^{k}\lambda _{j}%
	\Bigg)%
	f_{W_{a_{1}}\ldots W_{a_{k}}}\left( \lambda _{1},...,\lambda _{k-1}\right) +o%
	\big(%
	n^{1-2d-\frac{k}{2}}%
	\big)%
	.  \tag{A.6}  \label{Cum_DFT1}
	\end{equation}%
	where $f_{W_{a_{1}}\ldots W_{a_{k}}}\left( \lambda _{1},...,\lambda
	_{k-1}\right) $ is the $k^{th}$-order spectrum of the series $\mathbf{W}_{t}$%
	, with $a_{1},\ldots ,a_{k}=1,2,\ldots ,h,$ and $k=1,2,\ldots $.
\end{lemma}

\begin{proof}
	By Lemma P4.2 of\ \color{blue}\cite{Brillinger1981}\color{black}, the
	cumulant, $\kappa 
	\big\{%
	D_{W_{a_{1}}}^{\left( n\right) }\left( \lambda _{1}\right) ,$ $...,$ $%
	D_{W_{a_{k}}}^{\left( n\right) }\left( \lambda _{k}\right) 
	\big\}%
	$ has the form%
	\begin{equation*}
	\sum\limits_{t_{1}=-\infty }^{\infty }...\sum\limits_{t_{k}=-\infty
	}^{\infty }\exp 
	\Bigg(%
	-\imath \sum\limits_{j=1}^{k}\lambda _{j}t_{j}%
	\Bigg)%
	\kappa _{W_{a_{1}}\ldots W_{a_{k}}}(t_{1}-t_{k},...,t_{k-1}-t_{k})
	\end{equation*}%
	Substituting, $u_{j}=t_{j}-t$ where $t=t_{k},$ and $-S\leq u_{j}\leq S,$ for 
	$j=1,\ldots ,k-1$ with $S=2\left( n-1\right) $ we have that%
	\begin{eqnarray*}
		&&\kappa 
		\big\{%
		D_{W_{a_{1}}}^{\left( n\right) }\left( \lambda _{1}\right)
		,D_{W_{a_{2}}}^{\left( n\right) }\left( \lambda _{2}\right)
		,...,D_{W_{a_{k}}}^{\left( n\right) }\left( \lambda _{k}\right) 
		\big\}
		\\
		&=&\left( 2\pi n\right) ^{-\frac{k}{2}}\sum_{t=-\infty }^{\infty
		}\sum_{u_{1}=-S}^{S}\cdots \sum_{u_{k}=-S}^{S}\exp 
		\Bigg(%
		-\imath \sum\limits_{j=1}^{k}\lambda _{j}\left( u_{j}+t\right) 
		\Bigg)%
		\kappa _{W_{a_{1}}\ldots W_{a_{k}}}\left( u_{1},...,u_{k-1}\right) \\
		&=&\left( 2\pi n\right) ^{-\frac{k}{2}}\sum_{u_{1}=-S}^{S}\cdots
		\sum_{u_{k}=-S}^{S}\exp 
		\Bigg(%
		-\imath \sum\limits_{j=1}^{k-1}\lambda _{j}u_{j}%
		\Bigg)%
		\kappa _{W_{a_{1}}\ldots W_{a_{k}}}\left( u_{1},...,u_{k-1}\right)
		\sum_{t=-\infty }^{\infty }\exp 
		\Bigg(%
		-\imath \sum\limits_{j=1}^{k}\lambda _{j}t%
		\Bigg)
		\\
		&=&\left( 2\pi \right) ^{-\frac{k}{2}+1}n^{-\frac{k}{2}}\Delta ^{\left(
			n\right) }\left( \sum\nolimits_{j=1}^{k}\lambda _{j}\right)
		\sum_{u_{1}=-S}^{S}\cdots \sum_{u_{k}=-S}^{S}\exp 
		\Bigg(%
		-\imath \sum\limits_{j=1}^{k-1}\lambda _{j}u_{j}%
		\Bigg)%
		\kappa _{W_{a_{1}}\ldots W_{a_{k}}}\left( u_{1},...,u_{k-1}\right) .
	\end{eqnarray*}
	
	The rapidity of the convergence of $\sum_{u_{1}=-S}^{S} \cdots
	\sum\nolimits_{u_{k}=-S}^{S}\exp (-\imath \sum\nolimits_{j=1}^{k-1}\lambda
	_{j}u_{j})\kappa _{W_{a_{1}}\ldots W_{a_{k}}}\left( u_{1},...,u_{k-1}\right) 
	$ to $f_{W_{a_{1}}\ldots W_{a_{k}}}\left( \lambda _{1},\ldots ,\lambda
	_{k-1}\right) $ as $n\rightarrow \infty $ is measured as follows.%
	\begin{eqnarray*}
		&&%
		\Bigg|%
		\sum_{u_{1}=-S}^{S}\cdots \sum_{u_{k}=-S}^{S}\exp 
		\Bigg(%
		-\imath \sum\limits_{j=1}^{k-1}\lambda _{j}u_{j}%
		\Bigg)%
		\kappa _{W_{a_{1}}\ldots W_{a_{k}}}\left( u_{1},...,u_{k-1}\right)
		-f_{W_{a_{1}}\ldots W_{a_{k}}}\left( \lambda _{1},\ldots ,\lambda
		_{k-1}\right) 
		\Bigg|
		\\
		&=&%
		\Bigg|%
		\sum_{\left\vert u_{1}\right\vert >S}\cdots \sum_{\left\vert
			u_{k}\right\vert >S}\exp 
		\Bigg(%
		-\imath \sum\limits_{j=1}^{k-1}\lambda _{j}u_{j}%
		\Bigg)%
		\kappa _{W_{a_{1}}\ldots W_{a_{k}}}\left( u_{1},...,u_{k-1}\right) 
		\Bigg|
		\\
		&\leq &\sum_{\left\vert u_{1}\right\vert >S}\cdots \sum_{\left\vert
			u_{k}\right\vert >S}\left\vert \kappa _{W_{a_{1}}\ldots W_{a_{k}}}\left(
		u_{1},...,u_{k-1}\right) \right\vert \\
		&\leq &n^{-1+2d}\sum_{\left\vert u_{1}\right\vert >S}\cdots \sum_{\left\vert
			u_{k}\right\vert >S}\left( \left\vert \frac{u_{1}}{n}\right\vert
		^{1-2d}+\cdots +\left\vert \frac{u_{k-1}}{n}\right\vert ^{1-2d}\right)
		\left\vert \kappa _{W_{a_{1}}\ldots W_{a_{k}}}\left(
		u_{1},...,u_{k-1}\right) \right\vert .
	\end{eqnarray*}%
	Hence the proof is completed since Assumption $\left( A.1\right) $ holds and 
	$n^{-1+2d}\left( \left\vert u_{1}\right\vert +\cdots +\left\vert
	u_{k-1}\right\vert \right) \rightarrow 0$ as $n\rightarrow \infty .$
\end{proof}

The above Lemma shows that when the DFTs correspond to multivariate time
series with the same number of observations in their sample, the $k^{th}$%
-order cumulant of the multivariate series can be approximated with the
expression given in (\ref{Cum_DFT1}). The only difference between this Lemma
and Proposition 1 is that the proposition deals with different sample sizes
for the time series in the multivariate set-up.

\begin{proof}[Proof of Proposition 1]
	The proof of the proposition can be established in a similar fashion to the
	above proof. Hence, we omit the proof here.
\end{proof}

\begin{proof}[Proof of Theorem 2]
	The expectation of the DFT of the full sample or the sub-sample is%
	\begin{eqnarray*}
		E\left( D_{X_{a_{i}}}^{\left( L_{i}\right) }\left( \lambda \right) \right) 
		&=&\tfrac{1}{\sqrt{2\pi n}}\sum\nolimits_{t=1}^{n}\exp \left( -\imath
		\lambda t\right) E\left( y_{t}\right)  \\
		&=&\tfrac{\mu _{Y}}{\sqrt{2\pi L_{i}}}\Delta ^{\left( L_{i}\right) }\left(
		\lambda \right)  \\
		&=&\left\{ 
		\begin{tabular}{lll}
			$0$ & $if$ & $\lambda 
			\not%
			\equiv 0\left( mod \ \pi \right) $ \\ 
			$\sqrt{\tfrac{L_{i}}{2\pi }}\mu _{Y}$ & $if$ & $\lambda \equiv \pi \left( 
			mod \ 2\pi \right) $ \\ 
			$0$ $or$ $\sqrt{\tfrac{L_{i}}{2\pi }}\mu _{Y}$ & $if$ & $\lambda =\pm \pi
			,\pm 3\pi ,\ldots $%
		\end{tabular}%
		\right. ,
	\end{eqnarray*}%
	where $E\left( y_{t}\right) =\mu _{Y}.$ Therefore, $D_{X_{a_{i}}}^{\left(
		L_{i}\right) }\left( \lambda \right) $ behaves in the manner required by the
	theorem as the first-order cumulant provides the mean of the random variable
	of interest.
	
	The covariance between $D_{X_{a_{i}}}^{\left( L_{i}\right) }\left( \lambda
	\right) $ and $D_{X_{a_{j}}}^{\left( L_{j}\right) }\left( \mu \right) $ is
	measured by the second-order cumulant and Proposition 1 gives that%
	\begin{equation*}
	Cov\left( D_{X_{a_{i}}}^{\left( L_{i}\right) }\left( \lambda \right)
	,D_{X_{a_{j}}}^{\left( L_{j}\right) }\left( \mu \right) \right) =\tfrac{1}{L}%
	\Delta ^{\left( L\right) }\left( \lambda +\mu \right)
	f_{X_{a_{i}},X_{a_{j}}}\left( \lambda \right) +o\left( L^{^{-2d}}\right) ,
	\end{equation*}%
	where $L=\min \left( L_{i},L_{j}\right) $. Thus, the covariance between the
	DFTs of the full sample and the sub-sample tends to $0$ as $n\rightarrow
	\infty .$
\end{proof}

\begin{proof}[Proof of Theorem 3]
	The covariance between $I_{X_{a_{i}}}^{\left( L_{i}\right) }\left( \lambda
	\right) $ and $I_{X_{a_{j}}}^{\left( L_{j}\right) }\left( \mu \right) $ is
	given by,%
	\begin{eqnarray*}
		Cov\left( I_{X_{a_{i}}}^{\left( L_{i}\right) }\left( \lambda \right)
		,I_{X_{a_{j}}}^{\left( L_{j}\right) }\left( \mu \right) \right)  &=&E\left(
		I_{X_{a_{i}}}^{\left( L_{i}\right) }\left( \lambda \right)
		I_{X_{a_{j}}}^{\left( L_{j}\right) }\left( \mu \right) \right) -E\left(
		I_{X_{a_{i}}}^{\left( L_{i}\right) }\left( \lambda \right) \right) E\left(
		I_{X_{a_{j}}}^{\left( L_{j}\right) }\left( \mu \right) \right)  \\
		&=&E\left( D_{X_{a_{i}}}^{\left( L_{i}\right) }\left( \lambda \right)
		D_{X_{a_{i}}}^{\left( L_{i}\right) }\left( -\lambda \right)
		D_{X_{a_{j}}}^{\left( L_{j}\right) }\left( \mu \right) D_{X_{a_{j}}}^{\left(
			L_{j}\right) }\left( -\mu \right) \right)  \\
		&&-E\left( D_{X_{a_{i}}}^{\left( L_{i}\right) }\left( \lambda \right)
		D_{X_{a_{i}}}^{\left( L_{i}\right) }\left( -\lambda \right) \right) E\left(
		D_{X_{a_{j}}}^{\left( L_{j}\right) }\left( \mu \right) D_{X_{a_{j}}}^{\left(
			L_{j}\right) }\left( -\mu \right) \right) .
	\end{eqnarray*}%
	Since the expectations can be expressed in terms of cumulants (see Appendix
	B for more details), we may express the covariance term as follows,%
	\begin{eqnarray*}
		Cov\left( I_{X_{a_{i}}}^{\left( L_{i}\right) }\left( \lambda \right)
		,I_{X_{a_{j}}}^{\left( L_{j}\right) }\left( \mu \right) \right)  &=&\kappa
		\left( D_{X_{a_{i}}}^{\left( L_{i}\right) }\left( \lambda \right)
		,D_{X_{a_{i}}}^{\left( L_{i}\right) }\left( -\lambda \right)
		,D_{X_{a_{j}}}^{\left( L_{j}\right) }\left( \mu \right)
		,D_{X_{a_{j}}}^{\left( L_{j}\right) }\left( -\mu \right) \right)  \\
		&&+ \ \kappa \left( D_{X_{a_{i}}}^{\left( L_{i}\right) }\left( -\lambda \right)
		,D_{X_{a_{j}}}^{\left( L_{j}\right) }\left( \mu \right) \right) \kappa
		\left( D_{X_{a_{i}}}^{\left( L_{i}\right) }\left( \lambda \right)
		,D_{X_{a_{j}}}^{\left( L_{j}\right) }\left( -\mu \right) \right)  \\
		&&+ \ \kappa \left( D_{X_{a_{i}}}^{\left( L_{i}\right) }\left( \lambda \right)
		,D_{X_{a_{j}}}^{\left( L_{j}\right) }\left( \mu \right) \right) \kappa
		\left( D_{X_{a_{i}}}^{\left( L_{i}\right) }\left( -\lambda \right)
		,D_{X_{a_{j}}}^{\left( L_{j}\right) }\left( -\mu \right) \right) .
	\end{eqnarray*}%
	Then Proposition 1 gives us that,%
	\begin{align}
	Cov\left( I_{X_{a_{i}}}^{\left( L_{i}\right) }\left( \lambda \right)
	,I_{X_{a_{j}}}^{\left( L_{j}\right) }\left( \mu \right) \right) 
	&=L^{-2}\left( 2\pi \right) \Delta ^{\left( L\right) }\left( \lambda +\mu
	-\lambda -\mu \right) f_{X_{a_{i}}X_{a_{i}}X_{a_{j}}X_{a_{j}}}\left( \lambda
	,-\lambda ,\mu \right) +o\left( L^{-1-2d}\right)   \notag \\
	& \ \ \ +\left( L^{-1}\Delta ^{\left( L\right) }\left( -\lambda +\mu \right)
	f_{X_{a_{i}}X_{a_{j}}}\left( -\lambda \right) +o\left( L^{-2d}\right)
	\right)   \notag \\
	& \ \ \ \times \left( L^{-1}\Delta ^{\left( L\right) }\left( \lambda -\mu \right)
	f_{X_{a_{i}}X_{a_{j}}}\left( \lambda \right) +o\left( L^{-2d}\right) \right) 
	\notag \\
	& \ \ \ +\left( L^{-1}\Delta ^{^{\left( L\right) }}\left( \lambda +\mu \right)
	f_{X_{a_{i}}X_{a_{j}}}\left( \lambda \right) +o\left( L^{-2d}\right) \right) 
	\notag \\
	& \ \ \ \times \left( L^{-1}\Delta ^{\left( L\right) }\left( -\lambda -\mu \right)
	f_{X_{a_{i}}X_{a_{j}}}\left( -\lambda \right) +o\left( L^{-2d}\right)
	\right)   \notag \\
	&=L^{-2}\left( 2\pi \right) \Delta ^{\left( L\right) }\left( 0\right)
	f_{X_{a_{i}}X_{a_{i}}X_{a_{j}}X_{a_{j}}}\left( \lambda ,-\lambda ,\mu
	\right) +o\left( L^{-1-2d}\right)   \notag \\
	& \ \ \ +L^{-2}\Delta ^{\left( L\right) }\left( -\lambda +\mu \right) \Delta
	^{\left( L\right) }\left( \lambda -\mu \right) \left(
	f_{X_{a_{i}}X_{a_{j}}}\left( \lambda \right) \right) ^{2}  \notag \\
	& \ \ \ +L^{-1}\left( \Delta ^{\left( L\right) }\left( -\lambda +\mu \right)
	+\Delta ^{\left( L\right) }\left( \lambda -\mu \right) \right)
	f_{X_{a_{i}}X_{a_{j}}}\left( \lambda \right) o\left( L^{-2d}\right)   \notag
	\\
	& \ \ \ +L^{-2}\Delta ^{\left( L\right) }\left( \lambda +\mu \right) \Delta
	^{\left( L\right) }\left( -\lambda -\mu \right) \left(
	f_{X_{a_{i}}X_{a_{j}}}\left( \lambda \right) \right) ^{2}  \notag \\
	& \ \ \ +L^{-1}\Delta ^{\left( L\right) }\left( \lambda +\mu \right)
	f_{X_{a_{i}}X_{a_{j}}}\left( -\lambda \right) +\Delta ^{\left( L\right)
	}\left( -\lambda -\mu \right) f_{X_{a_{i}}X_{a_{j}}}\left( -\lambda \right)
	o\left( L^{-2d}\right)   \notag \\
	&=L^{-1}\left( 2\pi \right) f_{X_{a_{i}}X_{a_{i}}X_{a_{j}}X_{a_{j}}}\left(
	\lambda ,-\lambda ,\mu \right) +L^{-2}\left[ \Delta ^{\left( L\right)
	}\left( -\lambda +\mu \right) \Delta ^{\left( L\right) }\left( \lambda -\mu
	\right) \right.   \notag \\
	& \ \ \ +\left. \Delta ^{(L)}\left( \lambda +\mu \right) \Delta ^{\left( L\right)
	}\left( -\lambda -\mu \right) \right] \left( f_{X_{a_{i}}X_{a_{j}}}\left(
	\lambda \right) \right) ^{2}+\left[ \Delta ^{\left( L\right) }\left(
	-\lambda +\mu \right) \right.   \notag \\
	& \ \ \ +\left. \Delta ^{\left( L\right) }\left( \lambda -\mu \right) +\Delta
	^{\left( L\right) }\left( \lambda +\mu \right) +\Delta ^{\left( L\right)
	}\left( -\lambda -\mu \right) \right] f_{X_{a_{i}}X_{a_{j}}}\left( \lambda
	\right) o\left( L^{-2d}\right)   \notag \\
	& \ \ \ +o\left( L^{-1-2d}\right) +o\left( L^{-4d}\right) .  \tag{A.7}
	\label{Covariance between periodograms step1}
	\end{align}%
	Using the two properties in (\ref{Sum of deltas}) and (\ref{Product of
		deltas}), the covariance in (\ref{Covariance between periodograms step1}) is
	simplified further as follows,%
	\begin{eqnarray*}
		Cov\left( I_{X_{a_{i}}}^{\left( L_{i}\right) }\left( \lambda \right)
		,I_{X_{a_{j}}}^{\left( L_{j}\right) }\left( \mu \right) \right)  &=&\frac{%
			2\pi }{L}\left[ \eta \left( \lambda -\mu \right) +\eta \left( \lambda +\mu
		\right) \right] \left\{ f_{X_{a_{i}}X_{a_{j}}}\left( \lambda \right)
		\right\} ^{2}+\frac{2\pi }{l^{\dag }}%
		f_{X_{a_{i}}X_{a_{i}}X_{a_{j}}X_{a_{j}}}\left( \lambda ,-\lambda ,\mu
		\right)  \\
		&&+ \ 2\pi \left[ \eta \left( \lambda -\mu \right) +\eta \left( \lambda +\mu
		\right) \right] f_{X_{a_{i}}X_{a_{j}}}\left( \lambda \right) o\left( l^{\dag
			^{-2d}}\right) +o\left( L^{-1-2d}\right) .
	\end{eqnarray*}
	
	Now let us consider the asymptotic distribution of $I_{X_{a_{i}}}^{\left(
		L_{i}\right) }\left( \lambda \right) .$ We may re-write the periodogram as
	follows,%
	\begin{equation*}
	I_{X_{a_{i}}}^{\left( L_{i}\right) }\left( \lambda \right) =\left[ \textit{Re}%
	D_{X_{a_{i}}}^{\left( L_{i}\right) }\left( \lambda \right) \right] ^{2}+%
	\left[ \textit{Im}D_{X_{a_{i}}}^{\left( L_{i}\right) }\left( \lambda \right) %
	\right] ^{2},
	\end{equation*}%
	where%
	\begin{equation*}
	\textit{Re}D_{X_{a_{i}}}^{\left( L_{i}\right) }\left( \lambda \right) =\tfrac{1%
	}{\sqrt{2\pi L_{i}}}\sum_{t=1}^{L_{i}}y_{t}\cos \left( \lambda t\right) ,%
	\text{ and, }\textit{Im}D_{X_{a_{i}}}^{\left( L_{i}\right) }\left( \lambda
	\right) =\tfrac{1}{\sqrt{2\pi L_{i}}}\sum_{t=1}^{L_{i}}y_{t}\sin \left(
	\lambda t\right) .
	\end{equation*}%
	Following Theorem $2.1$ of\ \color{blue}\cite{Lahiri2003}\color{black}, we
	have that%
	\begin{equation*}
	\left[ 
	\begin{array}{c}
	\dfrac{\textit{Re}D_{X_{a_{i}}}^{\left( L_{i}\right) }\left( \lambda \right)
		-E\left( \textit{Re}D_{X_{a_{i}}}^{\left( L_{i}\right) }\left( \lambda \right)
		\right) }{\sqrt{L_{i}f_{X_{a_{i}}X_{a_{i}}}\left( \lambda \right) }} \\ 
	\dfrac{\textit{Im}D_{X_{a_{i}}}^{\left( L_{i}\right) }\left( \lambda \right)
		-E\left( \textit{Im}D_{X_{a_{i}}}^{\left( L_{i}\right) }\left( \lambda \right)
		\right) }{\sqrt{L_{i}f_{X_{a_{i}}X_{a_{i}}}\left( \lambda \right) }}%
	\end{array}%
	\right] \rightarrow ^{D}N\left( \mathbf{0,I}_{2}\right) .
	\end{equation*}%
	Hence the result.
\end{proof}

\begin{proof}[Proof of Theorem 4]
	Recall that $x_{j}=\ln (2\sin \left( \lambda _{j}/2\right) ),$ $a_{j}=x_{j}-%
	\overline{x}$ and $S_{xx}=\sum\limits_{j=1}^{N_{n}}\left( X_{j}-\overline{X}%
	\right) ^{2}.$ From\ \color{blue}\cite{HurvichDeoBrodshky1998}\color{black}\
	we have that $S_{xx}=N_{n}\left( 1+o\left( 1\right) \right) $ and $%
	a_{j}=\log j-\log N_{n}+1+o\left( 1\right) +o\left( \frac{N_{n}^{2}}{n^{2}}%
	\right) ,$ $j=1,\ldots ,N_{n}.$ Thus,%
	\begin{equation*}
	\sup_{j}\left\vert a_{j}\right\vert =1+o\left( 1\right) +O\left( \frac{%
		N_{n}^{2}}{n^{2}}\right) .
	\end{equation*}%
	Using Appendix B we have that%
	\begin{eqnarray*}
		Cov\left( \log I_{X_{a_{i}}}^{\left( L_{i}\right) }\left( \lambda
		_{j}\right) ,\log I_{X_{a_{j}}}^{\left( L_{j}\right) }\left( \mu _{k}\right)
		\right)  &=&\left( 1-\rho ^{2}\right) ^{\frac{1}{2}}\sum_{k=1}^{\infty
		}\left( \Psi \left( \frac{1}{2}+k\right) +\Psi \left( \frac{1}{2}\right)
		\right) ^{2}\frac{\Gamma \left( \frac{1}{2}+k\right) }{\Gamma \left( \frac{1%
			}{2}\right) }\frac{\left( \rho ^{2}\right) ^{k}}{k!} \\
		&&-\left( 1-\rho ^{2}\right) \left( \sum_{k=1}^{\infty }\left( \Psi \left( 
		\frac{1}{2}+k\right) +\Psi \left( \frac{1}{2}\right) \right) \frac{\Gamma
			\left( \frac{1}{2}+k\right) }{\Gamma \left( \frac{1}{2}\right) }\frac{\left(
			\rho ^{2}\right) ^{k}}{k!}\right) ^{2} \\
		&\leq &\left( 1-\rho ^{2}\right) ^{\frac{1}{2}}\sum_{k=1}^{\infty }\left(
		\Psi \left( \frac{1}{2}+k\right) +\Psi \left( \frac{1}{2}\right) \right) ^{2}%
		\frac{\Gamma \left( \frac{1}{2}+k\right) }{\Gamma \left( \frac{1}{2}\right) }%
		\frac{\left( \rho ^{2}\right) ^{k}}{k!},
	\end{eqnarray*}%
	where $\rho =Corr\left( I_{X_{a_{i}}}^{\left( L_{i}\right) }\left( \lambda
	_{j}\right) ,I_{X_{a_{j}}}^{\left( L_{j}\right) }\left( \mu _{k}\right)
	\right) =o\left( n^{-1}\right) $ by Theorem 3. Thus,%
	\begin{equation*}
	Cov%
	\Big(%
	\log I_{X_{a_{i}}}^{\left( L_{i}\right) }\left( \lambda _{j}\right) ,\log
	I_{X_{a_{j}}}^{\left( L_{j}\right) }\left( \mu _{k}\right) 
	\Big)%
	=o\left( n^{-1}\right) .
	\end{equation*}%
	This leads to%
	\begin{eqnarray*}
		Cov\left( \widehat{d}_{n},\widehat{d}_{i}\right)  &=&\frac{1}{4S_{xx}}\frac{1%
		}{S_{xx}^{^{\prime }}}\sum\limits_{j=1}^{N_{n}}\sum%
		\limits_{k=1}^{N_{l}}a_{j}a_{k}^{(i)}Cov\left( \log I_{X_{a_{i}}}^{\left(
			L_{i}\right) }\left( \lambda _{j}\right) ,\log I_{X_{a_{j}}}^{\left(
			L_{j}\right) }\left( \mu _{k}\right) \right)  \\
		&\leq &\sup_{j,k}\frac{1}{4S_{xx}}\frac{1}{S_{xx}^{^{\prime }}}%
		N_{n}N_{l}\left\vert a_{j}a_{k}^{(i)}Cov\left( \log I_{X_{a_{i}}}^{\left(
			L_{i}\right) }\left( \lambda _{j}\right) ,\log I_{X_{a_{j}}}^{\left(
			L_{j}\right) }\left( \mu _{k}\right) \right) \right\vert  \\
		&=&\frac{\left( 1+o\left( 1\right) \right) ^{-2}}{4}\sup_{j,k}\left\vert
		a_{j}\right\vert \left\vert a_{k}^{(i)}\right\vert \left\vert Cov\left( \log
		I_{X_{a_{i}}}^{\left( L_{i}\right) }\left( \lambda _{j}\right) ,\log
		I_{X_{a_{j}}}^{\left( L_{j}\right) }\left( \mu _{k}\right) \right)
		\right\vert  \\
		&=&\frac{\left( 1+o\left( 1\right) \right) ^{-2}}{4}\left( 1+o\left(
		1\right) +O\left( \frac{N_{n}^{2}}{n^{2}}\right) \right)
		^{2}\sup_{j,k}\left\vert Cov\left( \log I_{X_{a_{i}}}^{\left( L_{i}\right)
		}\left( \lambda _{j}\right) ,\log I_{X_{a_{j}}}^{\left( L_{j}\right) }\left(
		\mu _{k}\right) \right) \right\vert  \\
		&=&o\left( n^{-1}\right) .
	\end{eqnarray*}%
	Similarly, we can prove that $Cov\left( \widehat{d}_{i},\widehat{d}%
	_{j}\right) =o\left( n^{-1}\right) .$ Hence the result.
\end{proof}

\begin{proof}[Proof of Theorem 5]
	Consider,%
	\begin{equation}
	\left( \widehat{d}_{J,m}^{Opt}-d_{0}\right) =w_{n}^{\ast }\left( \widehat{d}%
	_{n}-d_{0}\right) -\sum_{i=1}^{m}w_{i}^{\ast }\left( \widehat{d}%
	_{i,m}-d_{0}\right) .  \tag{A.8}  \label{dhat-d}
	\end{equation}%
	Recall that $w_{n}^{\ast }=\left[ 1-\left( \frac{1}{m}\frac{N_{n}}{n}\frac{l%
	}{N_{l}}\right) ^{2}\right] ^{-1}$ and $\sum_{i=1}^{m}w_{i}^{\ast
	}=w_{n}^{\ast }-1;$ for $i=1,\ldots ,m$. Let us firstly consider $%
	w_{n}^{\ast }$. For fixed $m$ and for the choice of $N_{n}$ such that $%
	N_{n}\log N_{n}/n\rightarrow 0$,%
	\begin{equation}
	w_{n}^{\ast }=\frac{1}{1-\left( n^{-1}ln^{-1+\alpha }l^{1-\alpha }\right)
		^{2}}=1+o\left( 1\right) ,  \tag{A.9}  \label{limit of wn}
	\end{equation}%
	and hence%
	\begin{equation}
	\sum_{i=1}^{m}w_{i}^{\ast }=o\left( 1\right) ,  \tag{A.10}
	\label{limit of wi,m}
	\end{equation}%
	with $w_{i}^{\ast }\rightarrow 0$ as $n\rightarrow \infty $ (see the proof
	of Theorem 4).
	
	By virtue of the consistency of $\widehat{d}_{n},$ we have that the first
	term in (\ref{dhat-d}) such that $w_{n}^{\ast }\left( \widehat{d}%
	_{n}-d\right) =o_{p}\left( 1\right) $, using (\ref{limit of wn}).
	
	Now, we show that the second term in (\ref{dhat-d}) is $o_{p}\left( 1\right)
	.$%
	\begin{eqnarray*}
		\lim_{n\rightarrow \infty }\Pr \left[ \left\vert \sum_{i=1}^{m}w_{i}^{\ast
		}\left( \widehat{d}_{i}-d_{0}\right) \right\vert \geq \varepsilon \right] 
		&\leq &\lim_{n\rightarrow \infty }\frac{E\left( \sum_{i=1}^{m}w_{i}^{\ast
			}\left( \widehat{d}_{i}-d_{0}\right) \right) ^{2}}{\varepsilon ^{2}} \\
		&=&\lim_{n\rightarrow \infty }\frac{Var\left( \widehat{d}_{i}\right) }{%
			\varepsilon ^{2}}\sum_{i=1}^{m}\left( w_{i}^{\ast }\right) ^{2} \\
		&&+\frac{2}{\varepsilon ^{2}}\lim_{n\rightarrow \infty
		}\sum_{i=1}^{m}\sum_{j=i+1}^{m}w_{i}^{\ast }w_{j}^{\ast }Cov\left( \widehat{d%
		}_{i},\widehat{d}_{j}\right)  \\
		&=&0,
	\end{eqnarray*}%
	since $\lim_{n\rightarrow \infty }Var\left( \widehat{d}_{i}\right) =0$ from
	Theorem 1, $\lim_{n\rightarrow \infty }Cov\left( \widehat{d}_{i},\widehat{d}%
	_{j}\right) =0$ directly from Theorem 2 and the limit of $%
	\sum_{i=1}^{m}w_{i}^{\ast }$ given in (\ref{limit of wi,m}). This completes
	the proof of consistency.
	
	The proof of asymptotic normality of the optimal jackknife estimator depends
	on the joint convergence of $\widehat{d}_{n}$ and $\widehat{d}_{i,m}.$
	Firstly, let us consider the following standardized optimal jackknife
	estimator,%
	\begin{equation}
	\sqrt{N_{n}}\left( \widehat{d}_{J,m}^{Opt}-d_{0}\right) =w_{n}^{\ast }\sqrt{%
		N_{n}}\left( \widehat{d}_{n}-d_{0}\right) -\sum_{i=1}^{m}w_{i}^{\ast }\sqrt{%
		N_{n}}\left( \widehat{d}_{i}-d_{0}\right) .  \tag{A.11}
	\label{asymptotic normality}
	\end{equation}%
	Using Theorem 1 we have that $\sqrt{N_{n}}\left( \widehat{d}%
	_{n}-d_{0}\right) \rightarrow ^{D}N\left( 0,\frac{\pi ^{2}}{24}\right) .$
	Therefore, regarding the first component in (\ref{asymptotic normality}), it
	immediately follows that%
	\begin{equation*}
	w_{n}^{\ast }\sqrt{N_{n}}\left( \widehat{d}_{n}-d_{0}\right) \rightarrow
	^{d}N\left( 0,\frac{\pi ^{2}}{24}\right) ,\text{ using (\ref{limit of wn})}.
	\end{equation*}
	
	Now, let us consider the second term in (\ref{asymptotic normality}):%
	\begin{align}
	\lim_{n\rightarrow \infty }\Pr \left[ \left\vert \sum_{i=1}^{m}w_{i}^{\ast }%
	\sqrt{N_{n}}\left( \widehat{d}_{i}-d_{0}\right) \right\vert \geq \varepsilon %
	\right]  &\leq \lim_{n\rightarrow \infty }\frac{E\left(
		\sum_{i=1}^{m}w_{i}^{\ast }\left( \widehat{d}_{i}-d_{0}\right) \right) ^{2}}{%
		\varepsilon ^{2}}N_{n}  \notag \\
	&=\lim_{n\rightarrow \infty }\frac{Var\left( \widehat{d}_{i}\right) }{%
		\varepsilon ^{2}}N_{n}\sum_{i=1}^{m}\left( w_{i}^{\ast }\right) ^{2}  \notag
	\\
	& \ \ \ +\lim_{n\rightarrow \infty }\frac{2N_{n}}{\varepsilon ^{2}}%
	\sum_{i=1}^{m}\sum_{j=i+1}^{m}w_{i}^{\ast }w_{j}^{\ast }Cov\left( \widehat{d}%
	_{i},\widehat{d}_{j}\right) .  \tag{A.12}
	\label{asymptotic normality 2nd term}
	\end{align}%
	By considering the first term in (\ref{asymptotic normality 2nd term}), for
	fixed $m$ we have that%
	\begin{equation*}
	\lim_{n\rightarrow \infty }\frac{Var\left( \widehat{d}_{i}\right) }{%
		\varepsilon ^{2}}N_{n}\sum_{i=1}^{m}\left( w_{i}^{\ast }\right)
	^{2}=\lim_{n\rightarrow \infty }\frac{\sum_{i=1}^{m}\left( w_{i}^{\ast
		}\right) ^{2}}{\varepsilon ^{2}}\left[ \frac{\pi ^{2}}{24}+o\left( 1\right) %
	\right] =0,
	\end{equation*}%
	using Theorem 1 and (\ref{limit of wn}). The second term in (\ref{asymptotic
		normality 2nd term}) would give us that,%
	\begin{equation*}
	\lim_{n\rightarrow \infty }\frac{2N_{n}}{\varepsilon ^{2}}%
	\sum_{i=1}^{m}\sum_{j=i+1}^{m}w_{i}^{\ast }w_{j}^{\ast }Cov\left( \widehat{d}%
	_{i},\widehat{d}_{j}\right) =0,
	\end{equation*}%
	immediately from (\ref{limit of wn}). Therefore, $\Pr \left[ \left\vert
	\sum_{i=1}^{m}w_{i}^{\ast }\sqrt{N_{l}}\left( \widehat{d}_{i}-d_{0}\right)
	\right\vert \geq \varepsilon \right] \rightarrow 0$ as $n\rightarrow \infty .
	$ Hence the proof completes.
\end{proof}

\section*{Appendix B: Additional technical results\label{Appendix:B}}

Recall that the covariance between the full-sample LPR estimator and each
sub-sample LPR estimator, $Cov\left( \widehat{d}_{n},\widehat{d}_{i}\right) ,
$ and the covariances\textbf{\ }between the different sub-sample LPR
estimators, $Cov\left( \widehat{d}_{i},\widehat{d}_{j}\right) ,$ for $i\neq
j,$ $i,j=1,2,...,m,$ are given respectively by,%
\begin{align}
Cov%
\big(%
\widehat{d}_{n},\widehat{d}_{i}%
\big)
&=\frac{1}{4S_{xx}}\frac{1}{S_{xx}^{^{\prime }}}\sum\limits_{j=1}^{N_{n}}%
\sum\limits_{k=1}^{N_{l}}a_{j}a_{k}^{(i)}Cov%
\big(%
\log I_{Y}^{\left( n\right) }\left( \lambda _{j}\right) ,\log
I_{Y_{i}}^{\left( l\right) }\left( \mu _{k}\right) 
\big)
\tag{B.1}  \label{Covariance(dhat,dhat_i)} \\
Cov%
\big(%
\widehat{d}_{i},\widehat{d}_{i^{\prime }}%
\big)
&=\frac{1}{4}\frac{1}{\left( S_{xx}^{^{\prime }}\right) ^{2}}%
\sum\limits_{j=1}^{N_{l}}\sum\limits_{k=1}^{N_{l}}a_{j}^{\prime
}a_{k}^{\prime }Cov%
\big(%
\log I_{Y_{i}}^{\left( l\right) }\left( \mu _{j}\right) ,\log
I_{Y_{i^{\prime }}}^{\left( l\right) }\left( \mu _{k}\right) 
\big)%
,  \tag{B.2}  \label{Covariance(dhat_i,dhat_j)}
\end{align}%
with all notation as defined in Table 1.

\subsection*{Evaluation of the covariance terms in (\protect\ref{Covariance(dhat,dhat_i)}) and (\protect\ref{Covariance(dhat_i,dhat_j)})}

The main purpose of this exercise is to calculate the covariances between
the full-sample and sub-sample LPR estimators (refer to (\ref%
{Covariance(dhat,dhat_i)})) and the covariance between two distinct
sub-sample LPR estimators (refer to (\ref{Covariance(dhat_i,dhat_j)})).
These covariance terms depend on the covariance between the log-periodograms
associated with either the full sample and a given sub-sample or two
different sub-samples.

To obtain the covariance between the log-periodograms associated with the
full sample and a given sub-sample, or between sub-samples, we follow the
method stated below.
\begin{enumerate}
	\item[Step 1:] Write down the joint distribution of the periodograms $%
	(I_{X_{a_{i}}}^{\left( L_{i}\right) }\left( \lambda \right)
	,I_{X_{a_{j}}}^{\left( L_{j}\right) }\left( \mu \right) )$.
	\item[Step 2:] Write down the joint distribution of the log transformed
	periodograms $(\log I_{X_{a_{i}}}^{\left( L_{i}\right) }\left( \lambda
	\right) ,$ $\log I_{X_{a_{j}}}^{\left( L_{j}\right) }\left( \mu \right) )$
	using the expression of the covariance between the two different
	periodograms.	
	\item[Step 3:] Find the expression for the covariance between the above
	mentioned log-periodograms, $Cov(\log I_{X_{a_{i}}}^{\left( L_{i}\right)
	}\left( \lambda \right) ,$ $\log I_{X_{a_{j}}}^{\left( L_{j}\right) }\left(
	\mu \right) )$, using the moment generating function.
\end{enumerate}

\textbf{In relation to Step 1: }Using the results of Theorem 3, we can say
that the periodograms associated with the full sample and the sub-sample
have a limiting distribution of the form $f_{X_{1}X_{1}}(\lambda )\left.
\chi _{\left( 2\right) }^{2}\right/ 2$. For notational convenience, let us
denote by $\left( U,V\right) $ the bivariate $\chi _{k}^{2}$ random
variables, $(I_{X_{a_{i}}}^{\left( L_{i}\right) }\left( \lambda \right)
,I_{X_{a_{j}}}^{\left( L_{j}\right) }\left( \mu \right) )$. Although $k=2$,
we use the generic notation for the degrees of freedom, $k$. Note that we
ignore the constant term $\left. f_{X_{1}X_{1}}(\lambda )\right/ 2$ for
convenience, as these terms will disappear in the calculation of the
covariance between two different LPR estimators (either the full- and
sub-sample LPR estimators or two distinct sub-sample LPR estimators).

The joint probability density function (pdf), $f_{U,V}\left( u,v\right) ,$
is defined by %
\citep[see,\color{blue}][]{Krishnaiah1963}
\begin{equation*}
f_{U,V}\left( u,v\right) =\left( 1-\rho ^{2}\right) ^{\frac{k-1}{2}%
}\sum_{i=0}^{\infty }\frac{\Gamma \left( \frac{k-1}{2}+i\right) \rho
	^{2i}\left( uv\right) ^{\frac{k-3+2i}{2}}\exp \left[ -\frac{u+v}{2\left(
		1-\rho ^{2}\right) }\right] }{\Gamma \left( \frac{k-1}{2}\right) i!\left[ 2^{%
		\frac{k-1}{2}+i}\Gamma \left( \frac{k-1}{2}+i\right) \left( 1-\rho
	^{2}\right) ^{\frac{k-1}{2}+i}\right] ^{2}},
\end{equation*}%
where $\rho =\dfrac{\sigma _{uv}}{\sigma _{u}\sigma _{v}}.$ Here, $\sigma
_{uv}=cov\left( U,V\right) .$ Then, the marginal densities of $U$ and $V,$ $%
f_{U}\left( u\right) $ and $f_{V}\left( v\right) $, are respectively given
by,%
\begin{equation*}
f_{U}\left( u\right) =\frac{1}{2^{\frac{k}{2}}\Gamma \left( \frac{k}{2}%
	\right) }u^{\frac{k}{2}}\exp \left\{ -\frac{u}{2}\right\} ,\text{ and, }%
f_{V}\left( v\right) =\frac{1}{2^{\frac{k}{2}}\Gamma \left( \frac{k}{2}%
	\right) }v^{\frac{k}{2}}\exp \left\{ -\frac{v}{2}\right\} .
\end{equation*}

\textbf{In relation to Step 2: }Let $W=\log U=$ $\log I_{X_{a_{i}}}^{\left(
	L_{i}\right) }\left( \lambda \right) $ and $Z=\log V=\log
I_{X_{a_{j}}}^{\left( L_{j}\right) }\left( \mu \right) .$ Then, the joint
pdf of $W$ and $Z$ is given by,%
\begin{eqnarray*}
	f_{W,Z}\left( w,z\right) &=&f_{U,V}\left( \exp w,\exp z\right) \left\vert 
	\begin{array}{cc}
		\frac{\partial \exp w}{\partial w} & \frac{\partial \exp w}{\partial z} \\ 
		\frac{\partial \exp z}{\partial w} & \frac{\partial \exp z}{\partial z}%
	\end{array}%
	\right\vert \\
	&=&\left( 1-\rho ^{2}\right) ^{\frac{k-1}{2}}\sum_{i=0}^{\infty }\frac{%
		\Gamma \left( \frac{k-1}{2}+i\right) \rho ^{2i}\left( \exp w\exp z\right) ^{%
			\frac{k-3+2i}{2}}\exp \left[ -\frac{\exp w+\exp z}{2\left( 1-\rho
			^{2}\right) }\right] }{\Gamma \left( \frac{k-1}{2}\right) i!\left[ 2^{\frac{%
				k-1}{2}+i}\Gamma \left( \frac{k-1}{2}+i\right) \left( 1-\rho ^{2}\right) ^{%
			\frac{k-1}{2}+i}\right] ^{2}}\exp w\exp z \\
	&=&\left( 1-\rho ^{2}\right) ^{\frac{k-1}{2}}\sum_{i=0}^{\infty }\frac{%
		\Gamma \left( \frac{k-1}{2}+i\right) \rho ^{2i}\exp \left( \frac{k-1}{2}%
		+i\right) \left( w+z\right) \exp \left[ -\frac{\exp w+\exp z}{2\left( 1-\rho
			^{2}\right) }\right] }{\Gamma \left( \frac{k-1}{2}\right) i!\left[ 2^{\frac{%
				k-1}{2}+i}\Gamma \left( \frac{k-1}{2}+i\right) \left( 1-\rho ^{2}\right) ^{%
			\frac{k-1}{2}+i}\right] ^{2}}.
\end{eqnarray*}

\textbf{In relation to Step 3: }The moment generating function (MGF) of $%
(W,Z)$ is given by,%
\begin{align}
M_{W,Z}\left( t_{1},t_{2}\right)  &=E\left( \exp \left(
t_{1}W+t_{2}Z\right) \right) =\int_{0}^{\infty }\int_{0}^{\infty }\exp
\left( t_{1}w+t_{2}z\right) f_{W,Z}\left( w,z\right) dwdz  \notag \\
&=\left( 1-\rho ^{2}\right) ^{\frac{k-1}{2}}\sum_{i=0}^{\infty }\frac{%
	\Gamma \left( \frac{k-1}{2}+i\right) \rho ^{2i}}{\Gamma \left( \frac{k-1}{2}%
	\right) i!\left[ 2^{\frac{k-1}{2}+i}\Gamma \left( \frac{k-1}{2}+i\right)
	\left( 1-\rho ^{2}\right) ^{\frac{k-1}{2}+i}\right] ^{2}}  \notag \\
& \ \ \ \times \int_{0}^{\infty }\int_{0}^{\infty }\exp \left(
t_{1}w+t_{2}z\right) \exp \left( \tfrac{k-1}{2}+i\right) \left( w+z\right)
\exp \left[ -\tfrac{\exp w+\exp z}{2\left( 1-\rho ^{2}\right) }\right] dwdz 
\notag \\
&=\left( 1-\rho ^{2}\right) ^{\frac{k-1}{2}}\sum_{i=0}^{\infty }\frac{%
	\Gamma \left( \frac{k-1}{2}+i\right) \rho ^{2i}}{\Gamma \left( \frac{k-1}{2}%
	\right) i!\left[ 2^{\frac{k-1}{2}+i}\Gamma \left( \frac{k-1}{2}+i\right)
	\left( 1-\rho ^{2}\right) ^{\frac{k-1}{2}+i}\right] ^{2}}  \notag \\
& \ \ \ \times \int_{0}^{\infty }\exp \left( \tfrac{k-1}{2}+t_{1}+i\right) w\exp %
\left[ -\tfrac{\exp w}{2\left( 1-\rho ^{2}\right) }\right] dw  \notag \\
& \ \ \ \times \int_{0}^{\infty }\exp \left( \tfrac{k-1}{2}+t_{2}+i\right) z\exp %
\left[ -\tfrac{\exp z}{2\left( 1-\rho ^{2}\right) }\right] dz.
\tag{B.3} \label{MGF step}
\end{align}

Now let us consider the form of the last expression in (\ref{MGF step}). Let 
$\alpha _{1}=\tfrac{k-1}{2}+t_{2}+i$ and $\alpha _{2}=\frac{1}{2\left(
	1-\rho ^{2}\right) }.$ Then, substituting $x=\exp z$ would give us that 
\begin{equation}
\int_{0}^{\infty }\exp \alpha _{1}z\exp \left[ -\alpha _{2}\exp z\right]
dz=\int_{0}^{\infty }x^{\alpha _{1}-1}\exp \left[ -\alpha _{2}x\right] dx=%
\frac{\Gamma \left( \alpha _{1}\right) }{\alpha _{2}^{\alpha _{1}}}. 
\tag{B.4}  \label{Integral of exp of exp function}
\end{equation}

Therefore, using (\ref{Integral of exp of exp function}), the MGF given in (%
\ref{MGF step}) may be re-arranged as follows,%
\begin{eqnarray*}
	M_{W,Z}\left( t_{1},t_{2}\right) &=&\left[ 2\left( 1-\rho ^{2}\right) \right]
	^{t_{1}+t_{2}}\left( 1-\rho ^{2}\right) ^{\frac{k-1}{2}}\sum_{i=0}^{\infty }%
	\frac{\Gamma \left( \frac{k-1}{2}+i\right) \rho ^{2i}\Gamma \left( \frac{k-1%
		}{2}+t_{2}+i\right) \Gamma \left( \frac{k-1}{2}+t_{1}+i\right) }{i!\Gamma
		\left( \frac{k-1}{2}\right) \left[ \Gamma \left( \frac{k-1}{2}+i\right) %
		\right] ^{2}} \\
	&=&\left[ 2\left( 1-\rho ^{2}\right) \right] ^{t_{1}+t_{2}}\left( 1-\rho
	^{2}\right) ^{\frac{k-1}{2}}\frac{\Gamma \left( \frac{k-1}{2}+t_{1}\right)
		\Gamma \left( \frac{k-1}{2}+t_{2}\right) }{\left[ \Gamma \left( \frac{k-1}{2}%
		\right) \right] ^{2}} \\
	&&\times \sum_{i=0}^{\infty }\frac{\Gamma \left( \frac{k-1}{2}%
		+t_{1}+i\right) \Gamma \left( \frac{k-1}{2}+t_{2}+i\right) \Gamma \left( 
		\frac{k-1}{2}\right) }{\Gamma \left( \frac{k-1}{2}+t_{1}\right) \Gamma
		\left( \frac{k-1}{2}+t_{2}\right) \Gamma \left( \frac{k-1}{2}+i\right) }%
	\frac{\left( \rho ^{2}\right) ^{i}}{i!} \\
	&=&\left[ 2\left( 1-\rho ^{2}\right) \right] ^{t_{1}+t_{2}}\left( 1-\rho
	^{2}\right) ^{\frac{k-1}{2}}\frac{\Gamma \left( \frac{k-1}{2}+t_{1}\right)
		\Gamma \left( \frac{k-1}{2}+t_{2}\right) }{\left[ \Gamma \left( \frac{k-1}{2}%
		\right) \right] ^{2}} \\
	&&\times _{2}F_{1}\left( \tfrac{k-1}{2}+t_{1},\tfrac{k-1}{2}+t_{2};\tfrac{k-1%
	}{2};\rho ^{2}\right) .
\end{eqnarray*}

Setting $k=2$ gives,%
\begin{equation*}
M_{W,Z}\left( t_{1},t_{2}\right) =\left[ 2\left( 1-\rho ^{2}\right) \right]
^{t_{1}+t_{2}}\left( 1-\rho ^{2}\right) ^{\frac{1}{2}}\frac{\Gamma \left( 
	\frac{1}{2}+t_{1}\right) \Gamma \left( \frac{1}{2}+t_{2}\right) }{\left[
	\Gamma \left( \frac{1}{2}\right) \right] ^{2}}\text{ }_{2}F_{1}\left( \tfrac{%
	1}{2}+t_{1},\tfrac{1}{2}+t_{2};\tfrac{1}{2};\rho ^{2}\right) .
\end{equation*}%
Therefore the cumulant generating function is given by $K\left(
t_{1},t_{2}\right) =\log M_{W,Z}\left( t_{1},t_{2}\right) $ and 
\begin{eqnarray*}
	K\left( t_{1},t_{2}\right) &=&\left( t_{1}+t_{2}\right) \log \left[ 2\left(
	1-\rho ^{2}\right) \right] +\tfrac{1}{2}\log \left( 1-\rho ^{2}\right) +\log
	\Gamma \left( \tfrac{1}{2}+t_{1}\right) \\
	&&+\log \Gamma \left( \tfrac{1}{2}+t_{2}\right) -2\log \left[ \Gamma \left( 
	\tfrac{1}{2}\right) \right] +\log _{2}F_{1}\left( \tfrac{1}{2}+t_{1},\tfrac{1%
	}{2}+t_{2};\tfrac{1}{2};\rho ^{2}\right) .
\end{eqnarray*}

The covariance between $W$ and $Z$ when $k=2,$ is given by, $cov\left(
W,Z\right) =\left. \frac{\partial ^{2}K\left( t_{1},t_{2}\right) }{\partial
	t_{1}\partial t_{2}}\right\vert _{t_{1}=0,t_{2}=0}.$ Therefore, let us
firstly evaluate $\left. \partial K\left( t_{1},t_{2}\right) \right/
\partial t_{1}$, as%
\begin{align}
\dfrac{\partial K\left( t_{1},t_{2}\right) }{\partial t_{1}} &=\log \left[
2\left( 1-\rho ^{2}\right) \right] +\Psi \left( \tfrac{1}{2}+t_{1}\right)  
\notag \\
& \ \ \ +\left( _{2}F_{1}\left( \tfrac{1}{2}+t_{1},\tfrac{1}{2}+t_{2};\tfrac{1}{2}%
;\rho ^{2}\right) \right) ^{-1}\frac{\partial _{2}F_{1}\left( \frac{1}{2}%
	+t_{1},\frac{1}{2}+t_{2};\frac{1}{2};\rho ^{2}\right) }{\partial t_{1}}, 
\tag{B.5}  \label{derivative of K wrt t1}
\end{align}%
where $\Psi \left( .\right) $ is the digamma function and $\left. \partial
_{2}F_{1}\left( \frac{1}{2}+t_{1},\frac{1}{2}+t_{2};\frac{1}{2};\rho
^{2}\right) \right/ \partial t_{1}$ is given by,%
\begin{align}
&\sum_{i=1}^{\infty }\frac{\partial \left. \Gamma \left( \frac{1}{2}%
	+t_{1}+i\right) \right/ \Gamma \left( \frac{1}{2}+t_{1}\right) }{\partial
	t_{1}}\frac{\Gamma \left( \frac{1}{2}+t_{2}+i\right) \Gamma \left( \frac{1}{2%
	}\right) }{\Gamma \left( \frac{1}{2}+t_{2}\right) \Gamma \left( \frac{1}{2}%
	+i\right) }\frac{\left( \rho ^{2}\right) ^{i}}{i!}  \notag \\
&=\sum_{i=1}^{\infty }\left( \frac{\Gamma \left( \frac{1}{2}+t_{1}\right)
	\Gamma \left( \frac{1}{2}+t_{1}+i\right) \Psi \left( \frac{1}{2}%
	+t_{1}+i\right) }{\left( \Gamma \left( \frac{1}{2}+t_{1}\right) \right) ^{2}}%
+\frac{\Gamma \left( \frac{1}{2}+t_{1}+i\right) \Psi \left( \frac{1}{2}%
	+t_{1}\right) \Gamma \left( \frac{1}{2}+t_{1}\right) }{\left( \Gamma \left( 
	\frac{1}{2}+t_{1}\right) \right) ^{2}}\right)   \notag \\
& \ \ \ \times \ \frac{\Gamma \left( \frac{1}{2}+t_{2}+i\right) \Gamma \left( \frac{1%
	}{2}\right) }{\Gamma \left( \frac{1}{2}+t_{2}\right) \Gamma \left( \frac{1}{2%
	}+i\right) }\frac{\left( \rho ^{2}\right) ^{i}}{i!}  \notag \\
&=\sum_{i=1}^{\infty }\left( \frac{\Gamma \left( \frac{1}{2}+t_{1}+i\right)
	\Psi \left( \frac{1}{2}+t_{1}+i\right) +\Gamma \left( \frac{1}{2}%
	+t_{1}+i\right) \Psi \left( \frac{1}{2}+t_{1}\right) }{\Gamma \left( \frac{1%
	}{2}+t_{1}\right) }\right) \frac{\Gamma \left( \frac{1}{2}+t_{2}+i\right)
	\Gamma \left( \frac{1}{2}\right) }{\Gamma \left( \frac{1}{2}+t_{2}\right)
	\Gamma \left( \frac{1}{2}+i\right) }\frac{\left( \rho ^{2}\right) ^{i}}{i!}.
\notag \\
&  \tag{B.6}  \label{deri_t1}
\end{align}%
This leads to,%
\begin{equation*}
\left. \dfrac{\partial _{2}F_{1}\left( \frac{1}{2}+t_{1},\frac{1}{2}+t_{2};%
	\frac{1}{2};\rho ^{2}\right) }{\partial t_{1}}\right\vert
_{t_{1}=0,t_{2}=0}=\sum_{i=1}^{\infty }\left( \Psi \left( \tfrac{1}{2}%
+i\right) +\Psi \left( \tfrac{1}{2}\right) \right) \frac{\Gamma \left( \frac{%
		1}{2}+i\right) }{\Gamma \left( \frac{1}{2}\right) }\frac{\left( \rho
	^{2}\right) ^{i}}{i!}.
\end{equation*}%
The first derivative of $_{2}F_{1}\left( \frac{1}{2}+t_{1},\frac{1}{2}+t_{2};%
\frac{1}{2};\rho ^{2}\right) $ with respect to $t_{2}$ is also given by (\ref%
{deri_t1}).

Now let us evaluate the second order derivative of $K\left(
t_{1},t_{2}\right) ,$%
\begin{eqnarray*}
	\dfrac{\partial ^{2}K\left( t_{1},t_{2}\right) }{\partial t_{1}\partial t_{2}%
	} &=&\dfrac{\partial \left( _{2}F_{1}\left( \frac{1}{2}+t_{1},\frac{1}{2}%
		+t_{2};\frac{1}{2};\rho ^{2}\right) \right) ^{-1}\frac{\partial
			_{2}F_{1}\left( \frac{1}{2}+t_{1},\frac{1}{2}+t_{2};\frac{1}{2};\rho
			^{2}\right) }{\partial t_{1}}}{\partial t_{2}} \\
	&=&\left( _{2}F_{1}\left( \tfrac{1}{2}+t_{1},\tfrac{1}{2}+t_{2};\tfrac{1}{2}%
	;\rho ^{2}\right) \right) ^{-1}\dfrac{\partial ^{2}\text{ }_{2}F_{1}\left( 
		\frac{1}{2}+t_{1},\frac{1}{2}+t_{2};\frac{1}{2};\rho ^{2}\right) }{\partial
		t_{1}\partial t_{2}} \\
	&&-\left( _{2}F_{1}\left( \tfrac{1}{2}+t_{1},\tfrac{1}{2}+t_{2};\tfrac{1}{2}%
	;\rho ^{2}\right) \right) ^{-2}\frac{\partial _{2}F_{1}\left( \frac{1}{2}%
		+t_{1},\frac{1}{2}+t_{2};\frac{1}{2};\rho ^{2}\right) }{\partial t_{2}} \\
	&&\times \frac{\partial _{2}F_{1}\left( \frac{1}{2}+t_{1},\frac{1}{2}+t_{2};%
		\frac{1}{2};\rho ^{2}\right) }{\partial t_{1}},
\end{eqnarray*}%
where $\left. \partial ^{2}\text{ }_{2}F_{1}\left( \frac{1}{2}+t_{1},\frac{1%
}{2}+t_{2};\frac{1}{2};\rho ^{2}\right) \right/ \partial t_{1}\partial t_{2}$
is given by,%
\begin{eqnarray*}
	&&\sum_{i=1}^{\infty }\left( \frac{\Gamma \left( \frac{1}{2}+t_{1}+i\right)
		\Psi \left( \frac{1}{2}+t_{1}+i\right) }{\Gamma \left( \frac{1}{2}%
		+t_{1}\right) }+\frac{\Gamma \left( \frac{1}{2}+t_{1}+i\right) \Psi \left( 
		\frac{1}{2}+t_{1}\right) }{\Gamma \left( \frac{1}{2}+t_{1}\right) }\right) 
	\frac{\Gamma \left( \frac{1}{2}\right) }{\Gamma \left( \frac{1}{2}+i\right) }%
	\frac{\left( \rho ^{2}\right) ^{i}}{i!} \\
	&&\times \left( \frac{\Gamma \left( \frac{1}{2}+t_{2}+i\right) \Psi \left( 
		\frac{1}{2}+t_{2}+i\right) }{\Gamma \left( \frac{1}{2}+t_{2}\right) }+\frac{%
		\Gamma \left( \frac{1}{2}+t_{2}+i\right) \Psi \left( \frac{1}{2}%
		+t_{2}\right) }{\Gamma \left( \frac{1}{2}+t_{2}\right) }\right) ,
\end{eqnarray*}%
with%
\begin{equation*}
\left. \dfrac{\partial ^{2}\text{ }_{2}F_{1}\left( \frac{1}{2}+t_{1},\frac{1%
	}{2}+t_{2};\frac{1}{2};\rho ^{2}\right) }{\partial t_{1}\partial t_{2}}%
\right\vert _{t_{1}=0,t_{2}=0}=\sum_{i=1}^{\infty }\left( \Psi \left( \tfrac{%
	1}{2}+i\right) +\Psi \left( \tfrac{1}{2}\right) \right) ^{2}\frac{\Gamma
	\left( \frac{1}{2}+i\right) }{\Gamma \left( \frac{1}{2}\right) }\frac{\left(
	\rho ^{2}\right) ^{i}}{i!}.
\end{equation*}

Hence $cov\left( W,Z\right) $ is given by,%
\begin{align}
&\left( 1-\rho ^{2}\right) ^{\frac{1}{2}}\left. \dfrac{\partial ^{2}\text{ }%
	_{2}F_{1}\left( \frac{1}{2}+t_{1},\frac{1}{2}+t_{2};\frac{1}{2};\rho
	^{2}\right) }{\partial t_{1}\partial t_{2}}\right\vert _{t_{1}=0,t_{2}=0} 
\notag \\
& \ \ \ -\left( 1-\rho ^{2}\right) \left. \frac{\partial _{2}F_{1}\left( \frac{1}{2%
	}+t_{1},\frac{1}{2}+t_{2};\frac{1}{2};\rho ^{2}\right) }{\partial t_{1}}%
\frac{\partial _{2}F_{1}\left( \frac{1}{2}+t_{1},\frac{1}{2}+t_{2};\frac{1}{2%
	};\rho ^{2}\right) }{\partial t_{2}}\right\vert _{t_{1}=0,t_{2}=0}  \notag \\
&=\left( 1-\rho ^{2}\right) ^{\frac{1}{2}}\sum_{i=1}^{\infty }\left( \Psi
\left( \tfrac{1}{2}+i\right) +\Psi \left( \tfrac{1}{2}\right) \right) ^{2}%
\frac{\Gamma \left( \frac{1}{2}+i\right) }{\Gamma \left( \frac{1}{2}\right) }%
\frac{\left( \rho ^{2}\right) ^{i}}{i!}  \notag \\
& \ \ \ -\left( 1-\rho ^{2}\right) \left( \sum_{i=1}^{\infty }\left( \Psi \left( 
\tfrac{1}{2}+i\right) +\Psi \left( \tfrac{1}{2}\right) \right) \frac{\Gamma
	\left( \frac{1}{2}+i\right) }{\Gamma \left( \frac{1}{2}\right) }\frac{\left(
	\rho ^{2}\right) ^{i}}{i!}\right) ^{2},  \tag{B.7}
\label{Covariance(logX,logY)}
\end{align}%
using the fact $_{1}F_{0}\left( a;;z\right) =\left( 1-z\right) ^{-a}.$

Let us now provide the expression for $\rho $ in (\ref{Covariance(logX,logY)}). For example, consider calculating the correlation between the full- and
sub-sample periodograms. Using the similar arguments, the correlation
between two sub-samples periodograms, $\rho =corr\left( I_{Y}^{\left(
	n\right) }\left( \lambda \right) ,I_{Y_{i}}^{\left( l\right) }\left( \mu
\right) \right) $ can be derived using%
\begin{align}
Cov\left( I_{Y}^{\left( n\right) }\left( \lambda \right) ,I_{Y_{i}}^{\left(
	l\right) }\left( \mu \right) \right)  &\approx \frac{2\pi }{l}%
f_{YYY_{i}Y_{i}}\left( \lambda ,-\lambda ,\mu \right) +l^{-2}\left[ \Delta
^{\left( l\right) }\left( -\lambda +\mu \right) \Delta ^{\left( l\right)
}\left( \lambda -\mu \right) \right.   \notag \\
& \ \ \ +\left. \Delta ^{\left( l\right) }\left( \lambda +\mu \right) \Delta
^{\left( l\right) }\left( -\lambda -\mu \right) \right] \left\vert
f_{YY_{i}}\left( \lambda \right) \right\vert ^{2},  \tag{B.8}  \label{Cov}
\end{align}%
and $Var\left( I_{Y}^{\left( n\right) }\left( \lambda \right) \right) $ and $%
Var\left( I_{Y_{i}}^{\left( l\right) }\left( \mu \right) \right) $ can be
calculated from the above given covariance formula. The covariance and
variance terms rely upon certain joint spectral densities. Those spectral
densities can be expressed in closed form as follows. Let us firstly
consider the cross spectrum corresponding to the full sample and $j^{th}$
sub-sample, $f_{YY_{j}}\left( \lambda \right) .$~Suppose we consider the
jackknife approach using non-overlapping sub-samples. Then, the general
definition of spectral density gives that%
\begin{eqnarray*}
	f_{YY_{j}}\left( \lambda \right)  &=&\frac{1}{2\pi }\sum\limits_{k=-\infty
	}^{\infty }\exp \left( -ik\lambda \right) \kappa \left( Y_{t+k},Y_{t+\left(
		j-1\right) l}\right)  \\
	&=&\frac{1}{2\pi }\sum\limits_{k=-\infty }^{\infty }\exp \left( -ik\lambda
	\right) \gamma \left( k-\left( j-1\right) l\right)  \\
	&=&\frac{\exp \left( -i\left( j-1\right) l\lambda \right) }{2\pi }%
	\sum\limits_{k=-\infty }^{\infty }\exp \left( -i\left( k-\left( j-1\right)
	l\right) \lambda \right) \gamma \left( k-\left( j-1\right) l\right)  \\
	&=&\exp \left( -i\left( j-1\right) l\lambda \right) f_{YY}\left( \lambda
	\right) .
\end{eqnarray*}%
Similarly, for moving-block sub-samples we have the relationship $%
f_{YY_{j}}\left( \lambda \right) =\exp \left( -i\left( j+l-1\right) \lambda
\right) f_{YY}\left( \lambda \right) $ and $f_{Y_{j}Y_{k}}\left( \lambda
\right) =\exp \left( -i\left( j-k\right) l\lambda \right) f_{YY}\left(
\lambda \right) .$

Lemma 2 of \color{blue}\cite{Yajima1989}\color{black}\ immediately gives
that,%
\begin{equation*}
f_{YYYY}\left( \lambda ,-\lambda ,\mu \right) =\frac{1}{\left( 2\pi \right)
	^{3}}b\left( \lambda \right) b\left( -\lambda \right) b\left( \mu \right)
b\left( -\mu \right) f_{\varepsilon \varepsilon \varepsilon \varepsilon
}\left( \lambda ,-\lambda ,\mu \right) ,
\end{equation*}%
where $b\left( \lambda \right) =\sum_{j=0}^{\infty }b_{j}\exp \left( \imath
j\omega \right) $ with $b_{j}=\sum_{r=0}^{j}\dfrac{k\left( j-r\right) \Gamma
	\left( r+d\right) }{\Gamma \left( r+1\right) \Gamma \left( d\right) }$, and $%
k\left( z\right) $\ is the transfer function of a stable and invertible
autoregressive moving average (ARMA) process such that $\sum_{j=0}^{\infty
}\left\vert k\left( j\right) \right\vert <\infty .$ Here, 
\begin{equation*}
f_{\varepsilon \varepsilon \varepsilon \varepsilon }\left( \lambda ,-\lambda
,\mu \right) =\sum_{u_{1}=-\infty }^{\infty }\sum_{u_{2}=-\infty }^{\infty
}\sum_{u_{3}=-\infty }^{\infty }\exp \left( -i\left( \lambda u_{1}-\lambda
u_{2}+\mu u_{3}\right) \right) \kappa _{\varepsilon \varepsilon \varepsilon
	\varepsilon }\left( u_{1},u_{2},u_{3}\right) ,
\end{equation*}%
where%
\begin{eqnarray*}
	\kappa _{\varepsilon \varepsilon \varepsilon \varepsilon }\left(
	u_{1},u_{2},u_{3}\right)  &=&\kappa \left( \varepsilon
	_{t+u_{1}},\varepsilon _{t+u_{2}},\varepsilon _{t+u_{3}},\varepsilon
	_{t}\right)  \\
	&=&E\left( \varepsilon _{t+u_{1}}\varepsilon _{t+u_{2}}\varepsilon
	_{t+u_{3}}\varepsilon _{t}\right) -E\left( \varepsilon _{t+u_{1}}\varepsilon
	_{t+u_{2}}\right) E\left( \varepsilon _{t+u_{3}}\varepsilon _{t}\right)  \\
	&&- \ E\left( \varepsilon _{t+u_{2}}\varepsilon _{t+u_{3}}\right) E\left(
	\varepsilon _{t+u_{1}}\varepsilon _{t}\right) -E\left( \varepsilon
	_{t+u_{1}}\varepsilon _{t+u_{3}}\right) E\left( \varepsilon
	_{t+u_{2}}\varepsilon _{t}\right) .
\end{eqnarray*}%
Suppose the errors are $i.i.d$ normal random variables with zero mean and a
constant variance $\sigma ^{2},$%
\begin{eqnarray*}
	\kappa _{\varepsilon \varepsilon \varepsilon \varepsilon }\left(
	u_{1},u_{2},u_{3}\right)  &=&\left\{ 
	\begin{array}{cl}
		E\left( \varepsilon _{t}^{4}\right) -3\left( E\left( \varepsilon
		_{t}^{2}\right) \right) ^{2}, & if\text{ }u_{1}=u_{2}=u_{3}=0 \\ 
		0, & otherwise%
	\end{array}%
	\right.  \\
	&=&\left\{ 
	\begin{array}{cl}
		3\sigma ^{4}, & if\text{ }u_{1}=u_{2}=u_{3}=0 \\ 
		0, & otherwise%
	\end{array}%
	\right. .
\end{eqnarray*}%
Then $f_{YYYY}\left( \lambda ,-\lambda ,\mu \right) $ is simplified as
follows using the fact that $f_{YY}\left( \lambda \right) =\dfrac{\sigma ^{2}%
}{2\pi }b\left( \lambda \right) b\left( -\lambda \right) $.%
\begin{equation*}
f_{YYYY}\left( \lambda ,-\lambda ,\mu \right) =\frac{3\sigma ^{4}}{\left(
	2\pi \right) ^{3}}b\left( -\lambda \right) b\left( \lambda \right) b\left(
\mu \right) b\left( -\mu \right) =\frac{3}{2\pi }f_{YY}\left( \lambda
\right) f_{YY}\left( \mu \right) .
\end{equation*}

Now let us consider $f_{YYY_{j}Y_{j}}\left( \lambda ,-\lambda ,\mu \right) $.%
\begin{eqnarray*}
	f_{YYY_{j}Y_{j}}\left( \lambda ,-\lambda ,\mu \right) &=&\frac{1}{\left(
		2\pi \right) ^{3}}\sum\limits_{u_{1}=-\infty }^{\infty
	}\sum\limits_{u_{2}=-\infty }^{\infty }\sum\limits_{u_{3}=-\infty }^{\infty
	}\exp \left( -\imath \left( \lambda u_{1}-\lambda u_{2}+\mu u_{3}\right)
	\right) \\
	&&\times \kappa \left( Y_{t+u_{1}},Y_{t+u_{2}},Y_{t+\left( j-1\right)
		l+u_{3}},Y_{t+\left( j-1\right) l}\right) \\
	&=&\frac{1}{\left( 2\pi \right) ^{3}}\sum\limits_{u_{1}=-\infty }^{\infty
	}\sum\limits_{u_{2}=-\infty }^{\infty }\sum\limits_{u_{3}=-\infty }^{\infty
	}\exp \left( -\imath \left( \lambda \left( u_{1}-\left( j-1\right) l\right)
	-\lambda \left( u_{2}-\left( j-1\right) l\right) +\mu u_{3}\right) \right) \\
	&&\times \kappa \left( Y_{t-\left( j-1\right) l+u_{1}},Y_{t-\left(
		j-1\right) l+u_{2}},Y_{t+u_{3}},Y_{t}\right) \\
	&=&f_{YYYY}\left( \lambda ,-\lambda ,\mu \right) .
\end{eqnarray*}%
The covariance and variance terms in (\ref{Cov}) can thus be simplified as
follows.%
\begin{eqnarray*}
	Cov\left( I_{Y}^{\left( n\right) }\left( \lambda \right) ,I_{Y_{i}}^{\left(
		l\right) }\left( \mu \right) \right) &\approx &\frac{3}{l}f_{YY}\left(
	\lambda \right) f_{YY}\left( \mu \right) +\frac{1}{l^{2}}\left[ \Delta
	^{\left( l\right) }\left( -\lambda +\mu \right) \Delta ^{\left( l\right)
	}\left( \lambda -\mu \right) \right. \\
	&&+\left. \Delta ^{\left( l\right) }\left( \lambda +\mu \right) \Delta
	^{\left( l\right) }\left( -\lambda -\mu \right) \right] \left( f_{YY}\left(
	\lambda \right) \right) ^{2}, \\
	Var\left( I_{Y}^{\left( n\right) }\left( \lambda \right) \right) &\approx & 
	\left[ 1+\frac{3}{l}+\frac{1}{l^{2}}\Delta ^{\left( l\right) }\left(
	2\lambda \right) \Delta ^{\left( l\right) }\left( -2\lambda \right) \right]
	\left( f_{YY}\left( \lambda \right) \right) ^{2}.
\end{eqnarray*}%
Hence, the correlation is given by,%
\begin{equation*}
\rho \approx \frac{\frac{3}{l}+\frac{1}{l^{2}}\left[ \Delta ^{\left(
		l\right) }\left( -\lambda +\mu \right) \Delta ^{\left( l\right) }\left(
	\lambda -\mu \right) +\Delta ^{\left( l\right) }\left( \lambda +\mu \right)
	\Delta ^{\left( l\right) }\left( -\lambda -\mu \right) \right] \frac{%
		f_{YY}\left( \lambda \right) }{f_{YY}\left( \mu \right) }}{\sqrt{\left( 1+%
		\frac{3}{l}+\frac{1}{l^{2}}\Delta ^{\left( l\right) }\left( 2\lambda \right)
		\Delta ^{\left( l\right) }\left( -2\lambda \right) \right) }\sqrt{\left( 1+%
		\frac{3}{l}+\frac{1}{l^{2}}\Delta ^{\left( l\right) }\left( 2\mu \right)
		\Delta ^{\left( l\right) }\left( -2\mu \right) \right) }}.
\end{equation*}

\subsection*{Positiveness of the principle minors of the bordered Hessian
	matrix}

Here we show that for every $m\in 
\mathbb{N}
$, $\left\vert \mathbf{H}_{\left( m+3\right) \times \left( m+3\right)
}^{B}\right\vert >0$ using mathematical induction. For our convenience, we
assume that%
\begin{equation*}
\varphi _{\min }\left( \mathbf{H}_{\left( m+3\right) \times \left(
	m+3\right) }^{B}\right) >\left( m+3\right) ^{2}\frac{12N_{l}}{\pi ^{2}},
\end{equation*}%
where $\varphi _{\min }\left( \mathbf{A}\right) $ is the minimum eigenvalue
corresponding to the matrix $\mathbf{A.}$

Let us start with $m=1.$ The first minor of the bordered Hessian matrix, $%
\mathbf{H}_{4\times 4}^{B},$ is,%
\begin{eqnarray*}
	\left\vert \mathbf{H}_{4\times 4}^{B}\right\vert &=&\left\vert 
	\begin{array}{ccc}
		0 & 0 & -m^{2}\frac{N_{l}^{2}}{l^{2}} \\ 
		1 & \frac{N_{n}^{2}}{n^{2}} & -2c_{n,1}^{\ast } \\ 
		-1 & -m^{2}\frac{N_{l}^{2}}{l^{2}} & \frac{\pi ^{2}}{12N_{l}}%
	\end{array}%
	\right\vert +\left\vert 
	\begin{array}{ccc}
		0 & 0 & \frac{N_{n}^{2}}{n^{2}} \\ 
		1 & \frac{N_{n}^{2}}{n^{2}} & \frac{\pi ^{2}}{12N_{n}} \\ 
		-1 & -m^{2}\frac{N_{l}^{2}}{l^{2}} & -2c_{n,1}^{\ast }%
	\end{array}%
	\right\vert \\
	&=&-m^{2}\frac{N_{l}^{2}}{l^{2}}\left( -m^{2}\frac{N_{l}^{2}}{l^{2}}+\frac{%
		N_{n}^{2}}{n^{2}}\right) +\frac{N_{n}^{2}}{n^{2}}\left( -m^{2}\frac{N_{l}^{2}%
	}{l^{2}}+\frac{N_{n}^{2}}{n^{2}}\right) =\left( \frac{N_{n}^{2}}{n^{2}}-m^{2}%
	\frac{N_{l}^{2}}{l^{2}}\right) ^{2}>0.
\end{eqnarray*}%
That is, $\left\vert \mathbf{H}_{\left( m+3\right) \times \left( m+3\right)
}^{B}\right\vert >0$ for $m=1.$

Suppose that $\left\vert \mathbf{H}_{\left( m+3\right) \times \left(
	m+3\right) }^{B}\right\vert >0$ is true for $m=k,$ then we need to show that
it is true for $m=k+1$. To do so, we consider the partition of $\mathbf{H}%
_{\left( k+4\right) \times \left( k+4\right) }^{B}$ is as follows: 
\begin{equation*}
\mathbf{H}_{\left( k+4\right) \times \left( k+4\right) }^{B}=\left( 
\begin{array}{cc}
\mathbf{H}_{\left( k+3\right) \times \left( k+3\right) }^{B} & \mathbf{U} \\ 
\mathbf{U}^{T} & \frac{\pi ^{2}}{12N_{l}}%
\end{array}%
\right) ,
\end{equation*}%
where $\mathbf{U}^{\top }=\left[ 
\begin{array}{cccccc}
-1 & -\left( k+1\right) ^{2}\tfrac{N_{l}^{2}}{l^{2}} & -2c_{n,k+1}^{\ast } & 
2c_{1,k+1}^{\dagger } & \ldots  & 2c_{k,k+1}^{\dagger }%
\end{array}%
\right] $. Then,%
\begin{equation*}
\left\vert \mathbf{H}_{\left( k+4\right) \times \left( k+4\right)
}^{B}\right\vert =\left\vert \mathbf{H}_{\left( k+3\right) \times \left(
	k+3\right) }^{B}\right\vert \left( \frac{\pi ^{2}}{12N_{l}}-\mathbf{U}^{\top
}\left( \mathbf{H}_{\left( k+3\right) \times \left( k+3\right) }^{B}\right)
^{-1}\mathbf{U}\right) .
\end{equation*}%
Since $\left\vert \mathbf{H}_{\left( k+3\right) \times \left( k+3\right)
}^{B}\right\vert >0,$%
\begin{equation*}
0<\mathbf{U}^{\top }\left( \mathbf{H}_{\left( k+3\right) \times \left(
	k+3\right) }^{B}\right) ^{-1}\mathbf{U}\leq \tfrac{1}{\varphi _{\min }\left( 
	\mathbf{H}_{\left( k+3\right) \times \left( k+3\right) }^{B}\right) }\max_{%
	\mathbf{U}\in 
	\mathbb{R}
	^{k+3}\backslash \left\{ \mathbf{0}\right\} }\mathbf{U^{\top }U}<\frac{\pi
	^{2}}{12N_{l}},\text{ as }\max_{\mathbf{U}\in 
	\mathbb{R}
	^{k+3}\backslash \left\{ \mathbf{0}\right\} }\mathbf{U^{\top }U}=1.
\end{equation*}%
Hence this completes the proof.

\section*{Appendix C: Monte Carlo results: Tables 2 to 16}

\begin{table}[h]
\vspace{-4cm}%
\caption{{\small Bias estimates of the unadjusted LPR estimator, the optimal
jackknife estimator based on 2,3,4,6,8 non-overlapping (NO) sub-samples, the
optimal jackknife estimator based on 2 moving block (MB) sub-samples, both
versions of the GS estimator, the pre-filtered sieve bootstrap estimator,
the maximum likelihood estimator (MLE) and the pre-whitened (PW) estimator, for the
DGP: ARFIMA(}${\small 1,d}_{{\small 0}}{\small ,0}${\small ) with Gaussian
innovations. The optimal jackknife estimates are evaluated as described in
Section 5.1. The estimates are obtained by setting }${\small \alpha }${\small \ = 0.65 and assuming the model is correctly specified. The lowest values
are \textbf{bold-faced} and the second lowest values are \textit{italicized}.}}%
\label{Table:Bias_ARFIMA(1,d,0)_b}%
\renewcommand{\arraystretch}{0.9}%
\setlength{\tabcolsep}{2.3pt}%

\medskip

\hspace{-0.6in}%
\begin{tabular}{lllcccccccccccc}
\hline\hline
&  &  &  &  &  &  &  &  &  &  &  &  &  &  \\ 
${\small \phi }_{{\small 0}}$ & ${\small d}_{{\small 0}}$ & ${\small n}$ & $%
\widehat{{\small d}}_{{\small n}}$ & $\widehat{{\small d}}_{{\small J,2}}^{%
{\small Opt}\text{{\small -}}{\small NO}}$ & $\widehat{{\small d}}_{{\small %
J,3}}^{{\small Opt}\text{{\small -}}{\small NO}}$ & $\widehat{{\small d}}_{%
{\small J,4}}^{{\small Opt}\text{{\small -}}{\small NO}}$ & $\widehat{%
{\small d}}_{{\small J,6}}^{{\small Opt}\text{{\small -}}{\small NO}}$ & $%
\widehat{{\small d}}_{{\small J,8}}^{{\small Opt}\text{{\small -}}{\small NO}%
}$ & $\widehat{{\small d}}_{{\small J,2}}^{{\small Opt}\text{{\small -}}%
{\small MB}}$ & $\widehat{{\small d}}_{{\small 1}}^{{\small GS}}$ & $%
\widehat{{\small d}}_{{\small 1}}^{Opt\text{-}GS}$ & $\widehat{{\small d}}^{%
{\small PFSB}}$ & $\widehat{{\small d}}^{{\small MLE}}$ & $\widehat{{\small d%
}}^{{\small PW}}$ \\ \hline
\multicolumn{15}{c}{} \\ 
\multicolumn{1}{c}{\small -0.9} & \multicolumn{1}{c}{\small -0.25} & 
\multicolumn{1}{c}{\small 96} & {\small 0.8145} & \multicolumn{1}{r}{\small %
0.7852} & \multicolumn{1}{r}{\small 0.7903} & \multicolumn{1}{r}{\small %
0.7995} & \multicolumn{1}{r}{\small 0.8072} & \multicolumn{1}{r}{\small %
0.8120} & \multicolumn{1}{r}{\small 0.8156} & \multicolumn{1}{r}{\small %
0.8002} & \multicolumn{1}{r}{\small 0.7902} & \multicolumn{1}{r}{\small %
0.7908} & \multicolumn{1}{r}{{\small \small \textbf{0.6408}
}} & \multicolumn{1}{r}{{\small 
\small \textit{0.7047}
}} \\
\hline\hline
\end{tabular}

\end{table}%

\newpage

\begin{table}[h]\
\vspace{-7.0cm}%
\caption{{\small RMSE estimates of the unadjusted LPR estimator, the optimal
jackknife estimator based on 2,3,4,6,8 non-overlapping (NO) sub-samples, the
optimal jackknife estimator based on 2 moving block (MB) sub-samples, both
versions of the GS estimator, the pre-filtered sieve bootstrap estimator,
the maximum likelihood estimator (MLE) and the pre-whitened (PW) estimator, for the
DGP: ARFIMA(}${\small 1,d}_{{\small 0}}{\small ,0}${\small ) with Gaussian
innovations. The optimal jackknife estimates are evaluated as described in
Section 5.1. The estimates are obtained by setting }${\small \alpha }${\small \ = 0.65 and assuming the model is correctly specified. The lowest values
are \textbf{bold-faced} and the second lowest values are \textit{italicized}.}}%
\label{Table:MSE_ARFIMA(1,d,0)_b}%
\renewcommand{\arraystretch}{0.9}%
\setlength{\tabcolsep}{2.3pt}%

\medskip

\hspace{-0.35in}%
\begin{tabular}{lllcccccccccccc}
\hline\hline
&  &  &  &  &  &  &  &  &  &  &  &  &  &  \\ 
${\small \phi }_{{\small 0}}$ & ${\small d}_{{\small 0}}$ & ${\small n}$ & $%
\widehat{{\small d}}_{{\small n}}$ & $\widehat{{\small d}}_{{\small J,2}}^{%
{\small Opt}\text{{\small -}}{\small NO}}$ & $\widehat{{\small d}}_{{\small %
J,3}}^{{\small Opt}\text{{\small -}}{\small NO}}$ & $\widehat{{\small d}}_{%
{\small J,4}}^{{\small Opt}\text{{\small -}}{\small NO}}$ & $\widehat{%
{\small d}}_{{\small J,6}}^{{\small Opt}\text{{\small -}}{\small NO}}$ & $%
\widehat{{\small d}}_{{\small J,8}}^{{\small Opt}\text{{\small -}}{\small NO}%
}$ & $\widehat{{\small d}}_{{\small J,2}}^{{\small Opt}\text{{\small -}}%
{\small MB}}$ & $\widehat{{\small d}}_{{\small 1}}^{{\small GS}}$ & $%
\widehat{{\small d}}_{{\small 1}}^{Opt\text{-}GS}$ & $\widehat{{\small d}}^{%
{\small PFSB}}$ & $\widehat{{\small d}}^{{\small MLE}}$ & $\widehat{{\small d%
}}^{{\small PW}}$ \\ \hline
\multicolumn{15}{c}{} \\ 
\multicolumn{1}{c}{\small -0.9} & \multicolumn{1}{c}{\small -0.25} & 
\multicolumn{1}{c}{\small 96} & {\small 1.0359} & {\small 1.0627} & {\small %
1.0532} & {\small 1.0596} & {\small 1.0358} & {\small 1.0286} & {\small %
1.1837} & {\small 1.3386} & {\small 1.1864} & {\small 1.2885} & {\small 
\small \textbf{0.7257}%
} & {\small 
\small \textit{0.9158}%
} \\ 
\multicolumn{1}{c}{} & \multicolumn{1}{c}{} & \multicolumn{1}{c}{\small 576}
& {\small 0.7398} & {\small 0.7490} & {\small 0.7403} & {\small 0.7372} & 
{\small 0.7325} & {\small 0.7299} & {\small 0.7382} & {\small 0.7371} & 
{\small 0.7200} & {\small 0.7359} & {\small 
\small \textbf{0.6353}%
} & {\small 
\small \textit{0.6994}%
} \\ 
\multicolumn{1}{c}{} & \multicolumn{1}{c}{\small 0} & \multicolumn{1}{c}%
{\small 96} & {\small 1.1148} & {\small 1.1398} & {\small 1.1275} & {\small %
1.1158} & {\small 1.1080} & {\small 1.0966} & {\small 1.1576} & {\small %
1.1819} & {\small 1.1120} & {\small 1.2167} & {\small 
\small \textbf{0.7380}%
} & {\small 
\small \textit{0.9181}%
} \\ 
\multicolumn{1}{c}{} & \multicolumn{1}{c}{} & \multicolumn{1}{c}{\small 576}
& {\small 0.8288} & {\small 0.8370} & {\small 0.8311} & {\small 0.8294} & 
{\small 0.8216} & {\small 0.8157} & {\small 0.8215} & {\small 0.8173} & 
{\small 0.8173} & {\small 0.8053} & {\small 
\small \textbf{0.5261}%
} & {\small 
\small \textit{0.5429}%
} \\ 
\multicolumn{1}{c}{} & \multicolumn{1}{c}{\small 0.25} & \multicolumn{1}{c}%
{\small 96} & {\small 1.1618} & {\small 1.1857} & {\small 1.1066} & {\small %
1.0971} & {\small 1.0944} & {\small 1.0913} & {\small 1.1162} & {\small %
1.1484} & {\small 1.1285} & {\small 1.2299} & {\small 
\small \textbf{0.7492}%
} & {\small 
\small \textit{0.9726}%
} \\ 
\multicolumn{1}{c}{} & \multicolumn{1}{c}{} & \multicolumn{1}{c}{\small 576}
& {\small 0.9175} & {\small 0.9250} & {\small 0.9203} & {\small 0.9186} & 
{\small 0.9128} & {\small 0.9076} & {\small 0.9115} & {\small 1.1171} & 
{\small 1.0172} & {\small 1.1130} & {\small 
\small \textbf{0.5258}%
} & {\small 
\small \textit{0.5530}%
} \\ 
\multicolumn{1}{c}{} & \multicolumn{1}{c}{\small 0.45} & \multicolumn{1}{c}%
{\small 96} & {\small 1.1286} & {\small 1.1552} & {\small 1.1325} & {\small %
1.1294} & {\small 1.1200} & {\small 1.1168} & {\small 1.1132} & {\small %
1.4331} & {\small 1.3331} & {\small 1.5385} & {\small 
\small \textbf{0.6482}%
} & {\small 
\small \textit{0.9438}%
} \\ 
\multicolumn{1}{c}{} & \multicolumn{1}{c}{} & \multicolumn{1}{c}{\small 576}
& {\small 0.9708} & {\small 0.9781} & {\small 0.9732} & {\small 0.9650} & 
{\small 0.9558} & {\small 0.9546} & {\small 0.9687} & {\small 1.1124} & 
{\small 1.0524} & {\small 1.1647} & {\small 
\small \textbf{0.5263}%
} & {\small 
\small \textit{0.5492}%
} \\ 
\multicolumn{1}{c}{\small -0.4} & \multicolumn{1}{c}{\small -0.25} & 
\multicolumn{1}{c}{\small 96} & {\small 0.2568} & {\small 0.2292} & {\small %
0.2568} & {\small 0.2422} & {\small 0.2384} & {\small 0.2376} & {\small %
0.2576} & {\small 0.2594} & {\small 0.2441} & {\small 0.3028} & {\small 
\small%
}\textbf{0.1308} & {\small 
\small \textit{0.1953}%
} \\ 
\multicolumn{1}{c}{} & \multicolumn{1}{c}{} & \multicolumn{1}{c}{\small 576}
& {\small 0.1098} & {\small 0.0978} & {\small 0.0974} & {\small 0.0884} & 
{\small 
\small \textit{0.0873}%
} & {\small 0.0896} & {\small 0.1096} & {\small 0.1118} & 
{\small 0.0995} & {\small 0.1272} & {\small 
\small \textbf{0.0662}%
} & {\small 0.0948} \\ 
\multicolumn{1}{c}{} & \multicolumn{1}{c}{\small 0} & \multicolumn{1}{c}%
{\small 96} & {\small 0.2498} & {\small 0.2395} & {\small 0.2284} & {\small %
0.2146} & {\small 0.2138} & {\small 0.2117} & {\small 0.2517} & {\small %
0.2560} & {\small 0.2416} & {\small 0.2930} & {\small 
\small \textbf{0.1309}%
} & {\small 
\small \textit{0.1999}%
} \\ 
\multicolumn{1}{c}{} & \multicolumn{1}{c}{} & \multicolumn{1}{c}{\small 576}
& {\small 0.1069} & {\small 0.0837} & {\small 0.0879} & {\small 0.0819} & 
{\small 0.0787} & {\small 
\small \textit{0.0778}%
} & {\small 0.1078} & {\small 0.1104} & {\small 0.0967} & 
{\small 0.1247} & {\small 
\small \textbf{0.0530}%
} & {\small 0.1065} \\ 
\multicolumn{1}{c}{} & \multicolumn{1}{c}{\small 0.25} & \multicolumn{1}{c}%
{\small 96} & {\small 0.2490} & {\small 0.2678} & {\small 0.2574} & {\small %
0.2435} & {\small 0.2354} & {\small 0.2254} & {\small 0.3254} & {\small %
0.2580} & {\small 0.2404} & {\small 0.2879} & {\small 
\small \textbf{0.1382}%
} & {\small 
\small \textit{0.1896}%
} \\ 
\multicolumn{1}{c}{} & \multicolumn{1}{c}{} & \multicolumn{1}{c}{\small 576}
& {\small 0.1079} & {\small 0.1036} & {\small 0.0965} & {\small 0.0901} & 
{\small 0.0819} & {\small 
\small \textit{0.0797}%
} & {\small 0.1097} & {\small 0.1115} & {\small 0.1029} & 
{\small 0.1239} & {\small 
\small \textbf{0.0528}%
} & {\small 0.1047} \\ 
\multicolumn{1}{c}{} & \multicolumn{1}{c}{\small 0.45} & \multicolumn{1}{c}%
{\small 96} & {\small 0.2506} & {\small 0.2615} & {\small 0.2563} & {\small %
0.2434} & {\small 0.2390} & {\small 0.2243} & {\small 0.2544} & {\small %
0.2616} & {\small 0.2511} & {\small 0.2506} & {\small 
\small \textbf{0.1371}%
} & {\small 
\small \textit{0.1966}%
} \\ 
\multicolumn{1}{c}{} & \multicolumn{1}{c}{} & \multicolumn{1}{c}{\small 576}
& {\small 0.1115} & {\small 0.0963} & {\small 0.0878} & {\small 0.0808} & 
{\small 0.0777} & {\small 
\small \textit{0.0742}%
} & {\small 0.1142} & {\small 0.1143} & {\small 0.1005} & 
{\small 0.1230} & {\small 
\small \textbf{0.0593}%
} & {\small 0.1028} \\ 
\multicolumn{1}{c}{\small 0.4} & \multicolumn{1}{c}{\small -0.25} & 
\multicolumn{1}{c}{\small 96} & {\small 0.1917} & {\small 0.1721} & {\small %
0.1654} & {\small 0.1629} & {\small 0.1544} & {\small 0.1529} & {\small %
0.1929} & {\small 0.2212} & {\small 0.2157} & {\small 0.2717} & {\small 
\small \textbf{0.0904}%
} & {\small 
\small \textit{0.1445}%
} \\ 
\multicolumn{1}{c}{} & \multicolumn{1}{c}{} & \multicolumn{1}{c}{\small 576}
& {\small 0.0919} & {\small 0.0762} & {\small 0.0747} & {\small 0.0665} & 
{\small 0.0632} & {\small 
\small \textit{0.0624}%
} & {\small 0.0924} & {\small 0.1081} & {\small 0.0695} & 
{\small 0.1198} & {\small 
\small \textbf{0.0335}%
} & {\small 0.0764} \\ 
\multicolumn{1}{c}{} & \multicolumn{1}{c}{\small 0} & \multicolumn{1}{c}%
{\small 96} & {\small 0.1946} & {\small 0.1726} & {\small 0.1717} & {\small %
0.1631} & {\small 0.1569} & {\small 0.1557} & {\small 0.1957} & {\small %
0.2203} & {\small 0.2162} & {\small 0.2546} & {\small 
\small \textbf{0.0872}%
} & {\small 
\small%
}\textit{0.1439} \\ 
\multicolumn{1}{c}{} & \multicolumn{1}{c}{} & \multicolumn{1}{c}{\small 576}
& {\small 0.0920} & {\small 0.0890} & {\small 0.0793} & {\small 0.0751} & 
{\small 0.0730} & 0.0724 & {\small 0.0924} & {\small 0.1073} & {\small 
\small \textit{0.0684}%
} & {\small 0.1166} & {\small 
\small \textbf{0.0434}%
} & {\small 0.0753} \\ 
\multicolumn{1}{c}{} & \multicolumn{1}{c}{\small 0.25} & \multicolumn{1}{c}%
{\small 96} & {\small 0.1960} & {\small 0.2107} & {\small 0.2063} & {\small %
0.2008} & {\small 0.1913} & {\small 0.1966} & {\small 0.1966} & {\small %
0.2209} & {\small 0.2091} & {\small 0.2482} & {\small 
\small \textbf{0.0912}%
} & {\small 
\small \textit{0.1535}%
} \\ 
\multicolumn{1}{c}{} & \multicolumn{1}{c}{} & \multicolumn{1}{c}{\small 576}
& {\small 0.0922} & {\small 0.0705} & {\small 0.0696} & {\small 0.0644} & 
{\small 0.0627} & {\small 
\small \textit{0.0624}%
} & {\small 0.0924} & {\small 0.1076} & {\small 0.0688} & 
{\small 0.1158} & {\small 
\small \textbf{0.0381}%
} & {\small 0.0736} \\ 
\multicolumn{1}{c}{} & \multicolumn{1}{c}{\small 0.45} & \multicolumn{1}{c}%
{\small 96} & {\small 0.1955} & {\small 0.2178} & {\small 0.2140} & {\small %
0.2085} & {\small 0.2061} & {\small 0.2058} & {\small 0.1958} & {\small %
0.2218} & {\small 0.2143} & {\small 0.2453} & {\small 
\small \textbf{0.0944}%
} & {\small 
\small \textit{0.1538}%
} \\ 
\multicolumn{1}{c}{} & \multicolumn{1}{c}{} & \multicolumn{1}{c}{\small 576}
& {\small 0.0926} & {\small 0.0710} & {\small 0.0684} & {\small 0.0667} & 
{\small 0.0634} & {\small 
\small \textit{0.0569}%
} & {\small 0.0929} & {\small 0.1089} & {\small 0.0701} & 
{\small 0.1149} & {\small 
\small \textbf{0.0499}%
} & {\small 0.0752} \\ 
\multicolumn{1}{c}{\small 0.9} & \multicolumn{1}{c}{\small -0.25} & 
\multicolumn{1}{c}{\small 96} & {\small 0.1115} & {\small 0.1039} & {\small %
0.1006} & {\small 0.0994} & {\small 0.0913} & {\small 0.0886} & {\small %
0.0932} & {\small 0.1365} & {\small 0.1132} & {\small 0.1266} & {\small 
\small \textbf{0.0482}%
} & {\small 
\small \textit{0.0872}%
} \\ 
\multicolumn{1}{c}{} & \multicolumn{1}{c}{} & \multicolumn{1}{c}{\small 576}
& {\small 0.0624} & {\small 0.0522} & {\small 0.0513} & {\small 0.0482} & 
{\small 0.0440} & {\small 0.0402} & {\small 0.0399} & {\small 0.0708} & 
{\small 0.0659} & {\small 0.0600} & {\small 
\small \textbf{0.0127}%
} & {\small 
\small \textit{0.0331}%
} \\ 
\multicolumn{1}{c}{} & \multicolumn{1}{c}{\small 0} & \multicolumn{1}{c}%
{\small 96} & {\small 0.1010} & {\small 0.1012} & {\small 0.0954} & {\small %
0.0911} & {\small 0.0827} & {\small 0.0813} & {\small 0.0955} & {\small %
0.1121} & {\small 0.0992} & {\small 0.1093} & {\small 
\small \textbf{0.0438}%
} & {\small 
\small \textit{0.0838}%
} \\ 
\multicolumn{1}{c}{} & \multicolumn{1}{c}{} & \multicolumn{1}{c}{\small 576}
& {\small 0.0602} & {\small 0.0504} & {\small 0.0486} & {\small 0.0455} & 
{\small 0.0422} & {\small 0.0391} & {\small 0.0400} & {\small 0.0698} & 
{\small 0.0632} & {\small 0.0705} & {\small 
\small \textbf{0.0121}%
} & {\small 
\small \textit{0.0323}%
} \\ 
\multicolumn{1}{c}{} & \multicolumn{1}{c}{\small 0.25} & \multicolumn{1}{c}%
{\small 96} & {\small 0.1114} & {\small 0.1053} & {\small 0.1011} & {\small %
0.0942} & {\small 0.0930} & {\small 0.0913} & {\small 0.1106} & {\small %
0.1328} & {\small 0.1179} & {\small 0.1282} & {\small 
\small \textbf{0.0463}%
} & {\small 
\small \textit{0.0880}%
} \\ 
\multicolumn{1}{c}{} & \multicolumn{1}{c}{} & \multicolumn{1}{c}{\small 576}
& {\small 0.0518} & {\small 0.0500} & {\small 0.0482} & {\small 0.0438} & 
{\small 0.0419} & {\small 0.0374} & {\small 0.0491} & {\small 0.0626} & 
{\small 0.0573} & {\small 0.0581} & {\small 
\small \textbf{0.0139}%
} & {\small 
\small \textit{0.0341}%
} \\ 
\multicolumn{1}{c}{} & \multicolumn{1}{c}{\small 0.45} & \multicolumn{1}{c}%
{\small 96} & {\small 0.1053} & {\small 0.0992} & {\small 0.0914} & {\small %
0.0824} & {\small 0.0862} & {\small 
\small \textit{0.0801}%
} & {\small 0.0937} & {\small 0.1253} & {\small 0.1188} & 
{\small 0.1215} & {\small 
\small \textbf{0.0418}%
} & {\small 0.0868} \\ 
\multicolumn{1}{c}{} & \multicolumn{1}{c}{} & \multicolumn{1}{c}{\small 576}
& {\small 0.0526} & {\small 0.0518} & {\small 0.0583} & {\small 0.0503} & 
{\small 0.0455} & {\small 0.0412} & {\small 0.0527} & {\small 0.0769} & 
{\small 0.0600} & {\small 0.0684} & {\small 
\small \textbf{0.0122}%
} & {\small 
\small \textit{0.0351}%
} \\ \hline\hline
\end{tabular}

\end{table}%

\newpage

\begin{table}[h]\
\vspace{-7.0cm}%
\caption{{\small Bias estimates of the unadjusted LPR estimator, the optimal
jackknife estimator based on 2,3,4,6,8 non-overlapping (NO) sub-samples, the
optimal jackknife estimator based on 2 moving block (MB) sub-samples, both
versions of the GS estimator, the pre-filtered sieve bootstrap estimator,
the maximum likelihood estimator (MLE) and the pre-whitened (PW) estimator, for the
DGP: ARFIMA(}${\small 0,d}_{{\small 0}}{\small ,1}${\small ) with Gaussian
innovations. The optimal jackknife estimates are evaluated as described in
Section 5.1. The estimates are obtained by setting }${\small \alpha }${\small \ = 0.65 and assuming the model is correctly specified. The lowest values
are \textbf{bold-faced} and the second lowest values are \textit{italicized}.}}%
\label{Table:Bias_ARFIMA(0,d,1)_b}%
\renewcommand{\arraystretch}{0.9}%
\setlength{\tabcolsep}{2.3pt}%

\medskip

\hspace{-0.55in}%
\begin{tabular}{lllcccccccccccc}
\hline\hline
&  &  &  &  &  &  &  &  &  &  &  &  &  &  \\ 
${\small \theta }_{{\small 0}}$ & ${\small d}_{{\small 0}}$ & ${\small n}$ & 
$\widehat{{\small d}}_{{\small n}}$ & $\widehat{{\small d}}_{{\small J,2}}^{%
{\small Opt}\text{{\small -}}{\small NO}}$ & $\widehat{{\small d}}_{{\small %
J,3}}^{{\small Opt}\text{{\small -}}{\small NO}}$ & $\widehat{{\small d}}_{%
{\small J,4}}^{{\small Opt}\text{{\small -}}{\small NO}}$ & $\widehat{%
{\small d}}_{{\small J,6}}^{{\small Opt}\text{{\small -}}{\small NO}}$ & $%
\widehat{{\small d}}_{{\small J,8}}^{{\small Opt}\text{{\small -}}{\small NO}%
}$ & $\widehat{{\small d}}_{{\small J,2}}^{{\small Opt}\text{{\small -}}%
{\small MB}}$ & $\widehat{{\small d}}_{{\small 1}}^{{\small GS}}$ & $%
\widehat{{\small d}}_{{\small 1}}^{Opt\text{-}GS}$ & $\widehat{{\small d}}^{%
{\small PFSB}}$ & $\widehat{{\small d}}^{{\small MLE}}$ & $\widehat{{\small d%
}}^{{\small PW}}$ \\ \hline
\multicolumn{15}{c}{} \\ 
\multicolumn{1}{c}{\small -0.9} & \multicolumn{1}{c}{\small -0.25} & 
\multicolumn{1}{c}{\small 96} & \multicolumn{1}{r}{\small -0.5671} & 
\multicolumn{1}{r}{\small -0.5276} & \multicolumn{1}{r}{\small -0.5348} & 
\multicolumn{1}{r}{\small -0.5429} & \multicolumn{1}{r}{\small -0.5574} & 
\multicolumn{1}{r}{\small -0.5653} & \multicolumn{1}{r}{\small -0.5536} & 
\multicolumn{1}{r}{\small -0.5450} & \multicolumn{1}{r}{\small -0.5329} & 
\multicolumn{1}{r}{\small -0.5466} & \multicolumn{1}{r}{{\small 
\small \textbf{-0.3341}%
}} & \multicolumn{1}{r}{{\small 
\small \textit{-0.4090}%
}} \\
\hline\hline
\end{tabular}

\end{table}%

\newpage

\begin{table}[h]\
\vspace{-7.0cm}%
\caption{{\small RMSE Bias estimates of the unadjusted LPR estimator, the optimal
jackknife estimator based on 2,3,4,6,8 non-overlapping (NO) sub-samples, the
optimal jackknife estimator based on 2 moving block (MB) sub-samples, both
versions of the GS estimator, the pre-filtered sieve bootstrap estimator,
the maximum likelihood estimator (MLE) and the pre-whitened (PW) estimator, for the
DGP: ARFIMA(}${\small 0,d}_{{\small 0}}{\small ,1}${\small ) with Gaussian
innovations. The optimal jackknife estimates are evaluated as described in
Section 5.1. The estimates are obtained by setting }${\small \alpha }${\small \ = 0.65 and assuming the model is correctly specified. The lowest values
are \textbf{bold-faced} and the second lowest values are \textit{italicized}.}}%
\label{Table:MSE_ARFIMA(0,d,1)_b}%
\renewcommand{\arraystretch}{0.9}%
\setlength{\tabcolsep}{2.3pt}%

\medskip

\hspace{-0.35in}%
\begin{tabular}{lllcccccccccccc}
\hline\hline
&  &  &  &  &  &  &  &  &  &  &  &  &  &  \\ 
${\small \theta }_{{\small 0}}$ & ${\small d}_{{\small 0}}$ & ${\small n}$ & 
$\widehat{{\small d}}_{{\small n}}$ & $\widehat{{\small d}}_{{\small J,2}}^{%
{\small Opt}\text{{\small -}}{\small NO}}$ & $\widehat{{\small d}}_{{\small %
J,3}}^{{\small Opt}\text{{\small -}}{\small NO}}$ & $\widehat{{\small d}}_{%
{\small J,4}}^{{\small Opt}\text{{\small -}}{\small NO}}$ & $\widehat{%
{\small d}}_{{\small J,6}}^{{\small Opt}\text{{\small -}}{\small NO}}$ & $%
\widehat{{\small d}}_{{\small J,8}}^{{\small Opt}\text{{\small -}}{\small NO}%
}$ & $\widehat{{\small d}}_{{\small J,2}}^{{\small Opt}\text{{\small -}}%
{\small MB}}$ & $\widehat{{\small d}}_{{\small 1}}^{{\small GS}}$ & $%
\widehat{{\small d}}_{{\small 1}}^{Opt\text{-}GS}$ & $\widehat{{\small d}}^{%
{\small PFSB}}$ & $\widehat{{\small d}}^{{\small MLE}}$ & $\widehat{{\small d%
}}^{{\small PW}}$ \\ \hline
\multicolumn{15}{c}{} \\ 
\multicolumn{1}{c}{\small -0.9} & \multicolumn{1}{c}{\small -0.25} & 
\multicolumn{1}{c}{\small 96} & {\small 0.6233} & {\small 0.6345} & {\small %
0.6275} & {\small 0.6177} & {\small 0.6112} & {\small 0.6020} & {\small %
0.6284} & {\small 0.6385} & {\small 0.6086} & {\small 0.8247} & {\small 
\small \textbf{0.3671}%
} & {\small 
\small \textit{0.3729}%
} \\ 
\multicolumn{1}{c}{} & \multicolumn{1}{c}{} & \multicolumn{1}{c}{\small 576}
& {\small 0.4794} & {\small 0.4812} & {\small 0.4723} & {\small 0.4662} & 
{\small 0.4553} & {\small 0.4492} & {\small 0.4671} & {\small 0.4885} & 
{\small 0.4686} & {\small 0.4977} & {\small 
\small \textbf{0.1352}%
} & {\small 
\small \textit{0.1945}%
} \\ 
\multicolumn{1}{c}{} & \multicolumn{1}{c}{\small 0} & \multicolumn{1}{c}%
{\small 96} & {\small 0.7361} & {\small 0.8081} & {\small 0.7972} & {\small %
0.7875} & {\small 0.7726} & {\small 0.7642} & {\small 0.7815} & {\small %
0.8413} & {\small 0.7214} & {\small 0.8510} & {\small 
\small \textbf{0.6705}%
} & {\small 
\small \textit{0.6938}%
} \\ 
\multicolumn{1}{c}{} & \multicolumn{1}{c}{} & \multicolumn{1}{c}{\small 576}
& {\small 0.5687} & {\small 0.5919} & {\small 0.5822} & {\small 0.5719} & 
{\small 0.5641} & {\small 
\small \textit{0.5527}%
} & {\small 0.5637} & {\small 0.5838} & {\small 0.5639} & 
{\small 0.5942} & {\small 
\small \textbf{0.5426}%
} & {\small 0.5941} \\ 
\multicolumn{1}{c}{} & \multicolumn{1}{c}{\small 0.25} & \multicolumn{1}{c}%
{\small 96} & {\small 0.7996} & {\small 0.8096} & {\small 0.7918} & {\small %
0.7872} & {\small 0.7716} & {\small 
\small \textit{0.7615}%
} & {\small 0.7715} & {\small 0.8268} & {\small 0.7869} & 
{\small 0.8430} & {\small 
\small \textbf{0.7592}%
} & {\small 0.8081} \\ 
\multicolumn{1}{c}{} & \multicolumn{1}{c}{} & \multicolumn{1}{c}{\small 576}
& {\small 0.5951} & {\small 0.6193} & {\small 0.6022} & {\small 0.5976} & 
{\small 0.5843} & {\small 
\small \textit{0.5693}%
} & {\small 0.5826} & {\small 0.6219} & {\small 0.6019} & 
{\small 0.6590} & {\small 
\small \textbf{0.5513}%
} & {\small 0.5993} \\ 
\multicolumn{1}{c}{} & \multicolumn{1}{c}{\small 0.45} & \multicolumn{1}{c}%
{\small 96} & {\small 0.8219} & {\small 0.8410} & {\small 0.8325} & {\small %
0.8224} & {\small 0.8135} & {\small 
\small \textit{0.8064}%
} & {\small 0.8231} & {\small 0.8590} & {\small 0.8190} & 
{\small 0.8327} & {\small 
\small \textbf{0.7883}%
} & {\small 0.8184} \\ 
\multicolumn{1}{c}{} & \multicolumn{1}{c}{} & \multicolumn{1}{c}{\small 576}
& {\small 0.5950} & {\small 0.6066} & {\small 0.5953} & {\small 0.5871} & 
{\small 0.5763} & {\small 
\small \textit{0.5642}%
} & {\small 0.5783} & {\small 0.6298} & {\small 0.6198} & 
{\small 0.6487} & {\small 
\small \textbf{0.5609}%
} & {\small 0.6063} \\ 
\multicolumn{1}{c}{\small -0.4} & \multicolumn{1}{c}{\small -0.25} & 
\multicolumn{1}{c}{\small 96} & {\small 0.2376} & {\small 0.2253} & {\small %
0.2218} & {\small 0.2198} & {\small 0.2133} & {\small 0.2102} & {\small %
0.2401} & {\small 0.2488} & {\small 0.2255} & {\small 0.3103} & {\small 
\small \textbf{0.1682}%
} & {\small 
\small \textit{0.2094}%
} \\ 
\multicolumn{1}{c}{} & \multicolumn{1}{c}{} & \multicolumn{1}{c}{\small 576}
& {\small 0.1037} & {\small 0.0923} & {\small 0.0895} & {\small 0.0745} & 
{\small 0.0672} & {\small 
\small \textit{0.0652}%
} & {\small 0.1052} & {\small 0.1098} & {\small 0.1004} & 
{\small 0.1254} & {\small 
\small \textbf{0.0526}%
} & {\small 0.1027} \\ 
\multicolumn{1}{c}{} & \multicolumn{1}{c}{\small 0} & \multicolumn{1}{c}%
{\small 96} & {\small 0.2497} & {\small 0.2385} & {\small 0.2278} & {\small %
0.2142} & {\small 0.2136} & {\small 
\small \textit{0.2015}%
} & {\small 0.2514} & {\small 0.2559} & {\small 0.2512} & 
{\small 0.2883} & {\small 
\small \textbf{0.1644}%
} & {\small 0.2043} \\ 
\multicolumn{1}{c}{} & \multicolumn{1}{c}{} & \multicolumn{1}{c}{\small 576}
& {\small 0.1070} & {\small 0.0936} & {\small 0.0979} & {\small 0.0819} & 
{\small 0.0887} & {\small 
\small \textit{0.0778}%
} & {\small 0.1078} & {\small 0.1105} & {\small 0.0845} & 
{\small 0.1215} & {\small 
\small \textbf{0.0511}%
} & {\small 0.1011} \\ 
\multicolumn{1}{c}{} & \multicolumn{1}{c}{\small 0.25} & \multicolumn{1}{c}%
{\small 96} & {\small 0.2527} & {\small 0.2451} & {\small 0.2425} & {\small %
0.2379} & {\small 0.2343} & {\small 0.2335} & {\small 0.2535} & {\small %
0.2560} & {\small 0.2495} & {\small 0.2782} & {\small 
\small \textbf{0.1679}%
} & {\small 
\small \textit{0.2087}%
} \\ 
\multicolumn{1}{c}{} & \multicolumn{1}{c}{} & \multicolumn{1}{c}{\small 576}
& {\small 0.1068} & {\small 0.0987} & {\small 0.1052} & {\small 0.1057} & 
{\small 0.0964} & {\small 
\small \textit{0.0867}%
} & {\small 0.1067} & {\small 0.1103} & {\small 0.0934} & 
{\small 0.1199} & {\small 
\small \textbf{0.0518}%
} & {\small 0.1128} \\ 
\multicolumn{1}{c}{} & \multicolumn{1}{c}{\small 0.45} & \multicolumn{1}{c}%
{\small 96} & {\small 0.2496} & {\small 0.2524} & {\small 0.2459} & {\small %
0.2476} & {\small 0.2493} & {\small 0.2495} & {\small 0.2495} & {\small %
0.2518} & {\small 0.2441} & {\small 0.2725} & {\small 
\small \textbf{0.1682}%
} & {\small 
\small \textit{0.2093}%
} \\ 
\multicolumn{1}{c}{} & \multicolumn{1}{c}{} & \multicolumn{1}{c}{\small 576}
& {\small 0.1047} & {\small 0.0928} & {\small 0.0900} & {\small 0.0855} & 
{\small 0.0830} & {\small 
\small \textit{0.0740}%
} & {\small 0.1040} & {\small 0.1098} & {\small 0.0991} & 
{\small 0.1188} & {\small 
\small \textbf{0.0566}%
} & {\small 0.1066} \\ 
\multicolumn{1}{c}{\small 0.4} & \multicolumn{1}{c}{\small -0.25} & 
\multicolumn{1}{c}{\small 96} & {\small 0.1982} & {\small 0.1894} & {\small %
0.1875} & {\small 0.1825} & {\small 0.1793} & {\small 0.1687} & {\small %
0.1987} & {\small 0.2212} & {\small 0.2153} & {\small 0.2809} & {\small 
\small \textbf{0.1083}%
} & {\small 
\small \textit{0.1422}%
} \\ 
\multicolumn{1}{c}{} & \multicolumn{1}{c}{} & \multicolumn{1}{c}{\small 576}
& {\small 0.0932} & {\small 0.0858} & {\small 0.0988} & {\small 0.0947} & 
{\small 0.0935} & {\small 0.0933} & {\small 0.0933} & {\small 0.1078} & 
{\small 0.0812} & {\small 0.1268} & {\small 
\small \textbf{0.0594}%
} & {\small 
\small \textit{0.0739}%
} \\ 
\multicolumn{1}{c}{} & \multicolumn{1}{c}{\small 0} & \multicolumn{1}{c}%
{\small 96} & {\small 0.1944} & {\small 0.1826} & {\small 0.1815} & {\small %
0.1729} & {\small 0.1666} & {\small 0.1654} & {\small 0.1955} & {\small %
0.2203} & {\small 0.2146} & {\small 0.2701} & {\small 
\small \textbf{0.1042}%
} & {\small 
\small \textit{0.1492}%
} \\ 
\multicolumn{1}{c}{} & \multicolumn{1}{c}{} & \multicolumn{1}{c}{\small 576}
& {\small 0.0919} & {\small 0.0890} & {\small 0.0893} & {\small 0.0850} & 
{\small 0.0829} & {\small 0.0824} & {\small 0.0924} & {\small 0.1072} & 
{\small 0.0930} & {\small 0.1243} & {\small 
\small \textbf{0.0518}%
} & {\small 
\small \textit{0.0725}%
} \\ 
\multicolumn{1}{c}{} & \multicolumn{1}{c}{\small 0.25} & \multicolumn{1}{c}%
{\small 96} & {\small 0.1947} & {\small 0.1945} & {\small 0.1918} & {\small %
0.1878} & {\small 0.1780} & {\small 
\small \textit{0.1762}%
} & {\small 0.1962} & {\small 0.2213} & {\small 0.2048} & 
{\small 0.2663} & {\small 
\small \textbf{0.1015}%
} & {\small 0.1786} \\ 
\multicolumn{1}{c}{} & \multicolumn{1}{c}{} & \multicolumn{1}{c}{\small 576}
& {\small 0.0925} & {\small 0.0942} & {\small 0.1079} & {\small 0.0983} & 
{\small 0.0942} & {\small 0.0832} & {\small 0.0932} & {\small 0.1077} & 
{\small 0.0924} & {\small 0.1238} & {\small 
\small \textbf{0.0539}%
} & {\small 
\small \textit{0.0731}%
} \\ 
\multicolumn{1}{c}{} & \multicolumn{1}{c}{\small 0.45} & \multicolumn{1}{c}%
{\small 96} & {\small 0.1964} & {\small 0.1769} & {\small 0.1649} & {\small %
0.1544} & {\small 
\small \textit{0.1407}%
} & {\small 0.1483} & {\small 0.1984} & {\small 0.2223} & 
{\small 0.2175} & {\small 0.2643} & {\small 
\small \textbf{0.1028}%
} & {\small 0.1818} \\ 
\multicolumn{1}{c}{} & \multicolumn{1}{c}{} & \multicolumn{1}{c}{\small 576}
& {\small 0.0943} & {\small 0.0902} & {\small 0.0831} & {\small 0.0846} & 
{\small 0.0772} & {\small 
\small \textit{0.0756}%
} & {\small 0.0955} & {\small 0.1090} & {\small 0.0939} & 
{\small 0.1229} & {\small 
\small \textbf{0.0541}%
} & {\small 0.0857} \\ 
\multicolumn{1}{c}{\small 0.9} & \multicolumn{1}{c}{\small -0.25} & 
\multicolumn{1}{c}{\small 96} & {\small 0.0886} & {\small 0.0983} & {\small %
0.0944} & {\small 0.0907} & {\small 0.0883} & {\small 
\small \textit{0.0864}%
} & {\small 0.0938} & {\small 0.1105} & {\small 0.1073} & 
{\small 0.1253} & {\small 
\small \textbf{0.0543}%
} & {\small 0.0912} \\ 
\multicolumn{1}{c}{} & \multicolumn{1}{c}{} & \multicolumn{1}{c}{\small 576}
& {\small 
\small \textit{0.0344}%
} & {\small 0.0518} & {\small 0.0504} & {\small 0.0493} & 
{\small 0.0426} & {\small 0.0376} & {\small 0.0467} & {\small 0.0561} & 
{\small 0.0538} & {\small 0.0589} & {\small 
\small \textbf{0.0215}%
} & {\small 0.0421} \\ 
\multicolumn{1}{c}{} & \multicolumn{1}{c}{\small 0} & \multicolumn{1}{c}%
{\small 96} & {\small 0.0863} & {\small 0.1086} & {\small 0.1011} & {\small %
0.0977} & {\small 0.0945} & {\small 0.0912} & {\small 0.0975} & {\small %
0.1209} & {\small 0.1158} & {\small 0.1288} & {\small 
\small \textbf{0.0607}%
} & {\small 
\small \textit{0.0832}%
} \\ 
\multicolumn{1}{c}{} & \multicolumn{1}{c}{} & \multicolumn{1}{c}{\small 576}
& {\small 
\small \textit{0.0312}%
} & {\small 0.0572} & {\small 0.0542} & {\small 0.0519} & 
{\small 0.0482} & {\small 0.0460} & {\small 0.0493} & {\small 0.0674} & 
{\small 0.0625} & {\small 0.0729} & {\small 
\small \textbf{0.0202%
}} & {\small 0.0404} \\ 
\multicolumn{1}{c}{} & \multicolumn{1}{c}{\small 0.25} & \multicolumn{1}{c}%
{\small 96} & {\small 
\small \textit{0.0865}%
} & {\small 0.1113} & {\small 0.1086} & {\small 0.1012} & 
{\small 0.0974} & {\small 0.0928} & {\small 0.0972} & {\small 0.1287} & 
{\small 0.1176} & {\small 0.1286} & {\small 
\small \textbf{0.0653}%
} & {\small 0.0926} \\ 
\multicolumn{1}{c}{} & \multicolumn{1}{c}{} & \multicolumn{1}{c}{\small 576}
& {\small 
\small \textit{0.0304}%
} & {\small 0.0574} & {\small 0.0548} & {\small 0.0519} & 
{\small 0.0496} & {\small 0.0433} & {\small 0.0487} & {\small 0.0692} & 
{\small 0.0614} & {\small 0.0706} & {\small 
\small \textbf{0.0195}%
} & {\small 0.0451} \\ 
\multicolumn{1}{c}{} & \multicolumn{1}{c}{\small 0.45} & \multicolumn{1}{c}%
{\small 96} & {\small 
\small \textit{0.0885}%
} & {\small 0.1268} & {\small 0.1206} & {\small 0.1158} & 
{\small 0.1107} & {\small 0.1073} & {\small 0.1069} & {\small 0.1181} & 
{\small 0.1093} & {\small 0.1197} & {\small 
\small \textbf{0.0654}%
} & {\small 0.0933} \\ 
\multicolumn{1}{c}{} & \multicolumn{1}{c}{} & \multicolumn{1}{c}{\small 576}
& {\small 
\small \textit{0.0378}%
} & {\small 0.0592} & {\small 0.0541} & {\small 0.0517} & 
{\small 0.0482} & {\small 0.0459} & {\small 0.0528} & {\small 0.0647} & 
{\small 0.0580} & {\small 0.0695} & {\small 
\small \textbf{0.0122}%
} & {\small 0.0424} \\ \hline\hline
\end{tabular}

\end{table}%

\newpage

\begin{table}[h]\
\vspace{-7.0cm}%
\caption{{\small Bias estimates of the unadjusted LPR estimator, the feasible
jackknife estimator based on 2,3,4,6,8 non-overlapping (NO) sub-samples, the
feasible jackknife estimator based on 2 moving block (MB) sub-samples, both
versions of the GS estimator, the pre-filtered sieve bootstrap estimator,
the maximum likelihood estimator (MLE) and the pre-whitened (PW) estimator, for the
DGP: ARFIMA(}${\small 1,d}_{{\small 0}}{\small ,0}${\small ) with Gaussian
innovations. The feasible jackknife estimates are evaluated using the iterative procedure described in
Section 5.2.2. The estimates are obtained by setting }${\small \alpha }${\small \ = 0.65 and assuming the model is correctly specified. The lowest values
are \textbf{bold-faced} and the second lowest values are \textit{italicized}.}}%
\label{Table:Bias_ARFIMA(1,d,0)_b_I}%
\renewcommand{\arraystretch}{0.9}%
\setlength{\tabcolsep}{2.3pt}%

\medskip

\begin{tabular}{lllccccccccccc}
\hline\hline
&  &  &  &  &  &  &  &  &  &  &  &  &  \\ 
${\small \phi }_{{\small 0}}$ & ${\small d}_{{\small 0}}$ & ${\small n}$ & $%
\widehat{{\small d}}_{{\small n}}$ & $\widehat{{\small d}}_{{\small J,2}}^{%
{\small NO}}$ & $\widehat{{\small d}}_{{\small J,3}}^{{\small NO}}$ & $%
\widehat{{\small d}}_{{\small J,4}}^{{\small NO}}$ & $\widehat{{\small d}}_{%
{\small J,6}}^{{\small NO}}$ & $\widehat{{\small d}}_{{\small J,8}}^{{\small %
NO}}$ & $\widehat{{\small d}}_{{\small J,2}}^{{\small MB}}$ & $\widehat{%
{\small d}}_{{\small 1}}^{{\small GS}}$ & $\widehat{{\small d}}^{{\small PFSB%
}}$ & $\widehat{{\small d}}^{{\small MLE}}$ & $\widehat{{\small d}}^{{\small %
PW}}$ \\ \hline
\multicolumn{14}{c}{} \\ 
\multicolumn{1}{c}{\small -0.9} & \multicolumn{1}{c}{\small -0.25} & 
\multicolumn{1}{c}{\small 96} & \multicolumn{1}{r}{\small 0.8145} & 
\multicolumn{1}{r}{\small 0.8274} & \multicolumn{1}{r}{\small 0.8301} & 
\multicolumn{1}{r}{\small 0.8362} & \multicolumn{1}{r}{\small 0.8427} & 
\multicolumn{1}{r}{\small 0.8488} & \multicolumn{1}{r}{\small 0.8322} & 
\multicolumn{1}{r}{\small 0.8002} & \multicolumn{1}{r}{{\small 
\small \textit{0.7908}%
}} & \multicolumn{1}{r}{{\small 
\small \textbf{0.7331}%
}} & \multicolumn{1}{r}{\small 0.8022} \\
\hline\hline
\end{tabular}

\end{table}%

\newpage

\begin{table}[h]\
\vspace{-7.0cm}%
\caption{{\small RMSE estimates of the unadjusted LPR estimator, the feasible
jackknife estimator based on 2,3,4,6,8 non-overlapping (NO) sub-samples, the
feasible jackknife estimator based on 2 moving block (MB) sub-samples, both
versions of the GS estimator, the pre-filtered sieve bootstrap estimator,
the maximum likelihood estimator (MLE) and the pre-whitened (PW) estimator, for the
DGP: ARFIMA(}${\small 1,d}_{{\small 0}}{\small ,0}${\small ) with Gaussian
innovations. The feasible jackknife estimates are evaluated using the iterative procedure described in
Section 5.2.2. The estimates are obtained by setting }${\small \alpha }${\small \ = 0.65 and assuming the model is correctly specified. The lowest values
are \textbf{bold-faced} and the second lowest values are \textit{italicized}.}}%
\label{Table:MSE_ARFIMA(1,d,0)_b_I}%
\renewcommand{\arraystretch}{0.9}%
\setlength{\tabcolsep}{4pt}%

\medskip

\begin{tabular}{lllccccccccccc}
\hline\hline
&  &  &  &  &  &  &  &  &  &  &  &  &  \\ 
${\small \phi }_{{\small 0}}$ & ${\small d}_{{\small 0}}$ & ${\small n}$ & $%
\widehat{{\small d}}_{{\small n}}$ & $\widehat{{\small d}}_{{\small J,2}}^{%
{\small NO}}$ & $\widehat{{\small d}}_{{\small J,3}}^{{\small NO}}$ & $%
\widehat{{\small d}}_{{\small J,4}}^{{\small NO}}$ & $\widehat{{\small d}}_{%
{\small J,6}}^{{\small NO}}$ & $\widehat{{\small d}}_{{\small J,8}}^{{\small %
NO}}$ & $\widehat{{\small d}}_{{\small J,2}}^{{\small MB}}$ & $\widehat{%
{\small d}}_{{\small 1}}^{{\small GS}}$ & $\widehat{{\small d}}^{{\small PFSB%
}}$ & $\widehat{{\small d}}^{{\small MLE}}$ & $\widehat{{\small d}}^{{\small %
PW}}$ \\ \hline
\multicolumn{1}{c}{} &  &  &  &  &  &  &  &  &  &  &  &  &  \\ 
\multicolumn{1}{c}{\small -0.9} & \multicolumn{1}{c}{\small -0.25} & 
\multicolumn{1}{c}{\small 96} & {\small 
\small \textit{1.0359}%
} & {\small 1.2162} & {\small 1.2053} & {\small 1.1907} & 
{\small 1.1853} & {\small 1.1814} & {\small 1.2193} & {\small 1.3386} & 
{\small 1.2885} & {\small 
\small \textbf{0.9538}%
} & {\small 1.0562} \\ 
\multicolumn{1}{c}{} & \multicolumn{1}{c}{} & \multicolumn{1}{c}{\small 576}
& {\small 0.7398} & {\small 0.7688} & {\small 0.7621} & {\small 0.7586} & 
{\small 0.7549} & {\small 0.7451} & {\small 0.7618} & {\small 0.7371} & 
{\small 
\small \textit{0.7359}%
} & {\small 
\small \textbf{0.6315}%
} & {\small 0.7365} \\ 
\multicolumn{1}{c}{} & \multicolumn{1}{c}{\small 0} & \multicolumn{1}{c}%
{\small 96} & {\small 1.1148} & {\small 1.1491} & {\small 1.1400} & {\small %
1.1365} & {\small 1.1334} & {\small 1.1240} & {\small 1.1391} & {\small %
1.1819} & {\small 1.2167} & {\small 
\small \textbf{0.9243}%
} & {\small 
\small \textit{0.9847}%
} \\ 
\multicolumn{1}{c}{} & \multicolumn{1}{c}{} & \multicolumn{1}{c}{\small 576}
& {\small 0.8288} & {\small 0.8418} & {\small 0.8397} & {\small 0.8322} & 
{\small 0.8282} & {\small 0.8239} & {\small 0.8306} & {\small 0.8173} & 
{\small 0.8053} & {\small 
\small \textbf{0.6221}%
} & {\small 
\small \textit{0.6648}%
} \\ 
\multicolumn{1}{c}{} & \multicolumn{1}{c}{\small 0.25} & \multicolumn{1}{c}%
{\small 96} & {\small 1.1618} & {\small 1.1835} & {\small 1.1773} & {\small %
1.1666} & {\small 1.1537} & {\small 1.1428} & {\small 1.1588} & {\small %
1.1484} & {\small 1.2299} & {\small 
\small \textbf{0.9428}%
} & {\small 
\small \textit{1.0275}%
} \\ 
\multicolumn{1}{c}{} & \multicolumn{1}{c}{} & \multicolumn{1}{c}{\small 576}
& {\small 0.9175} & {\small 0.9409} & {\small 0.9348} & {\small 0.9275} & 
{\small 0.9212} & {\small 0.9187} & {\small 0.9334} & {\small 1.1171} & 
{\small 1.1130} & {\small 
\small \textbf{0.6385}%
} & {\small 
\small \textit{0.6611}%
} \\ 
\multicolumn{1}{c}{} & \multicolumn{1}{c}{\small 0.45} & \multicolumn{1}{c}%
{\small 96} & {\small 1.1286} & {\small 1.2150} & {\small 1.2061} & {\small %
1.1982} & {\small 1.1933} & {\small 1.1869} & {\small 1.2186} & {\small %
1.4331} & {\small 1.5385} & {\small 
\small \textbf{0.9382}%
} & {\small 
\small \textit{1.0335}%
} \\ 
\multicolumn{1}{c}{} & \multicolumn{1}{c}{} & \multicolumn{1}{c}{\small 576}
& {\small 0.9708} & {\small 0.9842} & {\small 0.9711} & {\small 0.9672} & 
{\small 0.9627} & {\small 0.9574} & {\small 0.9775} & {\small 1.1124} & 
{\small 1.1647} & {\small 
\small \textbf{0.6415}%
} & {\small 
\small \textit{0.7069}%
} \\ 
\multicolumn{1}{c}{\small -0.4} & \multicolumn{1}{c}{\small -0.25} & 
\multicolumn{1}{c}{\small 96} & {\small 
\small \textit{0.2568}%
} & {\small 0.2841} & {\small 0.2726} & {\small 0.2699} & 
{\small 0.2606} & {\small 0.2515} & {\small 0.2635} & {\small 0.2594} & 
{\small 0.3028} & {\small 
\small \textbf{0.1863}%
} & {\small 0.2671} \\ 
\multicolumn{1}{c}{} & \multicolumn{1}{c}{} & \multicolumn{1}{c}{\small 576}
& {\small 
\small \textit{0.1098}%
} & {\small 0.1249} & {\small 0.1168} & {\small 0.1134} & 
{\small 0.1121} & {\small 0.1276} & {\small 0.1149} & {\small 0.1118} & 
{\small 0.1272} & {\small 
\small \textbf{0.0946}%
} & {\small 0.1339} \\ 
\multicolumn{1}{c}{} & \multicolumn{1}{c}{\small 0} & \multicolumn{1}{c}%
{\small 96} & {\small 0.2498} & {\small 0.2772} & {\small 0.2724} & {\small %
0.2685} & {\small 0.2576} & {\small 
\small \textit{0.2418}%
} & {\small 0.2643} & {\small 0.2560} & {\small 0.2930} & 
{\small 
\small \textbf{0.1792}%
} & {\small 0.2496} \\ 
\multicolumn{1}{c}{} & \multicolumn{1}{c}{} & \multicolumn{1}{c}{\small 576}
& {\small 0.1069} & {\small 0.1278} & {\small 0.1218} & {\small 0.1106} & 
{\small 0.1073} & {\small 
\small \textit{0.1005}%
} & {\small 0.1055} & {\small 0.1104} & {\small 0.1247} & 
{\small 
\small \textbf{0.0867}%
} & {\small 0.1348} \\ 
\multicolumn{1}{c}{} & \multicolumn{1}{c}{\small 0.25} & \multicolumn{1}{c}%
{\small 96} & {\small 0.2490} & {\small 0.2835} & {\small 0.2782} & {\small %
0.2737} & {\small 0.2688} & {\small 0.2630} & {\small 0.3108} & {\small %
0.2580} & {\small 0.2879} & {\small 
\small \textbf{0.1814}%
} & {\small 
\small \textit{0.2473}%
} \\ 
\multicolumn{1}{c}{} & \multicolumn{1}{c}{} & \multicolumn{1}{c}{\small 576}
& {\small 0.1079} & {\small 0.1374} & {\small 0.1326} & {\small 0.1248} & 
{\small 0.1160} & {\small 
\small \textit{0.1053}%
} & {\small 0.1164} & {\small 0.1115} & {\small 0.1239} & 
{\small 
\small \textbf{0.0992}%
} & {\small 0.1340} \\ 
\multicolumn{1}{c}{} & \multicolumn{1}{c}{\small 0.45} & \multicolumn{1}{c}%
{\small 96} & {\small 0.2506} & {\small 0.2833} & {\small 0.2761} & {\small %
0.2619} & {\small 0.2598} & {\small 0.2541} & {\small 0.2759} & {\small %
0.2616} & {\small 0.2506} & {\small 
\small \textbf{0.1836}%
} & {\small 
\small \textit{0.2497}%
} \\ 
\multicolumn{1}{c}{} & \multicolumn{1}{c}{} & \multicolumn{1}{c}{\small 576}
& {\small 
\small \textit{0.1115}%
} & {\small 0.1428} & {\small 0.1411} & {\small 0.1337} & 
{\small 0.1276} & {\small 0.1128} & {\small 0.1221} & {\small 0.1143} & 
{\small 0.1230} & {\small 
\small \textbf{0.0934}%
} & {\small 0.1306} \\ 
\multicolumn{1}{c}{\small 0.4} & \multicolumn{1}{c}{\small -0.25} & 
\multicolumn{1}{c}{\small 96} & {\small 
\small \textit{0.1917}%
} & {\small 0.2350} & {\small 0.2335} & {\small 0.2278} & 
{\small 0.2210} & {\small 0.2172} & {\small 0.2266} & {\small 0.2212} & 
{\small 0.2717} & {\small 
\small \textbf{0.1244}%
} & {\small 0.1984} \\ 
\multicolumn{1}{c}{} & \multicolumn{1}{c}{} & \multicolumn{1}{c}{\small 576}
& {\small 0.0919} & {\small 0.1229} & {\small 0.1189} & {\small 0.1144} & 
{\small 0.1075} & {\small 0.1035} & {\small 0.1020} & {\small 0.1081} & 
{\small 0.1198} & {\small 
\small \textbf{0.0531}%
} & {\small 
\small \textit{0.0909}%
} \\ 
\multicolumn{1}{c}{} & \multicolumn{1}{c}{\small 0} & \multicolumn{1}{c}%
{\small 96} & {\small 0.1946} & {\small 0.2295} & {\small 0.2251} & {\small %
0.2177} & {\small 0.2114} & {\small 0.2001} & {\small 0.2163} & {\small %
0.2203} & {\small 0.2546} & {\small 
\small \textbf{0.1194}%
} & {\small 
\small \textit{0.1939}%
} \\ 
\multicolumn{1}{c}{} & \multicolumn{1}{c}{} & \multicolumn{1}{c}{\small 576}
& {\small 
\small \textit{0.0920}%
} & {\small 0.1246} & {\small 0.1208} & {\small 0.1145} & 
{\small 0.1185} & {\small 0.1099} & {\small 0.1176} & {\small 0.1073} & 
{\small 0.1166} & {\small 
\small \textbf{0.0617}%
} & {\small 0.0946} \\ 
\multicolumn{1}{c}{} & \multicolumn{1}{c}{\small 0.25} & \multicolumn{1}{c}%
{\small 96} & {\small 
\small \textit{0.1960}%
} & {\small 0.2281} & {\small 0.2219} & {\small 0.2163} & 
{\small 0.2267} & {\small 0.2296} & {\small 0.2225} & {\small 0.2209} & 
{\small 0.2482} & {\small 
\small \textbf{0.1273}%
} & {\small 0.2088} \\ 
\multicolumn{1}{c}{} & \multicolumn{1}{c}{} & \multicolumn{1}{c}{\small 576}
& {\small 
\small \textit{0.0922}%
} & {\small 0.1168} & {\small 0.1113} & {\small 0.1087} & 
{\small 0.1055} & {\small 0.1019} & {\small 0.1150} & {\small 0.1076} & 
{\small 0.1158} & {\small 
\small \textbf{0.0566}%
} & {\small 0.0999} \\ 
\multicolumn{1}{c}{} & \multicolumn{1}{c}{\small 0.45} & \multicolumn{1}{c}%
{\small 96} & {\small 
\small \textit{0.1955}%
} & {\small 0.2379} & {\small 0.2318} & {\small 0.2206} & 
{\small 0.2284} & {\small 0.2178} & {\small 0.2174} & {\small 0.2218} & 
{\small 0.2453} & {\small 
\small \textbf{0.1282}%
} & {\small 0.2094} \\ 
\multicolumn{1}{c}{} & \multicolumn{1}{c}{} & \multicolumn{1}{c}{\small 576}
& {\small 0.0926} & {\small 0.1241} & {\small 0.1241} & {\small 0.1179} & 
{\small 0.1055} & {\small 0.1013} & {\small 0.1084} & {\small 0.1089} & 
{\small 0.1149} & {\small 
\small \textbf{0.0476}%
} & {\small 
\small \textit{0.0913}%
} \\ 
\multicolumn{1}{c}{\small 0.9} & \multicolumn{1}{c}{\small -0.25} & 
\multicolumn{1}{c}{\small 96} & {\small 
\small \textit{0.1115}%
} & {\small 0.1385} & {\small 0.1306} & {\small 0.1282} & 
{\small 0.1243} & {\small 0.1210} & {\small 0.1284} & {\small 0.1365} & 
{\small 0.1266} & {\small 
\small \textbf{0.0712}%
} & {\small 0.1160} \\ 
\multicolumn{1}{c}{} & \multicolumn{1}{c}{} & \multicolumn{1}{c}{\small 576}
& {\small 0.0624} & {\small 0.0687} & {\small 0.0660} & {\small 0.0616} & 
{\small 0.0579} & {\small 
\small \textit{0.0548}%
} & {\small 0.0599} & {\small 0.0708} & {\small 0.0600} & 
{\small 
\small \textbf{0.0369}%
} & {\small 0.0649} \\ 
\multicolumn{1}{c}{} & \multicolumn{1}{c}{\small 0} & \multicolumn{1}{c}%
{\small 96} & {\small 
\small \textit{0.1010}%
} & {\small 0.1162} & {\small 0.1123} & {\small 0.1105} & 
{\small 0.1088} & {\small 0.1023} & {\small 0.1187} & {\small 0.1121} & 
{\small 0.1093} & {\small 
\small \textbf{0.0681}%
} & {\small 0.1097} \\ 
\multicolumn{1}{c}{} & \multicolumn{1}{c}{} & \multicolumn{1}{c}{\small 576}
& {\small 0.0602} & {\small 0.0629} & {\small 0.0609} & {\small 0.0549} & 
{\small 
\small \textit{0.0533}%
} & {\small 0.0521} & {\small 0.0577} & {\small 0.0698} & 
{\small 0.0705} & {\small 
\small \textbf{0.0344}%
} & {\small 0.0539} \\ 
\multicolumn{1}{c}{} & \multicolumn{1}{c}{\small 0.25} & \multicolumn{1}{c}%
{\small 96} & {\small 
\small \textit{0.1114}%
} & {\small 0.1324} & {\small 0.1318} & {\small 0.1268} & 
{\small 0.1222} & {\small 0.1198} & {\small 0.1229} & {\small 0.1328} & 
{\small 0.1282} & {\small 
\small \textbf{0.0771}%
} & {\small 0.1142} \\ 
\multicolumn{1}{c}{} & \multicolumn{1}{c}{} & \multicolumn{1}{c}{\small 576}
& {\small 
\small \textit{0.0518}%
} & {\small 0.0695} & {\small 0.0641} & {\small 0.0616} & 
{\small 0.0589} & {\small 0.0552} & {\small 0.0572} & {\small 0.0626} & 
{\small 0.0581} & {\small 
\small \textbf{0.0322}%
} & {\small 0.0530} \\ 
\multicolumn{1}{c}{} & \multicolumn{1}{c}{\small 0.45} & \multicolumn{1}{c}%
{\small 96} & {\small 0.1053} & {\small 0.1284} & {\small 0.1307} & {\small %
0.1284} & {\small 0.1229} & {\small 0.1216} & {\small 0.1179} & {\small %
0.1253} & {\small 0.1215} & {\small 
\small \textbf{0.0725}%
} & {\small 
\small \textit{0.1031}%
} \\ 
\multicolumn{1}{c}{} & \multicolumn{1}{c}{} & \multicolumn{1}{c}{\small 576}
& {\small 
\small \textit{0.0526}%
} & {\small 0.0681} & {\small 0.0635} & {\small 0.0610} & 
{\small 0.0595} & {\small 0.0549} & {\small 0.0581} & {\small 0.0769} & 
{\small 0.0684} & {\small 
\small \textbf{0.0349}%
} & {\small 0.0528} \\ \hline\hline
\end{tabular}

\end{table}%

\newpage

\begin{table}[h]\
\vspace{-7.0cm}%
\caption{{\small Bias estimates of the unadjusted LPR estimator, the feasible
jackknife estimator based on 2,3,4,6,8 non-overlapping (NO) sub-samples, the
feasible jackknife estimator based on 2 moving block (MB) sub-samples, both
versions of the GS estimator, the pre-filtered sieve bootstrap estimator,
the maximum likelihood estimator (MLE) and the pre-whitened (PW) estimator, for the
DGP: ARFIMA(}${\small 0,d}_{{\small 0}}{\small ,1}${\small ) with Gaussian
innovations. The feasible jackknife estimates are evaluated using the iterative procedure described in
Section 5.2.2. The estimates are obtained by setting }${\small \alpha }${\small \ = 0.65 and assuming the model is correctly specified. The lowest values
are \textbf{bold-faced} and the second lowest values are \textit{italicized}.}}%
\label{Table:Bias_ARFIMA(0,d,1)_b_I}%
\renewcommand{\arraystretch}{0.9}%
\setlength{\tabcolsep}{2.3pt}%

\medskip

\begin{tabular}{lllccccccccccc}
\hline\hline
&  &  &  &  &  &  &  &  &  &  &  &  &  \\ 
${\small \theta }_{{\small 0}}$ & ${\small d}_{{\small 0}}$ & ${\small n}$ & 
$\widehat{{\small d}}_{{\small n}}$ & $\widehat{{\small d}}_{{\small J,2}}^{%
{\small NO}}$ & $\widehat{{\small d}}_{{\small J,3}}^{{\small NO}}$ & $%
\widehat{{\small d}}_{{\small J,4}}^{{\small NO}}$ & $\widehat{{\small d}}_{%
{\small J,6}}^{{\small NO}}$ & $\widehat{{\small d}}_{{\small J,8}}^{{\small %
NO}}$ & $\widehat{{\small d}}_{{\small J,2}}^{{\small MB}}$ & $\widehat{%
{\small d}}_{{\small 1}}^{{\small GS}}$ & $\widehat{{\small d}}^{{\small PFSB%
}}$ & $\widehat{{\small d}}^{{\small MLE}}$ & $\widehat{{\small d}}^{{\small %
PW}}$ \\ \hline
\multicolumn{14}{c}{} \\ 
\multicolumn{1}{c}{\small -0.9} & \multicolumn{1}{c}{\small -0.25} & 
\multicolumn{1}{c}{\small 96} & \multicolumn{1}{r}{\small -0.5671} & 
\multicolumn{1}{r}{\small -0.5487} & \multicolumn{1}{r}{\small -0.5519} & 
\multicolumn{1}{r}{\small -0.5586} & \multicolumn{1}{r}{\small -0.5627} & 
\multicolumn{1}{r}{\small -0.5738} & \multicolumn{1}{r}{\small -0.5862} & 
\multicolumn{1}{r}{{\small 
\small \textit{-0.5450}%
}} & \multicolumn{1}{r}{\small -0.5466} & \multicolumn{1}{r}{%
{\small 
\small \textbf{-0.5372}%
}} & \multicolumn{1}{r}{\small -0.6278} \\
\hline\hline
\end{tabular}

\end{table}%

\newpage

\begin{table}[h]\
\vspace{-7.0cm}%
\caption{{\small RMSE estimates of the unadjusted LPR estimator, the feasible
jackknife estimator based on 2,3,4,6,8 non-overlapping (NO) sub-samples, the
feasible jackknife estimator based on 2 moving block (MB) sub-samples, both
versions of the GS estimator, the pre-filtered sieve bootstrap estimator,
the maximum likelihood estimator (MLE) and the pre-whitened (PW) estimator, for the
DGP: ARFIMA(}${\small 0,d}_{{\small 0}}{\small ,1}${\small ) with Gaussian
innovations. The feasible jackknife estimates are evaluated using the iterative procedure described in
Section 5.2.2. The estimates are obtained by setting }${\small \alpha }${\small \ = 0.65 and assuming the model is correctly specified. The lowest values
are \textbf{bold-faced} and the second lowest values are \textit{italicized}.}}%
\label{Table:MSE_ARFIMA(0,d,1)_b_I}%
\renewcommand{\arraystretch}{0.9}%
\setlength{\tabcolsep}{4pt}%

\medskip

\begin{tabular}{lllccccccccccc}
\hline\hline
&  &  &  &  &  &  &  &  &  &  &  &  &  \\ 
${\small \theta }_{{\small 0}}$ & ${\small d}_{{\small 0}}$ & ${\small n}$ & 
$\widehat{{\small d}}_{{\small n}}$ & $\widehat{{\small d}}_{{\small J,2}}^{%
{\small NO}}$ & $\widehat{{\small d}}_{{\small J,3}}^{{\small NO}}$ & $%
\widehat{{\small d}}_{{\small J,4}}^{{\small NO}}$ & $\widehat{{\small d}}_{%
{\small J,6}}^{{\small NO}}$ & $\widehat{{\small d}}_{{\small J,8}}^{{\small %
NO}}$ & $\widehat{{\small d}}_{{\small J,2}}^{{\small MB}}$ & $\widehat{%
{\small d}}_{{\small 1}}^{{\small GS}}$ & $\widehat{{\small d}}^{{\small PFSB%
}}$ & $\widehat{{\small d}}^{{\small MLE}}$ & $\widehat{{\small d}}^{{\small %
PW}}$ \\ \hline
\multicolumn{1}{c}{} &  &  &  &  &  &  &  &  &  &  &  &  &  \\ 
\multicolumn{1}{c}{\small -0.9} & \multicolumn{1}{c}{\small -0.25} & 
\multicolumn{1}{c}{\small 96} & {\small 
\small \textit{0.6233}%
} & {\small 0.6561} & {\small 0.6463} & {\small 0.6405} & 
{\small 0.6348} & {\small 0.6319} & {\small 0.6653} & {\small 0.6385} & 
{\small 0.8247} & {\small 
\small \textbf{0.5982}%
} & {\small 0.6663} \\ 
\multicolumn{1}{c}{} & \multicolumn{1}{c}{} & \multicolumn{1}{c}{\small 576}
& {\small 0.4794} & {\small 0.4980} & {\small 0.4927} & {\small 0.4854} & 
{\small 0.4771} & {\small 
\small \textit{0.4727}%
} & {\small 0.4785} & {\small 0.4888} & {\small 0.4977} & 
{\small 
\small \textbf{0.4406}%
} & {\small 0.4858} \\ 
\multicolumn{1}{c}{} & \multicolumn{1}{c}{\small 0} & \multicolumn{1}{c}%
{\small 96} & {\small 
\small \textit{0.7361}%
} & {\small 0.8327} & {\small 0.8291} & {\small 0.8247} & 
{\small 0.8189} & {\small 0.8006} & {\small 0.8114} & {\small 0.8413} & 
{\small 0.8510} & {\small 
\small \textbf{0.7234} %
}& {\small 0.7604} \\ 
\multicolumn{1}{c}{} & \multicolumn{1}{c}{} & \multicolumn{1}{c}{\small 576}
& {\small 
\small \textit{0.5687}%
} & {\small 0.6371} & {\small 0.6034} & {\small 0.6152} & 
{\small 0.6038} & 0{\small .5972} & {\small 0.6241} & {\small 0.5838} & 
{\small 0.5942} & {\small 
\small \textbf{0.5621}%
} & {\small 0.6329} \\ 
\multicolumn{1}{c}{} & \multicolumn{1}{c}{\small 0.25} & \multicolumn{1}{c}%
{\small 96} & {\small 0.7996} & {\small 0.8238} & {\small 0.8186} & {\small %
0.8013} & {\small 0.7926} & {\small 
\small \textit{0.7884}%
} & {\small 0.8108} & {\small 0.8268} & {\small 0.8430} & 
{\small 
\small \textbf{0.7429}%
} & {\small 0.8215} \\ 
\multicolumn{1}{c}{} & \multicolumn{1}{c}{} & \multicolumn{1}{c}{\small 576}
& {\small 0.5951} & {\small 0.6339} & {\small 0.6257} & {\small 0.6108} & 
{\small 0.6075} & {\small 
\small \textit{0.5922}%
} & {\small 0.6249} & {\small 0.6219} & {\small 0.6590} & 
{\small 
\small \textbf{0.5513}%
} & {\small 0.6359} \\ 
\multicolumn{1}{c}{} & \multicolumn{1}{c}{\small 0.45} & \multicolumn{1}{c}%
{\small 96} & {\small 0.8219} & {\small 0.8562} & {\small 0.8414} & {\small %
0.8393} & {\small 0.8242} & {\small 
\small \textit{0.8107}%
} & {\small 0.8233} & {\small 0.8590} & {\small 0.8327} & 
{\small 
\small \textbf{0.8107}%
} & {\small 0.8353} \\ 
\multicolumn{1}{c}{} & \multicolumn{1}{c}{} & \multicolumn{1}{c}{\small 576}
& {\small 
\small \textit{0.5950}%
} & {\small 0.6384} & {\small 0.6279} & {\small 0.6211} & 
{\small 0.6184} & {\small 0.6124} & {\small 0.6589} & {\small 0.6298} & 
{\small 0.6487} & {\small 
\small \textbf{0.5918}%
} & {\small 0.6337} \\ 
\multicolumn{1}{c}{\small -0.4} & \multicolumn{1}{c}{\small -0.25} & 
\multicolumn{1}{c}{\small 96} & {\small 0.2376} & {\small 0.2517} & {\small %
0.2441} & {\small 0.2384} & {\small 0.2300} & {\small 
\small \textit{0.2283}%
} & {\small 0.2566} & {\small 0.2488} & {\small 0.3103} & 
{\small 
\small \textbf{0.2173}%
} & {\small 0.2923} \\ 
\multicolumn{1}{c}{} & \multicolumn{1}{c}{} & \multicolumn{1}{c}{\small 576}
& {\small 
\small \textit{0.1037}%
} & {\small 0.1352} & {\small 0.1239} & {\small 0.1192} & 
{\small 0.1116} & {\small 0.1085} & {\small 0.1232} & {\small 0.1098} & 
{\small 0.1254} & {\small 
\small \textbf{0.0954}%
} & {\small 0.1326} \\ 
\multicolumn{1}{c}{} & \multicolumn{1}{c}{\small 0} & \multicolumn{1}{c}%
{\small 96} & {\small 0.2497} & {\small 0.2743} & {\small 0.2662} & {\small %
0.2545} & {\small 0.2449} & {\small 
\small \textit{0.2406}%
} & {\small 0.2457} & {\small 0.2559} & {\small 0.2883} & 
{\small 
\small \textbf{0.2284}%
} & {\small 0.2873} \\ 
\multicolumn{1}{c}{} & \multicolumn{1}{c}{} & \multicolumn{1}{c}{\small 576}
& {\small 0.1070} & {\small 0.1366} & {\small 0.1245} & {\small 0.1184} & 
{\small 0.1105} & {\small 
\small \textit{0.1044}%
} & {\small 0.1193} & {\small 0.1105} & {\small 0.1215} & 
{\small 
\small \textbf{0.1020}%
} & {\small 0.1478} \\ 
\multicolumn{1}{c}{} & \multicolumn{1}{c}{\small 0.25} & \multicolumn{1}{c}%
{\small 96} & {\small 0.2527} & {\small 0.2719} & {\small 0.2636} & {\small %
0.2591} & {\small 0.2530} & {\small 
\small \textit{0.2418}%
} & {\small 0.2495} & {\small 0.2648} & {\small 0.2782} & 
{\small 
\small \textbf{0.2305}%
} & {\small 0.2995} \\ 
\multicolumn{1}{c}{} & \multicolumn{1}{c}{} & \multicolumn{1}{c}{\small 576}
& {\small 
\small \textit{0.1068}%
} & {\small 0.1384} & {\small 0.1315} & {\small 0.1251} & 
{\small 0.1144} & {\small 0.1084} & {\small 0.1267} & {\small 0.1103} & 
{\small 0.1199} & {\small 
\small \textbf{0.1032}%
} & {\small 0.1526} \\ 
\multicolumn{1}{c}{} & \multicolumn{1}{c}{\small 0.45} & \multicolumn{1}{c}%
{\small 96} & {\small 
\small \textit{0.2496}%
} & {\small 0.2737} & {\small 0.2608} & {\small 0.2549} & 
{\small 0.2521} & {\small 0.2512} & {\small 0.2528} & {\small 0.2518} & 
{\small 0.2725} & {\small 
\small \textbf{0.2311}%
} & {\small 0.2880} \\ 
\multicolumn{1}{c}{} & \multicolumn{1}{c}{} & \multicolumn{1}{c}{\small 576}
& {\small 
\small \textit{0.1047}%
} & {\small 0.1392} & {\small 0.1384} & {\small 0.1300} & 
{\small 0.1243} & {\small 0.1154} & {\small 0.1276} & {\small 0.1098} & 
{\small 0.1188} & {\small 
\small \textbf{0.1044}%
} & {\small 0.1539} \\ 
\multicolumn{1}{c}{\small 0.4} & \multicolumn{1}{c}{\small -0.25} & 
\multicolumn{1}{c}{\small 96} & {\small 
\small \textit{0.1982}%
} & {\small 0.2348} & {\small 0.2217} & {\small 0.2106} & 
{\small 0.2058} & {\small 0.2007} & {\small 0.2245} & {\small 0.2212} & 
{\small 0.2809} & {\small 
\small \textbf{0.1729}%
} & {\small 0.2035} \\ 
\multicolumn{1}{c}{} & \multicolumn{1}{c}{} & \multicolumn{1}{c}{\small 576}
& {\small 
\small \textit{0.0932}%
} & {\small 0.1079} & {\small 0.1155} & {\small 0.1249} & 
{\small 0.1163} & {\small 0.1096} & {\small 0.0972} & {\small 0.1078} & 
{\small 0.1268} & {\small 
\small \textbf{0.0883}%
} & {\small 0.1185} \\ 
\multicolumn{1}{c}{} & \multicolumn{1}{c}{\small 0} & \multicolumn{1}{c}%
{\small 96} & {\small 
\small \textit{0.1944}%
} & {\small 0.2315} & {\small 0.2242} & {\small 0.2215} & 
{\small 0.2194} & {\small 0.2072} & {\small 0.2159} & {\small 0.2203} & 
{\small 0.2701} & {\small 
\small \textbf{0.1637}%
} & {\small 0.1927} \\ 
\multicolumn{1}{c}{} & \multicolumn{1}{c}{} & \multicolumn{1}{c}{\small 576}
& {\small 
\small \textit{0.0919}%
} & {\small 0.1232} & {\small 0.1119} & {\small 0.1076} & 
{\small 0.1026} & {\small 0.0944} & {\small 0.1036} & {\small 0.1072} & 
{\small 0.1243} & {\small 
\small \textbf{0.0906}%
} & {\small 0.1053} \\ 
\multicolumn{1}{c}{} & \multicolumn{1}{c}{\small 0.25} & \multicolumn{1}{c}%
{\small 96} & {\small 0.1947} & {\small 0.2224} & {\small 0.2153} & {\small %
0.2018} & {\small 0.1982} & {\small 
\small \textit{0.1902}%
} & {\small 0.2247} & {\small 0.2213} & {\small 0.2663} & 
{\small 
\small \textbf{0.1625}%
} & {\small 0.1901} \\ 
\multicolumn{1}{c}{} & \multicolumn{1}{c}{} & \multicolumn{1}{c}{\small 576}
& {\small 
\small \textit{0.0925}%
} & {\small 0.1105} & {\small 0.1151} & {\small 0.1172} & 
{\small 0.1069} & {\small 0.1010} & {\small 0.1108} & {\small 0.1077} & 
{\small 0.1238} & {\small 
\small \textbf{0.0853}%
} & {\small 0.1136} \\ 
\multicolumn{1}{c}{} & \multicolumn{1}{c}{\small 0.45} & \multicolumn{1}{c}%
{\small 96} & {\small 0.1964} & {\small 0.2247} & {\small 0.2172} & {\small %
0.2033} & {\small 0.1946} & {\small 
\small \textit{0.1916}%
} & {\small 0.2265} & {\small 0.2223} & {\small 0.2643} & 
{\small 
\small \textbf{0.1639}%
} & {\small 0.2084} \\ 
\multicolumn{1}{c}{} & \multicolumn{1}{c}{} & \multicolumn{1}{c}{\small 576}
& {\small 
\small \textit{0.0943}%
} & {\small 0.1221} & {\small 0.1134} & {\small 0.1076} & 
{\small 0.1016} & {\small 0.0992} & {\small 0.1157} & {\small 0.1090} & 
{\small 0.1229} & {\small 
\small \textbf{0.0927}%
} & {\small 0.1120} \\ 
\multicolumn{1}{c}{\small 0.9} & \multicolumn{1}{c}{\small -0.25} & 
\multicolumn{1}{c}{\small 96} & {\small 
\small \textit{0.0886}%
} & {\small 0.1213} & {\small 0.1189} & {\small 0.1083} & 
{\small 0.1035} & {\small 0.0953} & {\small 0.1166} & {\small 0.1105} & 
{\small 0.1253} & {\small 
\small \textbf{0.0834}%
} & {\small 0.1117} \\ 
\multicolumn{1}{c}{} & \multicolumn{1}{c}{} & \multicolumn{1}{c}{\small 576}
& {\small 
\small \textit{0.0344}%
} & {\small 0.0695} & {\small 0.0637} & {\small 0.0559} & 
{\small 0.0521} & {\small 0.0486} & {\small 0.0578} & {\small 0.0561} & 
{\small 0.0589} & {\small 
\small \textbf{0.0315}%
} & {\small 0.0538} \\ 
\multicolumn{1}{c}{} & \multicolumn{1}{c}{\small 0} & \multicolumn{1}{c}%
{\small 96} & {\small 
\small \textit{0.0863}%
} & {\small 0.1136} & {\small 0.1120} & {\small 0.1084} & 
{\small 0.1043} & {\small 0.0997} & {\small 0.1045} & {\small 0.1209} & 
{\small 0.1288} & {\small 
\small \textbf{0.0807}%
} & {\small 0.1013} \\ 
\multicolumn{1}{c}{} & \multicolumn{1}{c}{} & \multicolumn{1}{c}{\small 576}
& {\small 
\small \textit{0.0312}%
} & {\small 0.0613} & {\small 0.0595} & {\small 0.0546} & 
{\small 0.0504} & {\small 0.0488} & {\small 0.0491} & {\small 0.0674} & 
{\small 0.0729} & {\small 
\small \textbf{0.0299}%
} & {\small 0.0526} \\ 
\multicolumn{1}{c}{} & \multicolumn{1}{c}{\small 0.25} & \multicolumn{1}{c}%
{\small 96} & {\small 
\small \textit{0.0865}%
} & {\small 0.1224} & {\small 0.1210} & {\small 0.1171} & 
{\small 0.1123} & {\small 0.1161} & {\small 0.1140} & {\small 0.1287} & 
{\small 0.1286} & {\small 
\small \textbf{0.0818}%
} & {\small 0.1157} \\ 
\multicolumn{1}{c}{} & \multicolumn{1}{c}{} & \multicolumn{1}{c}{\small 576}
& {\small 
\small \textit{0.0304}%
} & {\small 0.0628} & {\small 0.0654} & {\small 0.0613} & 
{\small 0.0592} & {\small 0.0568} & {\small 0.0606} & {\small 0.0692} & 
{\small 0.0706} & {\small 
\small \textbf{0.0253}%
} & {\small 0.0530} \\ 
\multicolumn{1}{c}{} & \multicolumn{1}{c}{\small 0.45} & \multicolumn{1}{c}%
{\small 96} & {\small 
\small \textit{0.0885}%
} & {\small 0.1262} & {\small 0.1241} & {\small 0.1109} & 
{\small 0.1185} & {\small 0.1136} & {\small 0.1175} & {\small 0.1181} & 
{\small 0.1197} & {\small 
\small \textbf{0.0824}%
} & {\small 0.1128} \\ 
\multicolumn{1}{c}{} & \multicolumn{1}{c}{} & \multicolumn{1}{c}{\small 576}
& {\small 
\small \textit{0.0378}%
} & {\small 0.0630} & {\small 0.0581} & {\small 0.0542} & 
{\small 0.0473} & {\small 0.0495} & {\small 0.0536} & {\small 0.0647} & 
{\small 0.0695} & {\small 
\small \textbf{0.0271}%
} & {\small 0.0563} \\ \hline\hline
\end{tabular}

\end{table}%

\pagebreak

\begin{table}[h]\
\vspace{-3cm}%
\caption{{\small Bias estimates of the unadjusted LPR estimator, the feasible GS
estimator and the pre-filtered sieve bootstrap estimator, for the DGP:
ARFIMA(}${\small 1,d}_{{\small 0}}{\small ,0}${\small )  with Gaussian
innovations. The bias estimates of the feasible jackknife estimator based on
2,3,4,6,8 non-overlapping (NO) sub-samples, the feasible jackknife estimator
based on 2 moving block (MB) sub-samples, the maximum likelihood estimator (MLE)
and the pre-whitened (PW) estimator are obtained under the misspecified model:
ARFIMA(}${\small 0,d,0}${\small ) using the approach described in Section
5.3. The estimates are obtained by setting }${\small \alpha }${\small \ =
0.65. The lowest values are \textbf{bold-faced} and the second lowest values
are \textit{italicized}.}}%
\label{Table:Bias_ARFIMA(1,d,0)_mis}%
\renewcommand{\arraystretch}{0.9}%
\setlength{\tabcolsep}{3pt}%

\medskip

\begin{tabular}{lllccccccccccc}
\hline\hline
&  &  &  &  &  &  &  &  &  &  &  &  &  \\ 
${\small \phi }_{{\small 0}}$ & ${\small d}_{{\small 0}}$ & ${\small n}$ & $%
\widehat{{\small d}}_{{\small n}}$ & $\widehat{{\small d}}_{{\small J,2}}^{%
{\small NO}}$ & $\widehat{{\small d}}_{{\small J,3}}^{{\small NO}}$ & $%
\widehat{{\small d}}_{{\small J,4}}^{{\small NO}}$ & $\widehat{{\small d}}_{%
{\small J,6}}^{{\small NO}}$ & $\widehat{{\small d}}_{{\small J,8}}^{{\small %
NO}}$ & $\widehat{{\small d}}_{{\small J,2}}^{{\small MB}}$ & $\widehat{%
{\small d}}_{{\small 1}}^{{\small GS}}$ & $\widehat{{\small d}}^{{\small PFSB%
}}$ & $\widehat{{\small d}}^{{\small MLE}}$ & $\widehat{{\small d}}^{{\small %
PW}}$ \\ \hline
\multicolumn{1}{c}{} &  &  &  &  &  &  &  &  &  &  &  &  &  \\ 
\multicolumn{1}{c}{\small -0.9} & \multicolumn{1}{c}{\small -0.25} & 
\multicolumn{1}{c}{\small 96} & \multicolumn{1}{r}{\small 0.8145} & 
\multicolumn{1}{r}{\small 0.8566} & \multicolumn{1}{r}{\small 0.8595} & 
\multicolumn{1}{r}{\small 0.8643} & \multicolumn{1}{r}{\small 0.8681} & 
\multicolumn{1}{r}{\small 0.8734} & \multicolumn{1}{r}{\small 0.8540} & 
\multicolumn{1}{r}{{\small 
\small \textit{0.8002}%
}} & \multicolumn{1}{r}{{\small 
\small \textbf{0.7908}%
}} & \multicolumn{1}{r}{\small 0.8195} & \multicolumn{1}{r}%
{\small 0.8318} \\ 
\multicolumn{1}{c}{} & \multicolumn{1}{c}{} & \multicolumn{1}{c}{\small 576}
& \multicolumn{1}{r}{\small 0.5945} & \multicolumn{1}{r}{\small 0.6389} & 
\multicolumn{1}{r}{\small 0.6422} & \multicolumn{1}{r}{\small 0.6490} & 
\multicolumn{1}{r}{\small 0.6553} & \multicolumn{1}{r}{\small 0.6608} & 
\multicolumn{1}{r}{\small 0.6321} & \multicolumn{1}{r}{{\small 
\small \textbf{0.5724}%
}} & \multicolumn{1}{r}{{\small 
\small \textit{0.5898}%
}} & \multicolumn{1}{r}{\small 0.6104} & \multicolumn{1}{r}%
{\small 0.6444} \\ 
\multicolumn{1}{c}{} & \multicolumn{1}{c}{\small 0.25} & \multicolumn{1}{c}%
{\small 96} & \multicolumn{1}{r}{\small 0.7752} & \multicolumn{1}{r}{\small %
0.8075} & \multicolumn{1}{r}{\small 0.8134} & \multicolumn{1}{r}{\small %
0.8217} & \multicolumn{1}{r}{\small 0.8246} & \multicolumn{1}{r}{\small %
0.8392} & \multicolumn{1}{r}{\small 0.8273} & \multicolumn{1}{r}{{\small 
\small \textbf{0.7673}%
}} & \multicolumn{1}{r}{{\small 
\small \textit{0.7685}%
}} & \multicolumn{1}{r}{\small 0.7945} & \multicolumn{1}{r}%
{\small 0.8258} \\ 
\multicolumn{1}{c}{} & \multicolumn{1}{c}{} & \multicolumn{1}{c}{\small 576}
& \multicolumn{1}{r}{\small 0.5883} & \multicolumn{1}{r}{\small 0.6464} & 
\multicolumn{1}{r}{\small 0.6429} & \multicolumn{1}{r}{\small 0.6326} & 
\multicolumn{1}{r}{\small 0.6255} & \multicolumn{1}{r}{\small 0.6233} & 
\multicolumn{1}{r}{\small 0.6356} & \multicolumn{1}{r}{{\small 
\small \textit{0.5716}%
}} & \multicolumn{1}{r}{{\small 
\small \textbf{0.5638}%
}} & \multicolumn{1}{r}{\small 0.5943} & \multicolumn{1}{r}%
{\small 0.6305} \\ 
\multicolumn{1}{c}{} & \multicolumn{1}{c}{\small 0.45} & \multicolumn{1}{c}%
{\small 96} & \multicolumn{1}{r}{\small 0.7006} & \multicolumn{1}{r}{\small %
0.7426} & \multicolumn{1}{r}{\small 0.7482} & \multicolumn{1}{r}{\small %
0.7538} & \multicolumn{1}{r}{\small 0.7599} & \multicolumn{1}{r}{\small %
0.7646} & \multicolumn{1}{r}{\small 0.7520} & \multicolumn{1}{r}{{\small 
\small \textit{0.6946}%
}} & \multicolumn{1}{r}{{\small 
\small \textbf{0.6705}%
}} & \multicolumn{1}{r}{\small 0.7895} & \multicolumn{1}{r}%
{\small 0.8377} \\ 
\multicolumn{1}{c}{} & \multicolumn{1}{c}{} & \multicolumn{1}{c}{\small 576}
& \multicolumn{1}{r}{\small 0.5748} & \multicolumn{1}{r}{\small 0.6105} & 
\multicolumn{1}{r}{\small 0.6154} & \multicolumn{1}{r}{\small 0.6203} & 
\multicolumn{1}{r}{\small 0.6286} & \multicolumn{1}{r}{\small 0.6397} & 
\multicolumn{1}{r}{\small 0.6118} & \multicolumn{1}{r}{{\small 
\small \textit{0.5659}%
}} & \multicolumn{1}{r}{{\small 
\small \textbf{0.5451}%
}} & \multicolumn{1}{r}{\small 0.5986} & \multicolumn{1}{r}%
{\small 0.6390} \\ 
\multicolumn{1}{c}{\small -0.4} & \multicolumn{1}{c}{\small -0.25} & 
\multicolumn{1}{c}{\small 96} & \multicolumn{1}{r}{\small 0.1756} & 
\multicolumn{1}{r}{\small 0.1543} & \multicolumn{1}{r}{\small 0.1626} & 
\multicolumn{1}{r}{\small 0.1691} & \multicolumn{1}{r}{\small 0.1753} & 
\multicolumn{1}{r}{\small 0.1804} & \multicolumn{1}{r}{\small 0.1629} & 
\multicolumn{1}{r}{{\small 
\small \textbf{0.1367}%
}} & \multicolumn{1}{r}{{\small 
\small \textit{0.1435}%
}} & \multicolumn{1}{r}{\small 0.1947} & \multicolumn{1}{r}%
{\small 0.2021} \\ 
\multicolumn{1}{c}{} & \multicolumn{1}{c}{} & \multicolumn{1}{c}{\small 576}
& \multicolumn{1}{r}{\small 0.0607} & \multicolumn{1}{r}{\small 0.0419} & 
\multicolumn{1}{r}{\small 0.0484} & \multicolumn{1}{r}{\small 0.0507} & 
\multicolumn{1}{r}{\small 0.0541} & \multicolumn{1}{r}{\small 0.0585} & 
\multicolumn{1}{r}{\small 0.0426} & \multicolumn{1}{r}{{\small 
\small \textit{0.0304}%
}} & \multicolumn{1}{r}{{\small 
\small \textbf{0.0286}%
}} & \multicolumn{1}{r}{\small 0.0743} & \multicolumn{1}{r}%
{\small 0.1035} \\ 
\multicolumn{1}{c}{} & \multicolumn{1}{c}{\small 0.25} & \multicolumn{1}{c}%
{\small 96} & \multicolumn{1}{r}{\small 0.1629} & \multicolumn{1}{r}{\small %
0.1536} & \multicolumn{1}{r}{\small 0.1588} & \multicolumn{1}{r}{\small %
0.1614} & \multicolumn{1}{r}{\small 0.1669} & \multicolumn{1}{r}{\small %
0.1760} & \multicolumn{1}{r}{\small 0.1512} & \multicolumn{1}{r}{{\small 
\small \textit{0.1329}%
}} & \multicolumn{1}{r}{{\small 
\small \textbf{0.1294}%
}} & \multicolumn{1}{r}{\small 0.1840} & \multicolumn{1}{r}%
{\small 0.2064} \\ 
\multicolumn{1}{c}{} & \multicolumn{1}{c}{} & \multicolumn{1}{c}{\small 576}
& \multicolumn{1}{r}{\small 0.0571} & \multicolumn{1}{r}{\small 0.0348} & 
\multicolumn{1}{r}{\small 0.0376} & \multicolumn{1}{r}{\small 0.0491} & 
\multicolumn{1}{r}{\small 0.0527} & \multicolumn{1}{r}{\small 0.0598} & 
\multicolumn{1}{r}{\small 0.0506} & \multicolumn{1}{r}{{\small 
\small \textit{0.0289}%
}} & \multicolumn{1}{r}{{\small 
\small \textbf{0.0251}%
}} & \multicolumn{1}{r}{\small 0.0587} & \multicolumn{1}{r}%
{\small 0.0996} \\ 
\multicolumn{1}{c}{} & \multicolumn{1}{c}{\small 0.45} & \multicolumn{1}{c}%
{\small 96} & \multicolumn{1}{r}{\small 0.1653} & \multicolumn{1}{r}{\small %
0.1546} & \multicolumn{1}{r}{\small 0.1592} & \multicolumn{1}{r}{\small %
0.1648} & \multicolumn{1}{r}{\small 0.1688} & \multicolumn{1}{r}{\small %
0.1735} & \multicolumn{1}{r}{\small 0.1648} & \multicolumn{1}{r}{{\small 
\small \textit{0.1400}%
}} & \multicolumn{1}{r}{{\small 
\small \textbf{0.1277}%
}} & \multicolumn{1}{r}{\small 0.1886} & \multicolumn{1}{r}%
{\small 0.2145} \\ 
\multicolumn{1}{c}{} & \multicolumn{1}{c}{} & \multicolumn{1}{c}{\small 576}
& \multicolumn{1}{r}{\small 0.0625} & \multicolumn{1}{r}{\small 0.0445} & 
\multicolumn{1}{r}{\small 0.0498} & \multicolumn{1}{r}{\small 0.0572} & 
\multicolumn{1}{r}{\small 0.0645} & \multicolumn{1}{r}{\small 0.0686} & 
\multicolumn{1}{r}{\small 0.0749} & \multicolumn{1}{r}{{\small 
\small \textit{0.0359}%
}} & \multicolumn{1}{r}{{\small 
\small \textbf{0.0261}%
}} & \multicolumn{1}{r}{\small 0.0534} & \multicolumn{1}{r}%
{\small 0.1031} \\ 
\multicolumn{1}{c}{\small 0.4} & \multicolumn{1}{c}{\small -0.25} & 
\multicolumn{1}{c}{\small 96} & \multicolumn{1}{r}{\small -0.0363} & 
\multicolumn{1}{r}{\small -0.0286} & \multicolumn{1}{r}{\small -0.0359} & 
\multicolumn{1}{r}{\small -0.0397} & \multicolumn{1}{r}{\small -0.0462} & 
\multicolumn{1}{r}{\small -0.0481} & \multicolumn{1}{r}{\small -0.0388} & 
\multicolumn{1}{r}{{\small 
\small \textbf{-0.0047}%
}} & \multicolumn{1}{r}{{\small 
\small \textit{-0.0147}%
}} & \multicolumn{1}{r}{\small -0.0385} & \multicolumn{1}{r}%
{\small -0.0490} \\ 
\multicolumn{1}{c}{} & \multicolumn{1}{c}{} & \multicolumn{1}{c}{\small 576}
& \multicolumn{1}{r}{\small -0.0056*} & \multicolumn{1}{r}{\small -0.0105} & 
\multicolumn{1}{r}{\small -0.0128} & \multicolumn{1}{r}{\small -0.0156} & 
\multicolumn{1}{r}{\small -0.0184} & \multicolumn{1}{r}{\small -0.0229} & 
\multicolumn{1}{r}{\small -0.0154} & \multicolumn{1}{r}{{\small 
\small \textit{-0.0056}
}} & \multicolumn{1}{r}{{\small 
\small \textbf{-0.0004}%
}} & \multicolumn{1}{r}{\small -0.0163} & \multicolumn{1}{r}%
{\small -0.0258} \\ 
\multicolumn{1}{c}{} & \multicolumn{1}{c}{\small 0.25} & \multicolumn{1}{c}%
{\small 96} & \multicolumn{1}{r}{\small -0.0559} & \multicolumn{1}{r}{\small %
-0.0267} & \multicolumn{1}{r}{\small -0.0293} & \multicolumn{1}{r}{\small %
-0.0319} & \multicolumn{1}{r}{\small -0.0350} & \multicolumn{1}{r}{\small %
-0.0424} & \multicolumn{1}{r}{\small -0.0372} & \multicolumn{1}{r}{{\small 
\small \textbf{-0.0068}%
}} & \multicolumn{1}{r}{{\small 
\small \textit{-0.0153}%
}} & \multicolumn{1}{r}{\small -0.0372} & \multicolumn{1}{r}%
{\small -0.0593} \\ 
\multicolumn{1}{c}{} & \multicolumn{1}{c}{} & \multicolumn{1}{c}{\small 576}
& \multicolumn{1}{r}{\small -0.0115} & \multicolumn{1}{r}{\small -0.0104} & 
\multicolumn{1}{r}{\small -0.0131} & \multicolumn{1}{r}{\small -0.0177} & 
\multicolumn{1}{r}{\small -0.0176} & \multicolumn{1}{r}{\small -0.0239} & 
\multicolumn{1}{r}{\small -0.0168} & \multicolumn{1}{r}{{\small 
\small \textbf{0.0017}%
}} & \multicolumn{1}{r}{{\small 
\small \textit{-0.0027}%
}} & \multicolumn{1}{r}{\small -0.0095} & \multicolumn{1}{r}%
{\small -0.0269} \\ 
\multicolumn{1}{c}{} & \multicolumn{1}{c}{\small 0.45} & \multicolumn{1}{c}%
{\small 96} & \multicolumn{1}{r}{\small -0.0501} & \multicolumn{1}{r}{\small %
-0.0279} & \multicolumn{1}{r}{\small -0.0249} & \multicolumn{1}{r}{\small %
-0.0325} & \multicolumn{1}{r}{\small -0.0381} & \multicolumn{1}{r}{\small %
-0.0458} & \multicolumn{1}{r}{\small -0.0294} & \multicolumn{1}{r}{{\small 
\small \textbf{0.0032}%
}} & \multicolumn{1}{r}{{\small 
\small \textit{-0.0111}%
}} & \multicolumn{1}{r}{\small -0.0314} & \multicolumn{1}{r}%
{\small -0.0627} \\ 
\multicolumn{1}{c}{} & \multicolumn{1}{c}{} & \multicolumn{1}{c}{\small 576}
& \multicolumn{1}{r}{{\small 
\small \textit{-0.0058}%
}} & \multicolumn{1}{r}{\small -0.0115} & \multicolumn{1}{r}%
{\small -0.0157} & \multicolumn{1}{r}{\small -0.0186} & \multicolumn{1}{r}%
{\small -0.0195} & \multicolumn{1}{r}{\small -0.0210} & \multicolumn{1}{r}%
{\small -0.0153} & \multicolumn{1}{r}{\small 0.0089} & \multicolumn{1}{r}{%
{\small 
\small \textbf{0.0004}%
}} & \multicolumn{1}{r}{\small -0.0088} & \multicolumn{1}{r}%
{\small -0.0251} \\ 
\multicolumn{1}{c}{\small 0.9} & \multicolumn{1}{c}{\small -0.25} & 
\multicolumn{1}{c}{\small 96} & \multicolumn{1}{r}{\small -0.0291} & 
\multicolumn{1}{r}{{\small 
\small \textbf{-0.0123}%
}} & \multicolumn{1}{r}{{\small 
\small \textit{-0.0129}%
}} & \multicolumn{1}{r}{\small -0.0145} & \multicolumn{1}{r}%
{\small -0.0193} & \multicolumn{1}{r}{\small -0.0224} & \multicolumn{1}{r}%
{\small -0.0148} & \multicolumn{1}{r}{\small -0.0175} & \multicolumn{1}{r}%
{\small -0.0162} & \multicolumn{1}{r}{\small -0.0156} & \multicolumn{1}{r}%
{\small -0.0339} \\ 
\multicolumn{1}{c}{} & \multicolumn{1}{c}{} & \multicolumn{1}{c}{\small 576}
& \multicolumn{1}{r}{\small -0.0058} & \multicolumn{1}{r}{{\small 
\small \textbf{-0.0020}%
}} & \multicolumn{1}{r}{\small -0.0028} & \multicolumn{1}{r}%
{\small -0.0041} & \multicolumn{1}{r}{\small -0.0066} & \multicolumn{1}{r}%
{\small -0.0075} & \multicolumn{1}{r}{\small -0.0049} & \multicolumn{1}{r}%
{\small -0.0034} & \multicolumn{1}{r}{{\small 
\small \textit{-0.0023}%
}} & \multicolumn{1}{r}{\small -0.0037} & \multicolumn{1}{r}%
{\small -0.0142} \\ 
\multicolumn{1}{c}{} & \multicolumn{1}{c}{\small 0.25} & \multicolumn{1}{c}%
{\small 96} & \multicolumn{1}{r}{\small -0.0249} & \multicolumn{1}{r}{%
{\small 
\small \textbf{-0.0106}%
}} & \multicolumn{1}{r}{\small -0.0132} & \multicolumn{1}{r}%
{\small -0.0220} & \multicolumn{1}{r}{\small -0.0241} & \multicolumn{1}{r}%
{\small -0.0269} & \multicolumn{1}{r}{\small -0.0153} & \multicolumn{1}{r}%
{\small -0.0162} & \multicolumn{1}{r}{{\small 
\small \textit{-0.0117}%
}} & \multicolumn{1}{r}{\small -0.0163} & \multicolumn{1}{r}%
{\small -0.0326} \\ 
\multicolumn{1}{c}{} & \multicolumn{1}{c}{} & \multicolumn{1}{c}{\small 576}
& \multicolumn{1}{r}{\small -0.0044} & \multicolumn{1}{r}{{\small 
\small \textit{-0.0021}%
}} & \multicolumn{1}{r}{\small -0.0047} & \multicolumn{1}{r}%
{\small -0.0055} & \multicolumn{1}{r}{\small -0.0073} & \multicolumn{1}{r}%
{\small -0.0084} & \multicolumn{1}{r}{\small -0.0062} & \multicolumn{1}{r}%
{\small -0.0032} & \multicolumn{1}{r}{{\small 
\small \textbf{-0.0020}%
}} & \multicolumn{1}{r}{\small -0.0034} & \multicolumn{1}{r}%
{\small -0.0174} \\ 
\multicolumn{1}{c}{} & \multicolumn{1}{c}{\small 0.45} & \multicolumn{1}{c}%
{\small 96} & \multicolumn{1}{r}{\small -0.0241} & \multicolumn{1}{r}{%
{\small 
\small \textbf{-0.0095}%
}} & \multicolumn{1}{r}{\small -0.0126*} & \multicolumn{1}{r}%
{\small -0.0217} & \multicolumn{1}{r}{\small -0.0229} & \multicolumn{1}{r}%
{\small -0.0241} & \multicolumn{1}{r}{\small -0.0175} & \multicolumn{1}{r}%
{\small -0.0175} & \multicolumn{1}{r}{{\small 
\small \textit{-0.0126}%
}} & \multicolumn{1}{r}{\small -0.0145} & \multicolumn{1}{r}%
{\small -0.0366} \\ 
\multicolumn{1}{c}{} & \multicolumn{1}{c}{} & \multicolumn{1}{c}{\small 576}
& \multicolumn{1}{r}{\small -0.0077} & \multicolumn{1}{r}{{\small 
\small \textit{-0.0026}%
}} & \multicolumn{1}{r}{\small -0.0035} & \multicolumn{1}{r}%
{\small -0.0042} & \multicolumn{1}{r}{\small -0.0048} & \multicolumn{1}{r}%
{\small -0.0065} & \multicolumn{1}{r}{\small -0.0029} & \multicolumn{1}{r}%
{\small -0.0038} & \multicolumn{1}{r}{{\small 
\small \textbf{-0.0018}%
}} & \multicolumn{1}{r}{\small -0.0048} & \multicolumn{1}{r}%
{\small -0.0134} \\ \hline\hline
\end{tabular}

\end{table}%

\newpage

\begin{table}[h]\
\vspace{-3.0cm}%
\caption{{\small RMSE estimates of the unadjusted LPR estimator, the feasible GS
estimator and the pre-filtered sieve bootstrap estimator, for the DGP:
ARFIMA(}${\small 1,d}_{{\small 0}}{\small ,0}${\small )  with Gaussian
innovations. The bias estimates of the feasible jackknife estimator based on
2,3,4,6,8 non-overlapping (NO) sub-samples, the feasible jackknife estimator
based on 2 moving block (MB) sub-samples, the maximum likelihood estimator (MLE)
and the pre-whitened (PW) estimator are obtained under the misspecified model:
ARFIMA(}${\small 0,d,0}${\small ) using the approach described in Section
5.3. The estimates are obtained by setting }${\small \alpha }${\small \ =
0.65. The lowest values are \textbf{bold-faced} and the second lowest values
are \textit{italicized}.}}%
\label{Table:MSE_ARFIMA(1,d,0)_mis}%
\renewcommand{\arraystretch}{0.9}%
\setlength{\tabcolsep}{4pt}%

\medskip

\begin{tabular}{lllccccccccccc}
\hline\hline
&  &  &  &  &  &  &  &  &  &  &  &  &  \\ 
${\small \phi }_{{\small 0}}$ & ${\small d}_{{\small 0}}$ & ${\small n}$ & $%
\widehat{{\small d}}_{{\small n}}$ & $\widehat{{\small d}}_{{\small J,2}}^{%
{\small NO}}$ & $\widehat{{\small d}}_{{\small J,3}}^{{\small NO}}$ & $%
\widehat{{\small d}}_{{\small J,4}}^{{\small NO}}$ & $\widehat{{\small d}}_{%
{\small J,6}}^{{\small NO}}$ & $\widehat{{\small d}}_{{\small J,8}}^{{\small %
NO}}$ & $\widehat{{\small d}}_{{\small J,2}}^{{\small MB}}$ & $\widehat{%
{\small d}}_{{\small 1}}^{{\small GS}}$ & $\widehat{{\small d}}^{{\small PFSB%
}}$ & $\widehat{{\small d}}^{{\small MLE}}$ & $\widehat{{\small d}}^{{\small %
PW}}$ \\ \hline
\multicolumn{14}{c}{} \\ 
\multicolumn{1}{c}{\small -0.9} & \multicolumn{1}{c}{\small -0.25} & 
\multicolumn{1}{c}{\small 96} & {\small 
\small \textbf{1.0359}%
} & {\small 1.2624} & {\small 1.2506} & {\small 1.2406} & 
{\small 1.2416} & {\small 1.2384} & {\small 1.2395} & {\small 1.3386} & 
{\small 1.2885} & {\small 
\small \textit{1.2247}%
} & {\small 1.2372} \\ 
\multicolumn{1}{c}{} & \multicolumn{1}{c}{} & \multicolumn{1}{c}{\small 576}
& {\small 0.7398} & {\small 0.7851} & {\small 0.7827} & {\small 0.7795} & 
{\small 0.7628} & {\small 0.7695} & {\small 0.7726} & {\small 
\small \textit{0.7371}%
} & {\small 
\small \textbf{0.7359}%
} & {\small 0.8030} & {\small 0.8428} \\ 
\multicolumn{1}{c}{} & \multicolumn{1}{c}{\small 0.25} & \multicolumn{1}{c}%
{\small 96} & {\small 
\small \textit{1.1618}%
} & {\small 1.2137} & {\small 1.2154} & {\small 1.2042} & 
{\small 1.2069} & {\small 1.1958} & {\small 1.2073} & {\small 
\small \textbf{1.1484}%
} & {\small 1.2299} & {\small 1.2149} & {\small 1.2226} \\ 
\multicolumn{1}{c}{} & \multicolumn{1}{c}{} & \multicolumn{1}{c}{\small 576}
& {\small 
\small \textbf{0.9175}%
} & {\small 0.9769} & {\small 0.9732} & {\small 0.9618} & 
{\small 0.9594} & {\small 
\small \textit{0.9537}%
} & {\small 0.9678} & {\small 1.1171} & {\small 1.1130} & 
{\small 0.9556} & {\small 0.9638} \\ 
\multicolumn{1}{c}{} & \multicolumn{1}{c}{\small 0.45} & \multicolumn{1}{c}%
{\small 96} & {\small 
\small \textbf{1.1286}%
} & {\small 1.2446} & {\small 1.2384} & {\small 1.2329} & 
{\small 1.2280} & {\small 
\small \textit{1.2186}%
} & {\small 1.2349} & {\small 1.4331} & {\small 1.5385} & 
{\small 1.2285} & {\small 1.2385} \\ 
\multicolumn{1}{c}{} & \multicolumn{1}{c}{} & \multicolumn{1}{c}{\small 576}
& {\small 
\small \textbf{0.9708}%
} & {\small 1.0975} & {\small 1.0299} & {\small 1.0183} & 
{\small 0.9716} & {\small 
\small \textit{0.9824}%
} & {\small 1.0198} & {\small 1.1124} & {\small 1.1647} & 
{\small 0.9929} & {\small 1.0064} \\ 
\multicolumn{1}{c}{\small -0.4} & \multicolumn{1}{c}{\small -0.25} & 
\multicolumn{1}{c}{\small 96} & {\small 
\small \textbf{0.2568}%
} & {\small 0.3063} & {\small 0.3015} & {\small 0.2940} & 
{\small 0.2874} & {\small 0.2711} & {\small 0.2726} & {\small 
\small \textit{0.2594}%
} & {\small 0.3028} & {\small 0.2736} & {\small 0.2915} \\ 
\multicolumn{1}{c}{} & \multicolumn{1}{c}{} & \multicolumn{1}{c}{\small 576}
& {\small 
\small \textbf{0.1098}%
} & {\small 0.1408} & {\small 0.1377} & {\small 0.1305} & 
{\small 0.1249} & {\small 0.1153} & {\small 0.1281} & {\small 
\small \textit{0.1118}%
} & {\small 0.1272} & {\small 0.1242} & {\small 0.1463} \\ 
\multicolumn{1}{c}{} & \multicolumn{1}{c}{\small 0.25} & \multicolumn{1}{c}%
{\small 96} & {\small 
\small \textbf{0.2490}%
} & {\small 0.3034} & {\small 0.2956} & {\small 0.2874} & 
{\small 0.2713} & {\small 0.2624} & {\small 0.3097} & {\small 
\small \textit{0.2580}%
} & {\small 0.2879} & {\small 0.2794} & {\small 0.2905} \\ 
\multicolumn{1}{c}{} & \multicolumn{1}{c}{} & \multicolumn{1}{c}{\small 576}
& {\small 
\small \textbf{0.1079}%
} & {\small 0.1410} & {\small 0.1393} & {\small 0.1318} & 
{\small 0.1275} & {\small 0.1248} & {\small 0.1290} & {\small 0.1375} & 
{\small 
\small \textit{0.1239}%
} & {\small 0.1329} & {\small 0.1539} \\ 
\multicolumn{1}{c}{} & \multicolumn{1}{c}{\small 0.45} & \multicolumn{1}{c}%
{\small 96} & {\small 
\small \textbf{0.2506}%
} & {\small 0.3087} & {\small 0.3016} & {\small 0.2971} & 
{\small 0.2840} & {\small 0.2737} & {\small 0.2972} & {\small 0.2616} & 
{\small 
\small \textit{0.2506}%
} & {\small 0.2840} & {\small 0.3017} \\ 
\multicolumn{1}{c}{} & \multicolumn{1}{c}{} & \multicolumn{1}{c}{\small 576}
& {\small 
\small \textbf{0.1115}%
} & {\small 0.1672} & {\small 0.1508} & {\small 0.1438} & 
{\small 0.1378} & {\small 0.1225} & {\small 0.1482} & {\small 
\small \textit{0.1143}%
} & {\small 0.1230} & {\small 0.1385} & {\small 0.1228} \\ 
\multicolumn{1}{c}{\small 0.4} & \multicolumn{1}{c}{\small -0.25} & 
\multicolumn{1}{c}{\small 96} & {\small 
\small \textbf{0.1917}%
} & {\small 0.2441} & {\small 0.2318} & {\small 0.2473} & 
{\small 0.2319} & {\small 0.2301} & {\small 0.2333} & {\small 
\small \textit{0.2212}%
} & {\small 0.2717} & {\small 0.2425} & {\small 0.2573} \\ 
\multicolumn{1}{c}{} & \multicolumn{1}{c}{} & \multicolumn{1}{c}{\small 576}
& {\small 
\small \textbf{0.0919}%
} & {\small 0.1323} & {\small 0.1295} & {\small 0.1206} & 
{\small 0.1134} & {\small 0.1199} & {\small 0.1210} & {\small 0.1253} & 
{\small 
\small \textit{0.1198}%
} & {\small 0.1347} & {\small 0.1836} \\ 
\multicolumn{1}{c}{} & \multicolumn{1}{c}{\small 0.25} & \multicolumn{1}{c}%
{\small 96} & {\small 
\small \textbf{0.1960}%
} & {\small 0.2359} & {\small 0.2306} & {\small 0.2239} & 
{\small 0.2406} & {\small 0.2378} & {\small 0.2381} & {\small 
\small \textit{0.2209}%
} & {\small 0.2482} & {\small 0.2385} & {\small 0.2529} \\ 
\multicolumn{1}{c}{} & \multicolumn{1}{c}{} & \multicolumn{1}{c}{\small 576}
& {\small 
\small \textbf{0.0922}%
} & {\small 0.1330} & {\small 0.1289} & {\small 0.1211} & 
{\small 0.1281} & {\small 0.1166} & {\small 0.1199} & {\small 0.1366} & 
{\small 0.1158} & {\small 
\small \textit{0.1093}%
} & {\small 0.1545} \\ 
\multicolumn{1}{c}{} & \multicolumn{1}{c}{\small 0.45} & \multicolumn{1}{c}%
{\small 96} & {\small 
\small \textbf{0.1955}%
} & {\small 0.2416} & {\small 0.2345} & {\small 0.2394} & 
{\small 0.2296} & {\small 0.2244} & {\small 0.2267} & {\small 
\small \textit{0.2218}%
} & {\small 0.2453} & {\small 0.2236} & {\small 0.2530} \\ 
\multicolumn{1}{c}{} & \multicolumn{1}{c}{} & \multicolumn{1}{c}{\small 576}
& {\small 
\small \textbf{0.0926}%
} & {\small 0.1305} & {\small 0.1287} & {\small 0.1226} & 
{\small 0.1148} & {\small 0.1106} & {\small 0.1292} & {\small 0.1089} & 
{\small 0.1149} & {\small 
\small \textit{0.1058}%
} & {\small 0.1551} \\ 
\multicolumn{1}{c}{\small 0.9} & \multicolumn{1}{c}{\small -0.25} & 
\multicolumn{1}{c}{\small 96} & {\small 
\small \textbf{0.1115}%
} & {\small 0.1419} & {\small 0.1381} & {\small 0.1305} & 
{\small 0.1267} & {\small 
\small \textit{0.1212}%
} & {\small 0.1284} & {\small 0.1365} & {\small 0.1266} & 
{\small 0.1315} & {\small 0.1589} \\ 
\multicolumn{1}{c}{} & \multicolumn{1}{c}{} & \multicolumn{1}{c}{\small 576}
& {\small 
\small \textit{0.0624}%
} & {\small 0.0795} & {\small 0.0754} & {\small 0.0713} & 
{\small 0.0691} & {\small 0.0627} & {\small 0.0553} & {\small 0.0708} & 
{\small 
\small \textbf{0.0600}%
} & {\small 0.0703} & {\small 0.0846} \\ 
\multicolumn{1}{c}{} & \multicolumn{1}{c}{\small 0.25} & \multicolumn{1}{c}%
{\small 96} & {\small 
\small \textbf{0.1114}%
} & {\small 0.1482} & {\small 0.1469} & {\small 0.1337} & 
{\small 0.1304} & {\small 0.1289} & {\small 0.1376} & {\small 0.1328} & 
{\small 
\small \textit{0.1282}%
} & {\small 0.1293} & {\small 0.1428} \\ 
\multicolumn{1}{c}{} & \multicolumn{1}{c}{} & \multicolumn{1}{c}{\small 576}
& {\small 
\small \textbf{0.0518}%
} & {\small 0.0800} & {\small 0.0786} & {\small 0.0711} & 
{\small 0.0678} & {\small 0.0624} & {\small 0.0545} & {\small 0.0626} & 
{\small 0.0581} & {\small 
\small \textit{0.0540}%
} & {\small 0.0722} \\ 
\multicolumn{1}{c}{} & \multicolumn{1}{c}{\small 0.45} & \multicolumn{1}{c}%
{\small 96} & {\small 
\small \textbf{0.1053}%
} & {\small 0.1495} & {\small 0.1461} & {\small 0.1398} & 
{\small 0.1376} & {\small 0.1268} & {\small 0.1272} & {\small 0.1253} & 
{\small 0.1215} & {\small 
\small \textit{0.1214}%
} & {\small 0.1473} \\ 
\multicolumn{1}{c}{} & \multicolumn{1}{c}{} & \multicolumn{1}{c}{\small 576}
& {\small 
\small \textbf{0.0526}%
} & {\small 0.0714} & {\small 0.0790} & {\small 0.0774} & 
{\small 0.0722} & {\small 
\small \textit{0.0659}%
} & {\small 0.0698} & {\small 0.0769} & {\small 0.0684} & 
{\small 0.0648} & {\small 0.0739} \\ \hline\hline
\end{tabular}

\end{table}%

\newpage

\begin{table}[h]\
\vspace{-3cm}%
\caption{{\small Bias estimates of the unadjusted LPR estimator, the feasiblel GS
estimator and the pre-filtered sieve bootstrap estimator, for the DGP:
ARFIMA(}${\small 0,d}_{{\small 0}}{\small ,1}${\small )  with Gaussian
innovations. The bias estimates of the feasible jackknife estimator based on
2,3,4,6,8 non-overlapping (NO) sub-samples, the feasible jackknife estimator
based on 2 moving block (MB) sub-samples, the maximum likelihood estimator (MLE)
and the pre-whitened (PW) estimator are obtained under the misspecified model:
ARFIMA(}${\small 0,d,0}${\small ) using the approach described in Section
5.3. The estimates are obtained by setting }${\small \alpha }${\small \ =
0.65. The lowest values are \textbf{bold-faced} and the second lowest values
are \textit{italicized}.}}%
\label{Table:Bias_ARFIMA(0,d,1)_mis}%
\renewcommand{\arraystretch}{0.9}%
\setlength{\tabcolsep}{3pt}%

\medskip

\begin{tabular}{lllccccccccccc}
\hline\hline
&  &  &  &  &  &  &  &  &  &  &  &  &  \\ 
${\small \theta }_{{\small 0}}$ & ${\small d}_{{\small 0}}$ & ${\small n}$ & 
$\widehat{{\small d}}_{{\small n}}$ & $\widehat{{\small d}}_{{\small J,2}}^{%
{\small NO}}$ & $\widehat{{\small d}}_{{\small J,3}}^{{\small NO}}$ & $%
\widehat{{\small d}}_{{\small J,4}}^{{\small NO}}$ & $\widehat{{\small d}}_{%
{\small J,6}}^{{\small NO}}$ & $\widehat{{\small d}}_{{\small J,8}}^{{\small %
NO}}$ & $\widehat{{\small d}}_{{\small J,2}}^{{\small MB}}$ & $\widehat{%
{\small d}}_{{\small 1}}^{{\small GS}}$ & $\widehat{{\small d}}^{{\small PFSB%
}}$ & $\widehat{{\small d}}^{{\small MLE}}$ & $\widehat{{\small d}}^{{\small %
PW}}$ \\ \hline
\multicolumn{14}{c}{} \\ 
\multicolumn{1}{c}{\small -0.9} & \multicolumn{1}{c}{\small -0.25} & 
\multicolumn{1}{c}{\small 96} & \multicolumn{1}{r}{\small -0.5671} & 
\multicolumn{1}{r}{{\small 
\small \textbf{-0.5355}%
}} & \multicolumn{1}{r}{{\small 
\small%
}\textit{-0.5398}} & \multicolumn{1}{r}{\small -0.5437} & \multicolumn{1}{r}%
{\small -0.5541} & \multicolumn{1}{r}{\small -0.5626} & \multicolumn{1}{r}%
{\small -0.5584} & \multicolumn{1}{r}{\small -0.5450} & \multicolumn{1}{r}%
{\small -0.5466} & \multicolumn{1}{r}{\small -0.5692} & \multicolumn{1}{r}%
{\small -0.6361} \\ 
\multicolumn{1}{c}{} & \multicolumn{1}{c}{} & \multicolumn{1}{c}{\small 576}
& \multicolumn{1}{r}{\small -0.4527} & \multicolumn{1}{r}{{\small 
\small \textbf{-0.4216}%
}} & \multicolumn{1}{r}{{\small 
\small \textit{-0.4273}%
}} & \multicolumn{1}{r}{\small -0.4349} & \multicolumn{1}{r}%
{\small -0.4479} & \multicolumn{1}{r}{\small -0.4581} & \multicolumn{1}{r}%
{\small -0.4469} & \multicolumn{1}{r}{\small -0.4385} & \multicolumn{1}{r}%
{\small -0.4285} & \multicolumn{1}{r}{\small -0.4375} & \multicolumn{1}{r}%
{\small -0.4858} \\ 
\multicolumn{1}{c}{} & \multicolumn{1}{c}{\small 0.25} & \multicolumn{1}{c}%
{\small 96} & \multicolumn{1}{r}{\small -0.7763} & \multicolumn{1}{r}{%
{\small 
\small \textbf{-0.7394}%
}} & \multicolumn{1}{r}{{\small 
\small \textit{-0.7401}%
}} & \multicolumn{1}{r}{\small -0.7475} & \multicolumn{1}{r}%
{\small -0.7523} & \multicolumn{1}{r}{\small -0.7670} & \multicolumn{1}{r}%
{\small -0.7416} & \multicolumn{1}{r}{\small -0.7524} & \multicolumn{1}{r}%
{\small -0.7661} & \multicolumn{1}{r}{\small -0.7718} & \multicolumn{1}{r}%
{\small -0.8149} \\ 
\multicolumn{1}{c}{} & \multicolumn{1}{c}{} & \multicolumn{1}{c}{\small 576}
& \multicolumn{1}{r}{\small -0.5880} & \multicolumn{1}{r}{\small -0.5482} & 
\multicolumn{1}{r}{\small -0.5529} & \multicolumn{1}{r}{\small -0.5653} & 
\multicolumn{1}{r}{\small -0.5776} & \multicolumn{1}{r}{\small -0.5798} & 
\multicolumn{1}{r}{\small -0.5629} & \multicolumn{1}{r}{{\small 
\small \textit{-0.5473}%
}} & \multicolumn{1}{r}{\small -0.5621} & \multicolumn{1}{r}{%
{\small 
\small \textbf{-0.5426}%
}} & \multicolumn{1}{r}{\small -0.5846} \\ 
\multicolumn{1}{c}{} & \multicolumn{1}{c}{\small 0.45} & \multicolumn{1}{c}%
{\small 96} & \multicolumn{1}{r}{\small -0.8004} & \multicolumn{1}{r}{%
{\small 
\small \textbf{-0.7469}%
}} & \multicolumn{1}{r}{{\small 
\small \textit{-0.7504}%
}} & \multicolumn{1}{r}{\small -0.7638} & \multicolumn{1}{r}%
{\small -0.7769} & \multicolumn{1}{r}{\small -0.7716} & \multicolumn{1}{r}%
{\small -0.7553} & \multicolumn{1}{r}{\small -0.7600} & \multicolumn{1}{r}%
{\small -0.7854} & \multicolumn{1}{r}{\small -0.7713} & \multicolumn{1}{r}%
{\small -0.7942} \\ 
\multicolumn{1}{c}{} & \multicolumn{1}{c}{} & \multicolumn{1}{c}{\small 576}
& \multicolumn{1}{r}{\small -0.5880} & \multicolumn{1}{r}{{\small 
\small \textbf{-0.5033}%
}} & \multicolumn{1}{r}{{\small 
\small \textit{-0.5059}%
}} & \multicolumn{1}{r}{\small -0.5249} & \multicolumn{1}{r}%
{\small -0.5384} & \multicolumn{1}{r}{\small -0.5529} & \multicolumn{1}{r}%
{\small -0.5247} & \multicolumn{1}{r}{\small -0.5351} & \multicolumn{1}{r}%
{\small -0.5527} & \multicolumn{1}{r}{\small -0.5484} & \multicolumn{1}{r}%
{\small -0.5927} \\ 
\multicolumn{1}{c}{\small -0.4} & \multicolumn{1}{c}{\small -0.25} & 
\multicolumn{1}{c}{\small 96} & \multicolumn{1}{r}{\small -0.1437} & 
\multicolumn{1}{r}{\small -0.1468} & \multicolumn{1}{r}{\small -0.1288} & 
\multicolumn{1}{r}{\small -0.1201} & \multicolumn{1}{r}{{\small 
\small \textit{-0.1175}%
}} & \multicolumn{1}{r}{\small -0.1435} & \multicolumn{1}{r}%
{\small -0.1586} & \multicolumn{1}{r}{{\small 
\small \textbf{-0.1120}%
}} & \multicolumn{1}{r}{\small -0.1240} & \multicolumn{1}{r}%
{\small -0.1379} & \multicolumn{1}{r}{\small -0.1728} \\ 
\multicolumn{1}{c}{} & \multicolumn{1}{c}{} & \multicolumn{1}{c}{\small 576}
& \multicolumn{1}{r}{\small -0.0476} & \multicolumn{1}{r}{\small -0.0296} & 
\multicolumn{1}{r}{\small -0.0379} & \multicolumn{1}{r}{\small -0.0334} & 
\multicolumn{1}{r}{\small -0.0562} & \multicolumn{1}{r}{\small -0.0560} & 
\multicolumn{1}{r}{\small -0.0392} & \multicolumn{1}{r}{{\small 
\small \textbf{-0.0187}%
}} & \multicolumn{1}{r}{{\small 
\small \textit{-0.0271}%
}} & \multicolumn{1}{r}{\small -0.0584} & \multicolumn{1}{r}%
{\small -0.0957} \\ 
\multicolumn{1}{c}{} & \multicolumn{1}{c}{\small 0.25} & \multicolumn{1}{c}%
{\small 96} & \multicolumn{1}{r}{\small -0.1692} & \multicolumn{1}{r}{%
{\small 
\small \textbf{-0.1172}%
}} & \multicolumn{1}{r}{\small -0.1243} & \multicolumn{1}{r}%
{\small -0.1289} & \multicolumn{1}{r}{\small -0.1350} & \multicolumn{1}{r}%
{\small -0.1462} & \multicolumn{1}{r}{\small -0.1471} & \multicolumn{1}{r}%
{\small -0.1297} & \multicolumn{1}{r}{{\small 
\small \textit{-0.1200}%
}} & \multicolumn{1}{r}{\small -0.1381} & \multicolumn{1}{r}%
{\small -0.1783} \\ 
\multicolumn{1}{c}{} & \multicolumn{1}{c}{} & \multicolumn{1}{c}{\small 576}
& \multicolumn{1}{r}{\small -0.0552} & \multicolumn{1}{r}{{\small 
\small \textbf{-0.0222}%
}} & \multicolumn{1}{r}{\small -0.0318} & \multicolumn{1}{r}%
{\small -0.0425} & \multicolumn{1}{r}{\small -0.0531} & \multicolumn{1}{r}%
{\small -0.0573} & \multicolumn{1}{r}{\small -0.0469} & \multicolumn{1}{r}{%
{\small 
\small \textit{-0.0243}%
}} & \multicolumn{1}{r}{\small -0.0287} & \multicolumn{1}{r}%
{\small -0.0558} & \multicolumn{1}{r}{\small -0.0990} \\ 
\multicolumn{1}{c}{} & \multicolumn{1}{c}{\small 0.45} & \multicolumn{1}{c}%
{\small 96} & \multicolumn{1}{r}{\small -0.1630} & \multicolumn{1}{r}{%
{\small 
\small \textbf{-0.0716}%
}} & \multicolumn{1}{r}{{\small 
\small \textit{-0.1076}%
}} & \multicolumn{1}{r}{\small -0.1277} & \multicolumn{1}{r}%
{\small -0.1392} & \multicolumn{1}{r}{\small -0.1436} & \multicolumn{1}{r}%
{\small -0.1318} & \multicolumn{1}{r}{\small -0.1190} & \multicolumn{1}{r}%
{\small -0.1118} & \multicolumn{1}{r}{\small -0.1254} & \multicolumn{1}{r}%
{\small -0.1739} \\ 
\multicolumn{1}{c}{} & \multicolumn{1}{c}{} & \multicolumn{1}{c}{\small 576}
& \multicolumn{1}{r}{\small -0.0493} & \multicolumn{1}{r}{{\small 
\small \textbf{-0.0152}%
}} & \multicolumn{1}{r}{\small -0.0183} & \multicolumn{1}{r}%
{\small -0.0309} & \multicolumn{1}{r}{\small -0.0473} & \multicolumn{1}{r}%
{\small -0.0414} & \multicolumn{1}{r}{\small -0.0205} & \multicolumn{1}{r}{%
{\small 
\small \textit{-0.0169}%
}} & \multicolumn{1}{r}{\small -0.0244} & \multicolumn{1}{r}%
{\small -0.0509} & \multicolumn{1}{r}{\small -0.0948} \\ 
\multicolumn{1}{c}{\small 0.4} & \multicolumn{1}{c}{\small -0.25} & 
\multicolumn{1}{c}{\small 96} & \multicolumn{1}{r}{\small 0.0637} & 
\multicolumn{1}{r}{{\small 
\small \textit{0.0139}%
}} & \multicolumn{1}{r}{\small 0.0226} & \multicolumn{1}{r}%
{\small 0.0274} & \multicolumn{1}{r}{\small 0.0395} & \multicolumn{1}{r}%
{\small 0.0433} & \multicolumn{1}{r}{\small 0.0381} & \multicolumn{1}{r}{%
{\small 
\small \textbf{0.0154}%
}} & \multicolumn{1}{r}{\small 0.0651} & \multicolumn{1}{r}%
{\small 0.0312} & \multicolumn{1}{r}{\small 0.0661} \\ 
\multicolumn{1}{c}{} & \multicolumn{1}{c}{} & \multicolumn{1}{c}{\small 576}
& \multicolumn{1}{r}{\small 0.0175} & \multicolumn{1}{r}{{\small 
\small \textit{0.0066}%
}} & \multicolumn{1}{r}{\small 0.0091} & \multicolumn{1}{r}%
{\small 0.0095} & \multicolumn{1}{r}{\small 0.0122} & \multicolumn{1}{r}%
{\small 0.0146} & \multicolumn{1}{r}{\small 0.0070} & \multicolumn{1}{r}{%
{\small 
\small \textbf{0.0049}%
}} & \multicolumn{1}{r}{\small 0.0132} & \multicolumn{1}{r}%
{\small 0.0198} & \multicolumn{1}{r}{\small 0.0283} \\ 
\multicolumn{1}{c}{} & \multicolumn{1}{c}{\small 0.25} & \multicolumn{1}{c}%
{\small 96} & \multicolumn{1}{r}{\small 0.0504} & \multicolumn{1}{r}{{\small 
\small \textit{0.0141}%
}} & \multicolumn{1}{r}{\small 0.0364} & \multicolumn{1}{r}%
{\small 0.0490} & \multicolumn{1}{r}{\small 0.0526} & \multicolumn{1}{r}%
{\small 0.0548} & \multicolumn{1}{r}{\small 0.0301} & \multicolumn{1}{r}{%
{\small 
\small \textbf{0.0110}%
}} & \multicolumn{1}{r}{\small 0.0574} & \multicolumn{1}{r}%
{\small 0.0248} & \multicolumn{1}{r}{\small 0.0529} \\ 
\multicolumn{1}{c}{} & \multicolumn{1}{c}{} & \multicolumn{1}{c}{\small 576}
& \multicolumn{1}{r}{\small 0.0136} & \multicolumn{1}{r}{\small 0.0082} & 
\multicolumn{1}{r}{\small 0.0087} & \multicolumn{1}{r}{\small 0.0096} & 
\multicolumn{1}{r}{\small 0.0121} & \multicolumn{1}{r}{\small 0.0137} & 
\multicolumn{1}{r}{{\small 
\small \textit{0.0050}%
}} & \multicolumn{1}{r}{{\small 
\small \textbf{0.0031}%
}} & \multicolumn{1}{r}{\small 0.0108} & \multicolumn{1}{r}%
{\small 0.0153} & \multicolumn{1}{r}{\small 0.0238} \\ 
\multicolumn{1}{c}{} & \multicolumn{1}{c}{\small 0.45} & \multicolumn{1}{c}%
{\small 96} & \multicolumn{1}{r}{\small 0.0549} & \multicolumn{1}{r}{{\small 
\small \textbf{0.0189}%
}} & \multicolumn{1}{r}{\small 0.0321} & \multicolumn{1}{r}%
{\small 0.0458} & \multicolumn{1}{r}{\small 0.0529} & \multicolumn{1}{r}%
{\small 0.0585} & \multicolumn{1}{r}{\small 0.0362} & \multicolumn{1}{r}{%
{\small 
\small \textit{0.0204}%
}} & \multicolumn{1}{r}{\small 0.0570} & \multicolumn{1}{r}%
{\small 0.0274} & \multicolumn{1}{r}{\small 0.0668} \\ 
\multicolumn{1}{c}{} & \multicolumn{1}{c}{} & \multicolumn{1}{c}{\small 576}
& \multicolumn{1}{r}{\small 0.0192} & \multicolumn{1}{r}{{\small 
\small \textbf{0.0053}%
}} & \multicolumn{1}{r}{{\small 
\small \textit{0.0075}%
}} & \multicolumn{1}{r}{\small 0.0083} & \multicolumn{1}{r}%
{\small 0.0104} & \multicolumn{1}{r}{\small 0.0129} & \multicolumn{1}{r}%
{\small 0.0083} & \multicolumn{1}{r}{\small 0.0103} & \multicolumn{1}{r}%
{\small 0.0132} & \multicolumn{1}{r}{\small 0.0135} & \multicolumn{1}{r}%
{\small 0.0269} \\ 
\multicolumn{1}{c}{\small 0.9} & \multicolumn{1}{c}{\small -0.25} & 
\multicolumn{1}{c}{\small 96} & \multicolumn{1}{r}{\small 0.0359} & 
\multicolumn{1}{r}{{\small 
\small \textbf{0.0075}%
}} & \multicolumn{1}{r}{{\small 
\small \textit{0.0083}%
}} & \multicolumn{1}{r}{\small 0.0092} & \multicolumn{1}{r}%
{\small 0.0116} & \multicolumn{1}{r}{\small 0.0142} & \multicolumn{1}{r}%
{\small 0.0117} & \multicolumn{1}{r}{\small 0.0109} & \multicolumn{1}{r}%
{\small 0.0085} & \multicolumn{1}{r}{\small 0.0131} & \multicolumn{1}{r}%
{\small 0.0354} \\ 
\multicolumn{1}{c}{} & \multicolumn{1}{c}{} & \multicolumn{1}{c}{\small 576}
& \multicolumn{1}{r}{\small 0.0065} & \multicolumn{1}{r}{\small 0.0024} & 
\multicolumn{1}{r}{\small 0.0039} & \multicolumn{1}{r}{\small 0.0069} & 
\multicolumn{1}{r}{\small 0.0086} & \multicolumn{1}{r}{\small 0.0109} & 
\multicolumn{1}{r}{\small 0.0078} & \multicolumn{1}{r}{{\small 
\small \textit{0.0020}%
}} & \multicolumn{1}{r}{{\small 
\small \textbf{0.0014}%
}} & \multicolumn{1}{r}{\small 0.0045} & \multicolumn{1}{r}%
{\small 0.0151} \\ 
\multicolumn{1}{c}{} & \multicolumn{1}{c}{\small 0.25} & \multicolumn{1}{c}%
{\small 96} & \multicolumn{1}{r}{\small 0.0293} & \multicolumn{1}{r}{{\small 
\small \textit{0.0096}%
}} & \multicolumn{1}{r}{\small 0.0138} & \multicolumn{1}{r}%
{\small 0.0157} & \multicolumn{1}{r}{\small 0.0199} & \multicolumn{1}{r}%
{\small 0.0237} & \multicolumn{1}{r}{\small 0.0195} & \multicolumn{1}{r}%
{\small 0.0130} & \multicolumn{1}{r}{{\small 
\small \textbf{0.0073}%
}} & \multicolumn{1}{r}{\small 0.0099} & \multicolumn{1}{r}%
{\small 0.0322} \\ 
\multicolumn{1}{c}{} & \multicolumn{1}{c}{} & \multicolumn{1}{c}{\small 576}
& \multicolumn{1}{r}{\small 0.0083} & \multicolumn{1}{r}{{\small 
\small \textit{0.0022}%
}} & \multicolumn{1}{r}{\small 0.0036} & \multicolumn{1}{r}%
{\small 0.0043} & \multicolumn{1}{r}{\small 0.0055} & \multicolumn{1}{r}%
{\small 0.0063} & \multicolumn{1}{r}{\small 0.0029} & \multicolumn{1}{r}%
{\small 0.0057} & \multicolumn{1}{r}{{\small 
\small \textbf{0.0019}%
}} & \multicolumn{1}{r}{\small 0.0052} & \multicolumn{1}{r}%
{\small 0.0137} \\ 
\multicolumn{1}{c}{} & \multicolumn{1}{c}{\small 0.45} & \multicolumn{1}{c}%
{\small 96} & \multicolumn{1}{r}{\small 0.0235} & \multicolumn{1}{r}{{\small 
\small \textit{0.0109}%
}} & \multicolumn{1}{r}{\small 0.0128} & \multicolumn{1}{r}%
{\small 0.0146} & \multicolumn{1}{r}{\small 0.0165} & \multicolumn{1}{r}%
{\small 0.0198} & \multicolumn{1}{r}{\small 0.0149} & \multicolumn{1}{r}%
{\small 0.0132} & \multicolumn{1}{r}{{\small 
\small \textbf{0.0075}%
}} & \multicolumn{1}{r}{\small 0.0084} & \multicolumn{1}{r}%
{\small 0.0353} \\ 
\multicolumn{1}{c}{} & \multicolumn{1}{c}{} & \multicolumn{1}{c}{\small 576}
& \multicolumn{1}{r}{\small 0.0195} & \multicolumn{1}{r}{{\small 
\small \textbf{0.0030}%
}} & \multicolumn{1}{r}{\small 0.0075} & \multicolumn{1}{r}%
{\small 0.0081} & \multicolumn{1}{r}{\small 0.0076} & \multicolumn{1}{r}%
{\small 0.0116} & \multicolumn{1}{r}{\small 0.0084} & \multicolumn{1}{r}%
{\small 0.0071} & \multicolumn{1}{r}{{\small 
\small \textit{0.0042}%
}} & \multicolumn{1}{r}{\small 0.0044} & \multicolumn{1}{r}%
{\small 0.0120} \\ \hline\hline
\end{tabular}

\end{table}%

\newpage

\begin{table}[h]\
\vspace{-3cm}%
\caption{{\small RMSE estimates of the unadjusted LPR estimator, the feasible GS
estimator and the pre-filtered sieve bootstrap estimator, for the DGP:
ARFIMA(}${\small 0,d}_{{\small 0}}{\small ,1}${\small )  with Gaussian
innovations. The bias estimates of the feasible jackknife estimator based on
2,3,4,6,8 non-overlapping (NO) sub-samples, the feasible jackknife estimator
based on 2 moving block (MB) sub-samples, the maximum likelihood estimator (MLE)
and the pre-whitened (PW) estimator are obtained under the misspecified model:
ARFIMA(}${\small 0,d,0}${\small ) using the approach described in Section
5.3. The estimates are obtained by setting }${\small \alpha }${\small \ =
0.65. The lowest values are \textbf{bold-faced} and the second lowest values
are \textit{italicized}.}}%
\label{Table:MSE_ARFIMA(0,d,1)_mis}%
\renewcommand{\arraystretch}{0.9}%
\setlength{\tabcolsep}{4pt}%

\medskip

\begin{tabular}{lllccccccccccc}
\hline\hline
&  &  &  &  &  &  &  &  &  &  &  &  &  \\ 
${\small \theta }_{{\small 0}}$ & ${\small d}_{{\small 0}}$ & ${\small n}$ & 
$\widehat{{\small d}}_{{\small n}}$ & $\widehat{{\small d}}_{{\small J,2}}^{%
{\small NO}}$ & $\widehat{{\small d}}_{{\small J,3}}^{{\small NO}}$ & $%
\widehat{{\small d}}_{{\small J,4}}^{{\small NO}}$ & $\widehat{{\small d}}_{%
{\small J,6}}^{{\small NO}}$ & $\widehat{{\small d}}_{{\small J,8}}^{{\small %
NO}}$ & $\widehat{{\small d}}_{{\small J,2}}^{{\small MB}}$ & $\widehat{%
{\small d}}_{{\small 1}}^{{\small GS}}$ & $\widehat{{\small d}}^{{\small PFSB%
}}$ & $\widehat{{\small d}}^{{\small MLE}}$ & $\widehat{{\small d}}^{{\small %
PW}}$ \\ \hline
\multicolumn{14}{c}{} \\ 
\multicolumn{1}{c}{\small -0.9} & \multicolumn{1}{c}{\small -0.25} & 
\multicolumn{1}{c}{\small 96} & {\small 
\small \textbf{0.6233}%
} & {\small 0.6882} & {\small 0.6854} & {\small 0.6825} & 
{\small 0.6797} & {\small 0.6774} & {\small 0.6828} & {\small 
\small \textit{0.6385}%
} & {\small 0.8247} & {\small 0.6619} & {\small 0.7538} \\ 
\multicolumn{1}{c}{} & \multicolumn{1}{c}{} & \multicolumn{1}{c}{\small 576}
& {\small 
\small \textbf{0.4794}%
} & {\small 0.5335} & {\small 0.5391} & {\small 0.5447} & 
{\small 0.5482} & {\small 0.5505} & {\small 0.5549} & {\small 
\small \textit{0.4885}%
} & {\small 0.4977} & {\small 0.5043} & {\small 0.5578} \\ 
\multicolumn{1}{c}{} & \multicolumn{1}{c}{\small 0.25} & \multicolumn{1}{c}%
{\small 96} & {\small 
\small \textbf{0.7996}%
} & {\small 0.8695} & {\small 0.8662} & {\small 0.8533} & 
{\small 0.8504} & {\small 0.8442} & {\small 0.8451} & {\small 0.8268} & 
{\small 0.8430} & {\small 
\small \textit{0.8072}%
} & {\small 0.8459} \\ 
\multicolumn{1}{c}{} & \multicolumn{1}{c}{} & \multicolumn{1}{c}{\small 576}
& {\small 
\small \textbf{0.5951}%
} & {\small 0.6749} & {\small 0.6685} & {\small 0.6612} & 
{\small 0.6553} & {\small 0.6490} & {\small 0.6514} & {\small 0.6219} & 
{\small 0.6590} & {\small 
\small \textit{0.6049}%
} & {\small 0.6318} \\ 
\multicolumn{1}{c}{} & \multicolumn{1}{c}{\small 0.45} & \multicolumn{1}{c}%
{\small 96} & {\small 
\small \textbf{0.8219}%
} & {\small 0.8806} & {\small 0.8829} & {\small 0.8781} & 
{\small 0.8700} & {\small 0.8651} & {\small 0.8588} & {\small 0.8590} & 
{\small 
\small \textit{0.8327}%
} & {\small 0.8381} & {\small 0.8637} \\ 
\multicolumn{1}{c}{} & \multicolumn{1}{c}{} & \multicolumn{1}{c}{\small 576}
& {\small 
\small \textbf{0.5950}%
} & {\small 0.6474} & {\small 0.6433} & {\small 0.6419} & 
{\small 0.6349} & {\small 0.6355} & {\small 0.6671} & {\small 0.6298} & 
{\small 0.6487} & {\small 
\small \textit{0.6015}%
} & {\small 0.6372} \\ 
\multicolumn{1}{c}{\small -0.4} & \multicolumn{1}{c}{\small -0.25} & 
\multicolumn{1}{c}{\small 96} & {\small 
\small \textbf{0.2376}%
} & {\small 0.2996} & {\small 0.2963} & {\small 0.2927} & 
{\small 0.2903} & {\small 0.2847} & {\small 0.2834} & {\small 
\small \textit{0.2488}%
} & {\small 0.3103} & {\small 0.2688} & {\small 0.3142} \\ 
\multicolumn{1}{c}{} & \multicolumn{1}{c}{} & \multicolumn{1}{c}{\small 576}
& {\small 
\small \textbf{0.1037}%
} & {\small 0.1534} & {\small 0.1588} & {\small 0.1556} & 
{\small 0.1485} & {\small 0.1414} & {\small 0.1549} & {\small 
\small \textit{0.1098}%
} & {\small 0.1254} & {\small 0.1236} & {\small 0.1540} \\ 
\multicolumn{1}{c}{} & \multicolumn{1}{c}{\small 0.25} & \multicolumn{1}{c}%
{\small 96} & {\small 
\small \textbf{0.2527}%
} & {\small 0.2958} & {\small 0.2964} & {\small 0.2942} & 
{\small 0.2846} & {\small 0.2812} & {\small 0.2833} & {\small 
\small \textit{0.2560}%
} & {\small 0.2782} & {\small 0.2645} & {\small 0.3162} \\ 
\multicolumn{1}{c}{} & \multicolumn{1}{c}{} & \multicolumn{1}{c}{\small 576}
& {\small 
\small \textbf{0.1068}%
} & {\small 0.1576} & {\small 0.1529} & {\small 0.1486} & 
{\small 0.1455} & {\small 0.1438} & {\small 0.1482} & {\small 
\small \textit{0.1103}%
} & {\small 0.1199} & {\small 0.1182} & {\small 0.1573} \\ 
\multicolumn{1}{c}{} & \multicolumn{1}{c}{\small 0.45} & \multicolumn{1}{c}%
{\small 96} & {\small 
\small \textbf{0.2496}%
} & {\small 0.2999} & {\small 0.2927} & {\small 0.2900} & 
{\small 0.2867} & {\small 0.2815} & {\small 0.2795} & {\small 
\small \textit{0.2518}%
} & {\small 0.2725} & {\small 0.2691} & {\small 0.3029} \\ 
\multicolumn{1}{c}{} & \multicolumn{1}{c}{} & \multicolumn{1}{c}{\small 576}
& {\small 
\small \textbf{0.1047}%
} & {\small 0.1553} & {\small 0.1518} & {\small 0.1489} & 
{\small 0.1421} & {\small 0.1358} & {\small 0.1365} & {\small 
\small \textit{0.1098}%
} & {\small 0.1188} & {\small 0.1105} & {\small 0.1661} \\ 
\multicolumn{1}{c}{\small 0.4} & \multicolumn{1}{c}{\small -0.25} & 
\multicolumn{1}{c}{\small 96} & {\small 
\small \textbf{0.1982}%
} & {\small 0.2589} & {\small 0.2556} & {\small 0.2528} & 
{\small 0.2438} & {\small 0.2482} & {\small 0.2473} & {\small 0.2212} & 
{\small 0.2809} & {\small 
\small \textit{0.2140}%
} & {\small 0.2344} \\ 
\multicolumn{1}{c}{} & \multicolumn{1}{c}{} & \multicolumn{1}{c}{\small 576}
& {\small 
\small \textbf{0.0932}%
} & {\small 0.1442} & {\small 0.1418} & {\small 0.1397} & 
{\small 0.1362} & {\small 0.1315} & {\small 0.1428} & {\small 0.1078} & 
{\small 0.1268} & {\small 
\small \textit{0.1026}%
}\ & {\small 0.1473} \\ 
\multicolumn{1}{c}{} & \multicolumn{1}{c}{\small 0.25} & \multicolumn{1}{c}%
{\small 96} & {\small 
\small \textbf{0.1947}%
} & {\small 0.2546} & {\small 0.2439} & {\small 0.2418} & 
{\small 0.2317} & {\small 0.2338} & {\small 0.2496} & {\small 0.2213} & 
{\small 0.2663} & {\small 
\small \textit{0.2187}%
} & {\small 0.2365} \\ 
\multicolumn{1}{c}{} & \multicolumn{1}{c}{} & \multicolumn{1}{c}{\small 576}
& {\small 
\small \textbf{0.0925}%
} & {\small 0.1483} & {\small 0.1426} & {\small 0.1474} & 
{\small 0.1322} & {\small 0.1272} & {\small 0.1349} & {\small 0.1077} & 
{\small 0.1238} & {\small 
\small \textit{0.1043}%
} & {\small 0.1558} \\ 
\multicolumn{1}{c}{} & \multicolumn{1}{c}{\small 0.45} & \multicolumn{1}{c}%
{\small 96} & {\small 
\small \textbf{0.1964}%
} & {\small 0.2565} & {\small 0.2534} & {\small 0.2429} & 
{\small 0.2412} & {\small 0.2324} & {\small 0.2448} & {\small 0.2223} & 
{\small 0.2643} & {\small 
\small \textit{0.2125}%
} & {\small 0.2424} \\ 
\multicolumn{1}{c}{} & \multicolumn{1}{c}{} & \multicolumn{1}{c}{\small 576}
& {\small 
\small \textbf{0.0943}%
} & {\small 0.1338} & {\small 0.1288} & {\small 0.1265} & 
{\small 0.1169} & {\small 0.1142} & {\small 0.1396} & {\small 0.1090} & 
{\small 0.1229} & {\small 
\small \textit{0.0975}%
} & {\small 0.1466} \\ 
\multicolumn{1}{c}{\small 0.9} & \multicolumn{1}{c}{\small -0.25} & 
\multicolumn{1}{c}{\small 96} & {\small 
\small \textbf{0.0886}%
} & {\small 0.1384} & {\small 0.1357} & {\small 0.1329} & 
{\small 0.1310} & {\small 0.1294} & {\small 0.1367} & {\small 
\small \textit{0.1105}%
} & {\small 0.1253} & {\small 0.1254} & {\small 0.1427} \\ 
\multicolumn{1}{c}{} & \multicolumn{1}{c}{} & \multicolumn{1}{c}{\small 576}
& {\small 
\small \textbf{0.0344}%
} & {\small 0.0712} & {\small 0.0696} & {\small 0.0653} & 
{\small 0.0628} & {\small 0.0611} & {\small 0.0636} & {\small 
\small \textit{0.0561}%
} & {\small 0.0589} & {\small 0.0685} & {\small 0.0749} \\ 
\multicolumn{1}{c}{} & \multicolumn{1}{c}{\small 0.25} & \multicolumn{1}{c}%
{\small 96} & {\small 
\small \textbf{0.0865}%
} & {\small 0.1393} & {\small 0.1329} & {\small 0.1314} & 
{\small 0.1299} & {\small 0.1245} & {\small 0.1343} & {\small 0.1287} & 
{\small 0.1286} & {\small 
\small \textit{0.0952}%
} & {\small 0.1426} \\ 
\multicolumn{1}{c}{} & \multicolumn{1}{c}{} & \multicolumn{1}{c}{\small 576}
& {\small 
\small \textbf{0.0304}%
} & {\small 0.0794} & {\small 0.0768} & {\small 0.0752} & 
{\small 0.0746} & {\small 0.0699} & {\small 0.0712} & {\small 0.0692} & 
{\small 0.0706} & {\small 
\small \textit{0.0487}%
} & {\small 0.0758} \\ 
\multicolumn{1}{c}{} & \multicolumn{1}{c}{\small 0.45} & \multicolumn{1}{c}%
{\small 96} & {\small 
\small \textbf{0.0885}%
} & {\small 0.1375} & {\small 0.1348} & {\small 0.1311} & 
{\small 0.1297} & {\small 0.1249} & {\small 0.1324} & {\small 0.1181} & 
{\small 0.1197} & {\small 
\small \textit{0.0991}%
} & {\small 0.1424} \\ 
\multicolumn{1}{c}{} & \multicolumn{1}{c}{} & \multicolumn{1}{c}{\small 576}
& {\small 
\small \textbf{0.0378}%
} & {\small 0.0728} & {\small 0.0715} & {\small 0.0676} & 
{\small 0.0619} & {\small 0.0588} & {\small 0.0662} & {\small 0.0647} & 
{\small 0.0695} & {\small 
\small \textit{0.0454}%
} & {\small 0.0743} \\ \hline\hline
\end{tabular}

\end{table}%

\newpage

\begin{table}[h]\
\vspace{-1.07cm}%
\caption{{\small Bias estimates of the unadjusted LPR estimator, the feasible GS
estimator and the pre-filtered sieve bootstrap estimator, for the DGP:
ARFIMA(}${\small 1,d}_{{\small 0}}{\small ,1}${\small )  with Gaussian
innovations. The bias estimates of the feasible jackknife estimator based on
2,3,4,6,8 non-overlapping (NO) sub-samples, the feasible jackknife estimator
based on 2 moving block (MB) sub-samples, the maximum likelihood estimator (MLE)
and the pre-whitened (PW) estimator are obtained under the misspecified model:
ARFIMA(}${\small 2,d,0}${\small ) using the approach described in Section
5.3. The estimates are obtained by setting }${\small \alpha }${\small \ =
0.65. The lowest values are \textbf{bold-faced} and the second lowest values
are \textit{italicized}.}}%
\label{Table:Bias_ARFIMA(1,d,1)_mis}%
\renewcommand{\arraystretch}{0.80}%
\setlength{\tabcolsep}{3pt}%

\medskip

\begin{tabular}{llllccccccccccc}
\hline\hline
&  &  &  &  &  &  &  &  &  &  &  &  &  &  \\ 
${\small \phi }_{{\small 0}}$ & ${\small \theta }_{{\small 0}}$ & ${\small d}%
_{{\small 0}}$ & ${\small n}$ & $\widehat{{\small d}}_{{\small n}}$ & $%
\widehat{{\small d}}_{{\small J,2}}^{{\small NO}}$ & $\widehat{{\small d}}_{%
{\small J,3}}^{{\small NO}}$ & $\widehat{{\small d}}_{{\small J,4}}^{{\small %
NO}}$ & $\widehat{{\small d}}_{{\small J,6}}^{{\small NO}}$ & $\widehat{%
{\small d}}_{{\small J,8}}^{{\small NO}}$ & $\widehat{{\small d}}_{{\small %
J,2}}^{{\small MB}}$ & $\widehat{{\small d}}_{{\small 1}}^{{\small GS}}$ & $%
\widehat{{\small d}}^{{\small PFSB}}$ & $\widehat{{\small d}}^{{\small MLE}}$
& $\widehat{{\small d}}^{{\small PW}}$ \\ \hline
\multicolumn{15}{c}{} \\ 
\multicolumn{1}{c}{\small -0.9} & {\small -0.9} & \multicolumn{1}{c}{\small %
-0.25} & \multicolumn{1}{c}{\small 96} & \multicolumn{1}{r}{\small -0.5908}
& \multicolumn{1}{r}{\small -0.5874} & \multicolumn{1}{r}{\small -0.5986} & 
\multicolumn{1}{r}{\small -0.6020} & \multicolumn{1}{r}{\small -0.6135} & 
\multicolumn{1}{r}{\small -0.6159} & \multicolumn{1}{r}{\small -0.6954} & 
\multicolumn{1}{r}{{\small 
\small \textbf{-0.5065}%
}} & \multicolumn{1}{r}{{\small 
\small \textit{-0.5557}%
}} & \multicolumn{1}{r}{\small -0.6790} & \multicolumn{1}{r}%
{\small -0.6939} \\
\hline\hline
\end{tabular}

\end{table}%

\newpage

\begin{table}[h]\
\vspace{-1.07cm}%
\caption{{\small RMSE estimates of the unadjusted LPR estimator, the feasible GS
estimator and the pre-filtered sieve bootstrap estimator, for the DGP:
ARFIMA(}${\small 1,d}_{{\small 0}}{\small ,1}${\small )  with Gaussian
innovations. The bias estimates of the feasible jackknife estimator based on
2,3,4,6,8 non-overlapping (NO) sub-samples, the feasible jackknife estimator
based on 2 moving block (MB) sub-samples, the maximum likelihood estimator (MLE)
and the pre-whitened (PW) estimator are obtained under the misspecified model:
ARFIMA(}${\small 2,d,0}${\small ) using the approach described in Section
5.3. The estimates are obtained by setting }${\small \alpha }${\small \ =
0.65. The lowest values are \textbf{bold-faced} and the second lowest values
are \textit{italicized}.}}%
\label{Table:MSE_ARFIMA(1,d,1)_mis}%
\renewcommand{\arraystretch}{0.80}%
\setlength{\tabcolsep}{4pt}%

\medskip

\begin{tabular}{llllccccccccccc}
\hline\hline
&  &  &  &  &  &  &  &  &  &  &  &  &  &  \\ 
${\small \phi }_{{\small 0}}$ & ${\small \theta }_{{\small 0}}$ & ${\small d}%
_{{\small 0}}$ & ${\small n}$ & $\widehat{{\small d}}_{{\small n}}$ & $%
\widehat{{\small d}}_{{\small J,2}}^{{\small NO}}$ & $\widehat{{\small d}}_{%
{\small J,3}}^{{\small NO}}$ & $\widehat{{\small d}}_{{\small J,4}}^{{\small %
NO}}$ & $\widehat{{\small d}}_{{\small J,6}}^{{\small NO}}$ & $\widehat{%
{\small d}}_{{\small J,8}}^{{\small NO}}$ & $\widehat{{\small d}}_{{\small %
J,2}}^{{\small MB}}$ & $\widehat{{\small d}}_{{\small 1}}^{{\small GS}}$ & $%
\widehat{{\small d}}^{{\small PFSB}}$ & $\widehat{{\small d}}^{{\small MLE}}$
& $\widehat{{\small d}}^{{\small PW}}$ \\ \hline
\multicolumn{15}{c}{} \\ 
\multicolumn{1}{c}{\small -0.9} & {\small -0.9} & \multicolumn{1}{c}{\small %
-0.25} & \multicolumn{1}{c}{\small 96} & {\small 
\small \textbf{0.6822}%
} & {\small 0.7230} & {\small 0.7187} & {\small 0.7146} & 
{\small 0.7048} & {\small 0.6932} & {\small 0.7044} & {\small 
\small \textit{0.6912}%
} & {\small 0.6943} & {\small 0.7867} & {\small 0.9373} \\ 
\multicolumn{1}{c}{} &  & \multicolumn{1}{c}{} & \multicolumn{1}{c}{\small %
576} & {\small 
\small \textbf{0.4804}%
} & {\small 0.5062} & {\small 0.4935} & {\small 0.4895} & 
{\small 0.4893} & {\small 
\small \textit{0.4811}%
} & {\small 0.4928} & {\small 0.4830} & {\small 0.4916} & 
{\small 0.5538} & {\small 0.6328} \\
\hline\hline
\end{tabular}

\end{table}%

\bigskip

\begin{table}[h]\
\vspace{-11.5cm}%
\caption{{\small A ranking of the estimation methods.}}%
\label{rank}

\bigskip

\begin{center}
\begin{tabular}{lllllll}
\multicolumn{7}{c}{\textbf{Panel A: Correct Specification; known parameters
(Tables 2-5)}} \\ \hline
&  &  &  &  &  &  \\ 
\multicolumn{1}{c}{\underline{True DGP}} & \multicolumn{3}{c}{\underline{Bias%
}} & \multicolumn{3}{|c}{\underline{RMSE}} \\ 
\multicolumn{1}{c}{} & \multicolumn{1}{c}{First} & \multicolumn{1}{c}{Second}
& \multicolumn{1}{c}{Third} & \multicolumn{1}{|c}{First} & 
\multicolumn{1}{c}{Second} & \multicolumn{1}{c}{Third} \\ \cline{2-7}
& \multicolumn{1}{c}{} & \multicolumn{1}{c}{} & \multicolumn{1}{c}{} & 
\multicolumn{1}{|c}{} & \multicolumn{1}{c}{} & \multicolumn{1}{c}{} \\ 
\multicolumn{1}{c}{ARFIMA($1$,$d_{0},0$)} & \multicolumn{1}{c}{$\widehat{%
{\small d}}^{{\small MLE}}$} & \multicolumn{1}{c}{$\widehat{{\small d}}_{%
{\small J,2}}^{{\small Opt-NO}}$} & \multicolumn{1}{c}{$\widehat{{\small d}}%
_{{\small 1}}^{{\small GS}}$} & \multicolumn{1}{|c}{$\widehat{{\small d}}^{%
{\small MLE}}$} & \multicolumn{1}{c}{$\widehat{{\small d}}^{{\small PW}}$} & 
\multicolumn{1}{c}{$\widehat{{\small d}}_{{\small J,8}}^{{\small Opt-NO}}$}
\\ 
\multicolumn{1}{c}{ARFIMA($0$,$d_{0},1$)} & \multicolumn{1}{c}{$\widehat{%
{\small d}}^{{\small MLE}}$} & \multicolumn{1}{c}{$\widehat{{\small d}}_{%
{\small J,2}}^{{\small Opt-NO}}$} & \multicolumn{1}{c}{$\widehat{{\small d}}%
_{{\small 1}}^{{\small Opt-GS}}$} & \multicolumn{1}{|c}{$\widehat{{\small d}}%
^{{\small MLE}}$} & \multicolumn{1}{c}{$\widehat{{\small d}}_{{\small J,8}}^{%
{\small Opt-NO}}$} & \multicolumn{1}{c}{$\widehat{{\small d}}^{{\small PW}}$}
\\ 
\multicolumn{1}{c}{} & \multicolumn{1}{c}{} & \multicolumn{1}{c}{} & 
\multicolumn{1}{c}{} & \multicolumn{1}{c}{} & \multicolumn{1}{c}{} & 
\multicolumn{1}{c}{} \\
\hline\hline
\end{tabular}
\end{center}

\end{table}%

\end{subappendices}%

\end{document}